\documentclass{amsart}

\usepackage{amssymb}
\usepackage{amsthm}
\usepackage{amsfonts}
\usepackage{amsmath}
\usepackage{pmboxdraw}
\usepackage{verbatim} 
\usepackage{graphicx}
\usepackage{color}
\usepackage[colorlinks=true, citecolor=blue, filecolor=black, linkcolor=black, urlcolor=black]{hyperref}
\usepackage{cite}
\usepackage[normalem]{ulem}
\usepackage{subcaption}
\usepackage{bbm}
\usepackage{bm}
\usepackage{mathtools}
\usepackage{todonotes}
\usepackage{kantlipsum}
\usepackage{braket}
\usepackage{mathrsfs}
\usepackage{multirow,makecell}  
\usepackage{forest}

\usepackage{tcolorbox}
\tcbuselibrary{most}

\allowdisplaybreaks

\newcommand{\RN}[1]{%
	\textup{\uppercase\expandafter{\romannumeral#1}}%
}

\def\C{\mathbb{C}}

\def\E{\mathbb{E}}
\def\P{\mathbf{P}}
\def\R{\mathbb{R}}

\def\cO{\mathcal{O}}
\def\cZ{\mathcal{Z}}
\def\cL{\mathcal{L}}
\def\cI{\mathcal{I}}
\def\zbar{\overline{z}}

\newcommand{\Pf}{{\textup{Pf}}}
\newcommand{\erfc}{\operatorname{erfc}}
\newcommand{\erf}{\operatorname{erf}}
\newcommand{\bfR}{\mathbf{R}}

\newcommand{\bfkappa}{{\bm \varkappa}}

\newcommand{\re}{\operatorname{Re}}
\newcommand{\im}{\operatorname{Im}}

\theoremstyle{plain}
\newtheorem{thm}{Theorem}[section]

\newtheorem{lem}[thm]{Lemma}

\newtheorem{prop}[thm]{Proposition}

\theoremstyle{definition}

\newtheorem*{ex}{Example}

\theoremstyle{remark}

\numberwithin{equation}{section}

\newtcolorbox{example}[1]{breakable,
colbacktitle=gray!50!white, fonttitle=\bfseries, coltitle=black, title=Example: {#1}}

\newtcolorbox{definition}[1]{breakable,
colbacktitle=gray!50!white, fonttitle=\bfseries, coltitle=black, title=Definition: {#1}}

\begin{document}

\title[Pfaffian structure of the eigenvector overlap for the GinSE]{Pfaffian structure of the eigenvector overlap for\\ the symplectic Ginibre ensemble }

\author{Gernot Akemann}
\address{Faculty of Physics, Bielefeld University, P.O. Box 100131, 33501 Bielefeld, Germany;\newline
School of Mathematics, University of Bristol, Fry Building, Woodland Road,
Bristol BS8 1UG, UK}
\email{akemann@physik.uni-bielefeld.de}

\author{Sung-Soo Byun}
\address{Department of Mathematical Sciences and Research Institute of Mathematics, Seoul National University, Seoul 151-747, Republic of Korea}
\email{sungsoobyun@snu.ac.kr}

\author{Kohei Noda}
\address{Universit\'{e} catholique de Louvain, Institut de Recherche en Math\'{e}matique et Physique Chemin du Cyclotron 2, 1348 Louvain-La-Neuve, Belgium}
\email{kohei.noda@uclouvain.be}

\begin{abstract}
We study the integrable structure and scaling limits of the conditioned 
eigenvector overlap of the symplectic Ginibre ensemble of Gaussian non-Hermitian random matrices with independent quaternion elements. The average of the overlap matrix elements constructed from left and right eigenvectors, conditioned to $x$, are derived in terms of a Pfaffian determinant. Regarded as a two-dimensional Coulomb gas with the Neumann boundary condition along the real axis, it contains a kernel of skew-orthogonal polynomials with respect to the weight function $\omega^{\rm (over)}(z)=|z-\overline{x}|^2(1+|z-x|^2)e^{-2|z|^2}$, so including a non-trivial insertion of a point charge. The mean off-diagonal overlap is related to the diagonal (self-)overlap by a transposition, in analogy to the complex Ginibre ensemble. For $x$ conditioned to the real line, extending previous results at $x=0$, we determine the skew-orthogonal polynomials and their skew-kernel with respect to $\omega^{\rm (over)}(z)$. This is done in two steps and involves a Christoffel perturbation of the weight $\omega^{\rm (over)}(z)=|z-\overline{x}|^2\omega^{\rm (pre)}(z)$, by computing first the corresponding quantities for the unperturbed weight $\omega^{\rm (pre)}(z)$. Its kernel is shown to satisfy a differential equation at finite matrix size $N$. This allows us to take different large-$N$ limits, where we distinguish bulk and edge regime along the real axis. The limiting mean diagonal overlaps and corresponding eigenvalue correlation functions of the point processes with respect to $\omega^{\rm (over)}(z)$ are determined.
We also examine the effect on the planar orthogonal polynomials when changing the variance  in $\omega^{\rm (pre)}(z)$, as this appears in the eigenvector statistics of the complex Ginibre ensemble. 
\end{abstract}


\maketitle


\section{Introduction} 
In random matrix theory, the study of Hermitian and non-Hermitian random matrices exhibits notable differences from many viewpoints. A particular example is in the study of their eigenvector statistics. For instance, the Gaussian Unitary Ensemble (GUE), one of the most basic models among Hermitian random matrices, is invariant under conjugation by unitary matrices (see e.g. \cite{Forrester10}). This, in turn, implies that the eigenvectors of the GUE are distributed according to the Haar measure on the unitary group. On the other hand, such a simple consequence does not hold for its non-Hermitian counterpart, the complex Ginibre Ensemble (GinUE), where one needs to consider the left and right eigenvectors separately.
This naturally yields the notion of the eigenvector overlap of the GinUE, pioneered by Chalker and Mehlig \cite{CM98,CM00}. The study of eigenvector overlap enjoys intimate connections with free probability based on a diagrammatical approach \cite{WS15,N18}, applications to quantum chaotic scattering in physics \cite{FM02,FS03,FS12, Gros14,MS01},  the stochastic process of eigenvalues for non-Hermitian matrix-valued Brownian motion \cite{BY17,GW18,EKY21,Y20}, and the stability of the spectrum under perturbation through the condition number \cite{BNST17,FT21}.

After the seminal works of Chalker and Mehlig, the statistics of the GinUE eigenvectors has been actively studied, see e.g. \cite{BD21,ATTZ20,ATTZ20a,F18,BY17,CR22,JNNPZ99},\cite[Section 6.4]{BF24}, and references therein. 
There have been several different approaches to eigenvector overlap statistics, which include probability theory \cite{BZ18,BD21}, planar orthogonal polynomials \cite{ATTZ20}, and supersymmetry \cite{F18}. Several variants of the GinUE have been further investigated, such as the elliptic GinUE \cite{CFW24,CW24}, induced GinUE \cite{Noda23a}, truncated ensembles \cite{D21v1,D23}, product of GinUEs \cite{BSV17}, and spherical ensembles \cite{Noda23b,D21v1}. In addition to the symmetry class of the GinUE, the eigenvector statistics of the real Ginibre ensemble (GinOE) and its variants has been studied in the literature \cite{CS22,F18,WTF23,FT21,TAR24,CFW24,CW24}.
We also mention that beyond the Gaussian ensembles, the eigenvector statistics of non-invariant ensembles such as Wigner matrices has been analysed extensively from the viewpoint of the universal localisation/delocalisation phenomenon (see e.g. \cite{BY17,KY13,TV12}, and the review \cite{RVW16}). We also refer to \cite{CEHS23,EC23} and references therein for recent progress on i.i.d. matrices.

In this paper we will focus on the eigenvector statistics of the symplectic Ginibre ensemble (GinSE). So far, fewer works have been devoted to this symmetry class, see \cite{D21v2,AFK20}. We will generalise the approach \cite{ATTZ20} using planar orthogonal polynomials from to the GinUE to the GinSE. 
The reason to study the GinSE in detail is two-fold. First, it has interesting applications, to Hamiltonians with random potential in an imaginary magnetic field, relating to vortices in superconductors \cite{KE99}, and a conjectured map to fermionic field theory exists \cite{Has00}, cf.\cite{DLMS19}. Furthermore, it represents a particular two-dimensional Coulomb gas of interest in its own right \cite{Forrester10,Se24}. We refer to the review \cite{BF24} for further references. The second reason why the GinSE is interesting is that the local behaviour of its eigenvalue correlations along the real line gives rise to a universal behaviour that is different from that of the GinUE and GinOE, see \cite{Kan02,ABK22,BE23,BES23}. Below we will focus on this region for the eigenvector correlations.

In the study of planar symplectic ensembles, progress has been slower compared to their complex counterparts (see e.g. \cite{AHM11,HW21}) due to the lack of a general theory for analysing the kernel of planar skew-orthogonal polynomials (SOPs), which usually takes the form of a double summation. A way to overcome this difficulty was introduced in \cite{Kan02,Ak05}, where a certain differential equation for the large-$N$ limit of skew-kernels was obtained. The implementation of this idea for the finite-$N$ kernel was achieved in \cite{ABK22} for the GinSE and later extended to several variants such as the induced GinSE \cite{BC23}, the elliptic GinSE \cite{BE23,BES23}, the induced spherical ensemble \cite{BF23a}, the truncated ensemble \cite{BF24}, and the non-Hermitian Wishart ensemble \cite{BN24}. We will follow this strategy here which will be crucial to take scaling limits of large matrix size $N$ later.

In order to formulate our results on the eigenvector statistics in the GinSE and to put them into context, we recall definitions and known results on eigenvalues and eigenvectors in the next Subsection \ref{prelim} as part of this introduction. Our main results are presented in Section \ref{Main} and are as follows. In Subsection \ref{Intstructure}, based on \cite{AFK20}, we present that at finite-$N$ the mean diagonal and off-diagonal overlap can be expressed in terms of the partition function and kernel of planar SOPs with respect to a new non-Gaussian weight function $\omega^{\rm (over)}(z)=|z-\overline{x}|^2(1+|z-x|^2)e^{-2|z|^2}$. It depends on the eigenvalue $x\in\mathbb{C}$ of the corresponding eigenvector overlap matrix element. We give Pfaffian expressions for a generalised mean diagonal overlap conditioned to more than one eigenvalue. In analogy to \cite{ATTZ20}, we give a relation between the mean off-diagonal overlap and the mean diagonal overlap via a transposition lemma. 

A substantial part of this paper is then devoted to the construction of the planar SOPs for the weight $\omega^{\rm (over)}(z)$, with the main findings summarised in Subsection \ref{Subsection_Main SOP}. This weight function can be interpreted as a particular insertion of point charges into the Gaussian weight. For that reason we also present the eigenvalue correlation functions for this weight in their own right. 
Because the construction of planar SOPs is difficult for this weight, we split the task in two steps by factorising $\omega^{\rm (over)}(z)=|z-\overline{x}|^2\omega^{\rm (pre)}(z)$. First, the SOPs, norms and kernel with respect to $\omega^{\rm (pre)}(z)$ are determined, for which we require $x \in \R$. This latter is the weight function appearing in the eigenvector overlap in the GinUE \cite{ATTZ20}. 
In passing we investigate properties of this weight for planar orthogonal polynomials for different variances in Appendix \ref{Appendix_Planar OP}.
The full answer for $\omega^{\rm (over)}(z)$ follows from a so-called Christoffel perturbation, cf. \cite{AEP22}, when multiplying $\omega^{\rm (pre)}(z)$ by the polynomial $|z-\overline{x}|^2$. (We also refer to \cite{SS23,ST03,AKS23,LY23,Ch22} and the references therein for related subjects in the context of conditional point processes or root type singularities.) 
This step again limits our result to the conditional eigenvalue to be real in all the following, $x\in\mathbb{R}$. 
A further crucial step is the derivation of a differential equation for the unperturbed kernel. This puts us in the position to obtain the large-$N$ limits in the bulk and at the edge of the real line, presented in Subsection \ref{Scalinglim}, including comparisons to numerics.

For pedagogical reasons, Section~\ref{S_Conditional Origin}, following the presentation of our main results in Section \ref{Main}, is a warm-up case where we present the special rotationally invariant case $\omega^{\rm (over)}(z)|_{x=0}=|z|^2(1+|z|^2)e^{-2|z|^2}$. Here, all the computations become significantly simpler and the SOPs, their kernel, its differential equation and the large-$N$ limit are obtained easily. 
The rest of the paper is dedicated to the proofs of our results, see the end of Section \ref{Main} for a link to the corresponding theorems.
 

\subsection{Preliminaries: eigenvalue and eigenvector correlations in the GinSE}\label{prelim}

In this section, we will summarise known results for both eigenvalue and eigenvector statistics of the GinSE, insofar as they will be needed to formulate our main results in the next section.
We begin with the definition of the ensemble and recall what is known about it. 
By using the $2 \times 2$ matrix complex-valued representation of the quaternions, the GinSE is defined by 
\begin{equation}
\label{GinSEM}
    G_{2N}:=
    \begin{pmatrix}
        A_N & B_N 
        \smallskip 
        \\
        -\overline{B}_N & \overline{A}_N
    \end{pmatrix},
\end{equation}
where $A_N$ and $B_N$ are independent $N \times N$ matrices whose elements are i.i.d. complex normal Gaussian random variables with mean 0 and variance $1/2$ for their real parts and for their imaginary parts. The distribution of matrix elements can thus be written as 
\begin{equation}
P(G_{2N})=C_{2N} \exp\left[ -\mbox{Tr}(G_{2N}G^*_{2N})\right],
\label{PG}
\end{equation}
where $C_{2N}$ is an appropriate normalisation constant, and $^*$ denotes the adjoint. It is invariant under unitary transformations, and thus the matrix $G_{2N}$ can be brought into upper triangular block diagonal form by a unitary symplectic matrix $U_{2N}\in \mbox{USp}(2N)/\mbox{U}(1)^N$,
\begin{equation}
G_{2N}= U_{2N}(D_{2N}+T_{2N})U^{-1}_{2N}.
\label{Schur}
\end{equation}
Here, $D_{2N}$ denotes the $2\times2$ block diagonal matrix
\begin{equation}
\label{DIAM}
D_{2N}:=\mathrm{diag}
    \begin{pmatrix}
        z_j & 0 \\
        0 & \zbar_j
    \end{pmatrix}_{j=1}^{N}
    =\begin{pmatrix}
        z_1 & 0 \\
        0 & \zbar_1
    \end{pmatrix}
    \oplus
    \cdots
    \oplus
    \begin{pmatrix}
        z_N & 0 \\
        0 & \zbar_N
    \end{pmatrix}.
\end{equation}  
It contains the $2N$ eigenvalues $z_j$ that come in complex conjugated pairs\footnote{Due to the non-commutativity of the quaternions, the eigenvalues $z_j$ are infinitely degenerate, after a transformation $qz_jq^{-1}$ with $q\neq0$. In the following we thus consider equivalence classes of complex eigenvalues, see e.g. \cite{AFK20} for details.}.
The matrix $T_{2N}$ is strictly upper block diagonal, with $N(N-1)/2$ non-zero quaternion elements in $2\times2$ complex matrix representation. 
Because of this, we obtain independent Gaussian distributions for the elements of $D_{2N}$ and $T_{2N}$, 
\begin{equation}
P(G_{2N})=C_{2N} \exp\left[ -\mbox{Tr}(D_{2N}D_{2N}^*+T_{2N}T_{2N}^*)\right].
\label{PDT}
\end{equation}
More general distributions of matrix elements have been considered \cite{BF24}, e.g. the induced or elliptic GinSE, or normal random matrices (with $T_{2N}=0$).

We first recall some basic properties of the statistics of complex eigenvalues which is well-understood, see \cite{BF24} for a comprehensive review. It is based on the formalism of planar SOPs for general weight functions $\omega : \mathbb{C} \to \mathbb{R}_{\ge 0}$, following \cite{Kan02}. 
First, the Jacobian for the change of variables \eqref{PDT} is independent of the matrix $T_{2N}$ \cite{Mehta}.
Furthermore, in the case when the distribution of matrix elements $D_{2N}$ and $T_{2N}$ decouple as in \eqref{PDT} (or when the latter are absent), the dependence on $T_{2N}$ can be integrated out and the resulting joint probability distribution function of  the set of complex eigenvalues $\{z_1,z_2,\ldots,z_N\}\in\C^{N}$ (and their complex conjugates) is given by 
\begin{equation}
\label{Gibbs}
d\P_N(z_1,\ldots,z_N):=
\frac{1}{Z_N} \prod_{1\leq j<k\leq N}|z_j-z_k|^2|z_j-\overline{z}_k|^2\prod_{j=1}^{N}|z_j-\overline{z}_j|^2 \omega(z_j)\,dA(z_j),
\end{equation}
where $dA(z)=d^2z/\pi$ is the area measure in the plane, and $Z_N$ is the normalisation constant
\begin{equation}
Z_N:=\int_{\C^N}  \prod_{1\leq j<k\leq N}|z_j-z_k|^2|z_j-\overline{z}_k|^2 \prod_{j=1}^{N}|z_j-\overline{z}_j|^2 \omega(z_j)\,dA(z_j).
\label{ZNdef}
\end{equation}
We refer to \cite{Fo16} and \cite[Appendix A]{BF23a} for the interpretation of ensemble \eqref{Gibbs} from the perspective of Coulomb gas theory.

It turns out that for the eigenvector correlations we will have to consider more general weights than the Gaussian weight \eqref{Weight GinSE} resulting from \eqref{PDT}. 
The statistical behaviour of the complex eigenvalues in planar symplectic ensembles is encoded in the $k$-point correlation functions  defined as
\begin{equation}
\label{Rkdef}
    \bfR_{N,k}(
    z_1,\ldots,z_k) := \frac{N!}{(N-k!)}\int_{\C^{N-k}}
    \P_N(z_1,\dots,z_N)\prod_{j=k+1}^N  dA(z_j).
\end{equation}
A notable feature is that they form a Pfaffian point process, namely, for any $k=1,2,\dots,N$, we have 
\begin{equation} \label{def of Pfaff Struc}
\bfR_{N,k}(z_1,\ldots,z_k)
= \Pf \bigg[
\begin{pmatrix}
\bfkappa_N(z_j,z_{\ell}) & \bfkappa_N(z_j,\zbar_{\ell}) 
\smallskip 
\\
\bfkappa_N(\zbar_j,z_{\ell}) & \bfkappa_N(\zbar_j,\zbar_{\ell})
\end{pmatrix}
\bigg]_{j,\ell=1}^k 
\prod_{j=1}^{k}(\zbar_j-z_j) \, \omega(z_j). 
\end{equation}
Here, the $2\times 2$ matrix valued kernel contains the skew-kernel $\bfkappa_N$ (also sometimes called pre-kernel) as a building block. It can be expressed in terms of SOPs $q_k$ and their norms $r_k$ defined below as  
\begin{equation}
\bfkappa_{N}(z,w) = \sum_{k=0}^{N-1}
\frac{ q_{2k+1}(z)q_{2k}(w)-q_{2k}(z)q_{2k+1}(w) }{r_{k}}.
\end{equation}
Here, the skew-inner product with respect to some measure $d\mu(z)=w(z) \, dA(z)$ on $\mathbb{C}$ is defined as 
\begin{equation}
\label{inner_product}
\langle f, g\rangle_s : =\int_{\mathbb{C}} \Big( f(z)\overline{g(z)}-g(z)\overline{f(z)} \Big)(z-\overline{z}) \, d\mu(z).
\end{equation}
A family of polynomials $(q_k)_{k\in\mathbb{Z}}$ is called planar SOP associated with weight $w(z)$ if  
    \begin{equation}
        \langle q_{2k},q_{2\ell} \rangle_s=\langle q_{2k+1},q_{2\ell+1} \rangle_s=0,\qquad 
        \langle q_{2k},q_{2\ell+1} \rangle_s=-\langle q_{2\ell+1},q_{2k} \rangle_s=r_k \, \delta_{k,\ell},
        \label{SQPdef}
    \end{equation}
    where $r_k$ is their skew-norm which is positive. 
It also follows from de Bruijn’s integration formula that the partition function $Z_N$ can be expressed as
\begin{equation}
Z_N = N! \prod_{k=0}^{N-1} r_k,
\label{ZNprod}
\end{equation}
see e.g. \cite{Kan02}, \cite[Remark 2.5]{AEP22}.

In the case of the GinSE, the weight function is given by the Gaussian distribution
\begin{equation} \label{Weight GinSE}
\omega^{ \rm (g) } (z):= e^{-2|z|^2}, 
\end{equation}
with scalar product \eqref{inner_product} with respect to $d\mu(z)=\omega^{ \rm (g) } (z)\,dA(z)$,
and \eqref{Gibbs} gives the distribution of eigenvalues of the GinSE. 
In the sequel, we add the superscript $\rm (g)$ for the associated objects of the Gaussian weight \eqref{Weight GinSE}. 
In this case, we have \cite{Kan02}
\begin{equation}
\label{GinSESOP}
q_{2k+1}^{(\mathrm{g})}(z)=z^{2k+1},\qquad
q_{2k}^{(\mathrm{g})}(z)=\sum_{\ell=0}^{k}\frac{k!}{\ell!}z^{2\ell},\qquad
r_{k}^{(\mathrm{g})}=\frac{(2k+1)!}{2^{2k+1}}.
\end{equation}
As a consequence, it follows that 
\begin{equation}
\label{Partition_GinSE}
Z_N^{({\rm g})}=N!\prod_{k=0}^{N-1}\frac{(2k+1)!}{2^{2k+1}} .
\end{equation}
Furthermore, we have 
\begin{align}
\begin{split}
\label{GinSEKernel}
\bfkappa_{N}^{(\mathrm{g})}(z,w) = \sqrt{2} \bigg[
\sum_{k=0}^{N-1}\sum_{\ell=0}^{k} 
\Bigl( 
\frac{(\sqrt{2}z)^{2k+1}}{(2k+1)!!}\frac{(\sqrt{2}w)^{2\ell}}{(2\ell)!!}
-\frac{(\sqrt{2}w)^{2k+1}}{(2k+1)!!}\frac{(\sqrt{2}z)^{2\ell}}{(2\ell)!!}
\Bigr)
\bigg]. 
\end{split}
\end{align}  
We also mention that as $N \to \infty$, the normalised eigenvalues $z_j \mapsto z_j/\sqrt{N}$ of the GinSE tend to be uniformly distributed on the unit disc, known as the circular law \cite{BC12}. 

\medskip 

We now summarise what is known about  the GinSE eigenvector statistics, following \cite{AFK20,D21v2}.  
The eigenvalues $\{z_j\}_{j=1}^N$ induce the left and right eigenvectors  
\begin{equation}
\label{REV}
G_{2N}R_j=z_jR_j, \qquad 
L_j^{\mathsf{t}}G_{2N}=z_jL_j^{\mathsf{t}}. 
\end{equation}
We write $R_{\overline{j}}:=\cI R_{j}$ and $L_{\overline{j}}:=\cI L_j$,  where the map $\cI:\C^{2N}\to\C^{2N}$ is given by 
\begin{equation*}
\label{cImap}
\cI\left(
\begin{bmatrix}
\boldsymbol{u} \\
\boldsymbol{v}
\end{bmatrix}
\right)
=
\begin{bmatrix}
-\overline{\boldsymbol{v}} \\
\overline{\boldsymbol{u}}
\end{bmatrix},
\quad
\boldsymbol{u},\boldsymbol{v}\in\C^{N}. 
\end{equation*} 
Then, it follows that
\begin{equation}
\label{LEV}
G_{2N} R_{\overline{j}} =\overline{z}_j R_{\overline{j}} ,\qquad 
L_{\overline{j}}^{\mathsf{t}}G_{2N}=\overline{z}_j L_{\overline{j}}^{\mathsf{t}}, 
\end{equation}
see e.g. \cite[Lemma 2.3]{LOR12}. 
These left and right eigenvectors form a bi-orthogonal system \cite[Eq.(2.7)]{AFK20} 
\begin{equation}
\label{BOR}
 L_j \cdot R_k = \delta_{j,k},\qquad 
 L_j, \cdot R_{ \overline k } = 0,\qquad 
 L_{ \overline j } \cdot R_k = 0,\qquad 
 L_{ \overline j } \cdot R_{ \overline k }  = \delta_{j,k}, 
\end{equation}
where $ \boldsymbol{x} \cdot \boldsymbol{y} =\boldsymbol{x}^{\mathsf{t}}\overline{\boldsymbol{y}}$ is the Euclidean inner product on $\C^{2N}$. 
Then, the matrix of non-orthogonality overlaps, also called Chalker-Mehlig-correlators, is defined by
\begin{equation}
\label{cOP}
\cO
=
\begin{pmatrix}
\cO_{j,k} & \cO_{j,\overline{k}} 
\smallskip 
\\
\cO_{\overline{j},k} & \cO_{\overline{j},\overline{k}}
\end{pmatrix}_{j,k=1}^{N}
:= \begin{pmatrix}
 L_j \cdot L_{k} \,  R_j \cdot R_{k}  &  L_j \cdot L_{\overline{k}} \, R_j \cdot R_{\overline{k}}
\smallskip  
\\
 L_{\overline{j}} \cdot L_{k} \, R_{\overline{j}} \cdot R_{k}   &  L_{\overline{j}} \cdot L_{\overline{k}} \, R_{\overline{j}} \cdot R_{\overline{k}}  
\end{pmatrix}_{j,k=1}^{N}.
\end{equation}   
Its diagonal elements $\cO_{j,j}$ and  $\cO_{\overline{j},\overline{k}}$ are called diagonal overlaps (or self-overlaps) and the remaining elements off-diagonal overlaps.
We collect some basic properties of these, cf. \cite{AFK20}:
\begin{itemize}
    \item For each $j$,  $\cO_{j,\overline{j}}=\cO_{\overline{j},j}=0$, and for $j\not= k$, $\overline{\cO}_{\overline{j},k}=\cO_{j,\overline{k}}$; 
    \smallskip 
    \item  $\sum_{k=1}^{N}\cO_{k,k}=1$ and $\sum_{k=1}^{N}\cO_{\overline{k},\overline{k}}=1$;
    \smallskip 
    \item The overlap is scale-invariant under the scale-transform $R_j\mapsto c_j\,R_j$, $L_j\mapsto c_j^{-1}\,L_j$ for $c_j \not =0.$
\end{itemize}
The overlap matrix elements in \eqref{cOP} depend on both complex eigenvalues $D_{2N}$ and the matrix $T_{2N}$ in the Schur decomposition \eqref{Schur}. Therefore, their (conditional) expectation value with respect to the Gaussian measure \eqref{PG}, which we denote by $\E_N$, is a nontrivial computation. In fact, the integral over the matrix elements $T_{2N}$ in such an average can be performed \cite{D21v2,AFK20}, and the remaining integral is over the complex eigenvalues in $D_{2N}$ of the form \eqref{Gibbs}.
Let us define the mean diagonal and mean off-diagonal overlaps conditioned to certain $x_1,x_2\in\C$, with $x_1\neq x_2$, as follows: from \cite[Eqs. (2.23a), (2.23b)]{AFK20}\footnote{Compared to \cite{AFK20}, we use a slightly different convention for the normalisation, without dividing by \(N\) or \(N^2\).}
\begin{align}
D_{1,1}^{(N,1)}(x_1) &:= \E_N\bigg[
\sum_{l=1}^{N} 
	\delta(x_1-z_l)\mathcal{O}_{l,l} 
	\bigg]=
	N\E_N\left[\delta(x_1-z_1)\mathcal{O}_{1,1} \right],
	\label{O1point}
 \\
 D_{1,2}^{(N,2)}\left(x_1,x_2\right) &:=  
	\E_N\bigg[  
	\sum_{ \substack{ k,l=1;\ k\neq l}}^{N}  
	\delta(x_1-z_k)\delta(x_2-z_l)\mathcal{O}_{k,l} \bigg]
= 
 N(N-1)\E_N\left[  \delta(x_1-z_1)\delta(x_2-z_2)\mathcal{O}_{1,2} \right], 
	\label{O2point}
 \\ \widetilde{D}_{1,\overline{2}}^{(N,2)}\left(x_1,x_2\right) &:=  
	\E_N\bigg[ 
	\sum_{ \substack{ k,l=1;\ k\neq l}}^{N}  	
\delta(x_1-{z_k})\delta(x_2-\overline{z}_l)\mathcal{O}_{{k},\overline{l}}  
\bigg]
=
N(N-1)
\E_N\left[  \delta(x_1-z_1)\delta(x_2-\overline{z}_2)\mathcal{O}_{1,\overline{2}} \right].  \label{tO2point}
\end{align}    
In the respective second equalities we have used that the corresponding integrals over complex eigenvalues given in \eqref{O11ev} and \eqref{O12ev} below,  are invariant under permutation of the indices of eigenvalues and overlaps (and under complex conjugation). Thus the average over say all diagonal overlap matrix elements can be expressed in terms of those of a single matrix element $\mathcal{O}_{l,l} $, where without loss of generality we have chosen indices 1 and 2 of the matrix \eqref{cOP} (as well as the corresponding overlined indices). Furthermore, it turns out as a result of the computations in \cite{AFK20} that the averages of further overlap matrix elements, replacing $\mathcal{O}_{l,l}$ by $\mathcal{O}_{\overline{l},\overline{l}}$ in \eqref{O1point}, 
$\mathcal{O}_{l,k}$ by $\mathcal{O}_{\overline{l},\overline{k}}$ in \eqref{O2point} and 
$\mathcal{O}_{l,\overline{k}}$ by $\mathcal{O}_{\overline{l},{k}}$ in \eqref{tO2point}, respectively, lead to the same quantities. Thus the definitions \eqref{O1point}, \eqref{O2point} and \eqref{tO2point} above contain all cases of averaged overlap matrix elements. 
The expression for the mean diagonal overlap \eqref{O1point} in terms of the remaining integrals over the complex eigenvalues 
reads  \cite[Theorem 3.4]{D21v2} or \cite[Eq. (3.20)]{AFK20},
\begin{align}
 D_{1,1}^{(N,1)}(x) &=
 \frac{N}{Z_N^{(\rm g)}}
 |x-\overline{x}|^2e^{-2|x|^2}
 \int_{\C^{N-1}}
 \prod_{2\leq k<l\leq N}
 |z_l-z_k|^2|z_l-\overline{z}_k|^2 \prod_{l=2}^N|z_l-x|^2|z_l-\overline{x}|^2 \nonumber 
 \\
 & \quad \times
 \prod_{l=2}^{N}\Bigl( 
1+\frac{1}{2|z_l-x|^2}+\frac{1}{2|z_l-\overline{x}|^2}
 \Bigr)\prod_{l=2}^{N}|z_l-\overline{z}_l|^2e^{-2|z_l|^2}\, dA(z_l) \nonumber 
 \\
 & = \frac{N}{Z_{N}^{\rm (g)}} |x-\overline{x}|^2 e^{-2|x|^2} \int_{\mathbb{C}^{N-1}} 
	\prod_{ 2\leq k < l \leq N } 
	|z_l-z_k|^2|z_l-\overline{z}_k|^2\prod_{l=2}^{N}|z_l-\overline{z}_l|^2\ \omega^{ \rm (over) }(z_l) \,  dA(z_l),
	\label{O11ev}
\end{align} 
defining the following weight function,  after using symmetries of the integral:
\begin{equation}
\omega^{ \rm (over)}(z):= \omega^{ \rm (over)}(z,\overline{z}|x,\overline{x}):=
(z-\overline{x})(\overline{z}-x)\big(1+(z-x)(\overline{z}-\overline{x})\big)e^{-2z\overline{z}}.
\label{OverlapWeight}
\end{equation}
In addition to \(D_{1,1}^{(N,1)}\), the mean off-diagonal overlaps \(D_{1,2}^{(N,2)}\) and \(\widetilde{D}_{1,\overline{2}}^{(N,2)}\), given in \eqref{O2point} and \eqref{tO2point}, respectively, can also be defined in terms of an integral over the Gibbs measure, see \eqref{O12ev} below.
This new weight function depends on $z,\overline{z}, x, \overline{x}$, which we will suppress in our notation in some of the following. Due to the symmetry of the integral it is also possible to choose a weight $\omega^{ \rm (over)}(z)=
|z-{x}|^2\left(1+|z-\overline{x}|^2\right)e^{-2|z|^2}$. However, if we want to relate the mean off-diagonal overlap to this expression below we will need to choose \eqref{OverlapWeight}. 
As a consequence, the mean diagonal overlap is, up to prefactors,  given by the partition function $Z_{N-1}^{\rm (over)}(x)$ of $N-1$ points with respect to the new weight function $\omega^{ \rm (over)}(z)$ in \eqref{OverlapWeight}. We will thus have to determine the SOPs, norms and skew-kernel with respect to $\omega^{ \rm (over)}(z)$, in order to determine the mean diagonal overlap explicitly at finite-$N$. 
Thus we shall add the superscripts (over) for the associated SOPs, skew kernel and norms.

This computation will be done in two steps, for the following reason. It turns out that for the construction of these objects it is the third factor in \eqref{OverlapWeight} that poses most difficulties. We will therefore introduce an auxiliary {pre-weight} function 
\begin{equation}
\label{PreOverlapWeight}
\omega^{({\rm pre})}(z) \equiv \omega^{({\rm pre})}(z,\overline{z}|x,\overline{x}):=(1+|z-x|^2)e^{-2|z|^2} ,\qquad x\in\C.
\end{equation}
This is the weight function (with exponent $-|z|^2$ though) for planar orthogonal polynomials in the eigenvectors of the GinUE \cite{ATTZ20}, cf. Appendix \ref{Appendix_Planar OP}.
After having determined the objects labelled with (pre) such as $q_k^{ \rm (pre) }$, $r_k^{(\rm pre)}$ and $ \bfkappa_N^{ \rm (pre) }$, the corresponding quantities with respect to the weight \eqref{OverlapWeight} follow by applying the analogue of the so-called Christoffel perturbation in the complex plane, cf. \cite{AB07,AEP22}. This is because of the multiplicative relation
\begin{equation}
\label{Over-pre}
\omega^{(\rm over)}(z) = |z-\overline{x}|^2 \omega^{({\rm pre})}(z), 
\end{equation}
see Proposition~\ref{Prop_Christoffel pertubation} for details below. 
The limitation of our approach using \cite[Theorem 5.1]{AEP22} (that is based on \cite{AB07}) is that it only holds for a Christoffel perturbation by $x\in\mathbb{R}$, a real quantity\footnote{The proof in \cite{AB07} is based on the idea of extending the Vandermonde determinant of size $2N$ in \eqref{Gibbs} of eigenvalues and complex conjugates to $2N+1$, when including the Christoffel perturbation by $|z-m|^2$ of the initial weight. To obtain a dependence on $m$ and $\overline{m}$ as independent variables, we would have to perturb the weight by two factors  $|z-m|^2|z-\overline{m}|^2$ which is not the case here.}.

A similar result to \eqref{O11ev} was obtained in \cite{AFK20} for the mean off-diagonal overlap, see \eqref{O12ev} below. 
In the next subsection containing our main results we will present Lemma~\ref{Lem_D12 D11} below that allows to obtain the mean off-diagonal overlap   $D_{1,2}^{(N,2)}(x_1,x_2)$ as a function of $x_1,\overline{x}_1,x_2,\overline{x}_2$ from the mean diagonal overlap $D_{1,1}^{(N,2)}(x_1,x_2)$, conditioned on a second eigenvalue $x_2$, see \eqref{O11evk}  for $k=2$ below, by exchanging 
$\overline{x}_1\leftrightarrow\overline{x}_2$. This is in complete analogy to the result in \cite{ATTZ20} for the eigenvector overlaps in the complex Ginibre ensemble. Furthermore, when conditioning the mean diagonal and off-diagonal overlap on more eigenvalues, we obtain a Pfaffian structure for finite-$N$, generalising the integrable determinantal structure in \cite{ATTZ20}. 
Because in the present paper we have been unable to determine the mean diagonal overlap 
as a function of the independent variables $x_1$ and $\overline{x}_1$, being restricted to real $x_1=a\in\mathbb{R}$, we cannot exploit the relation between mean diagonal and off-diagonal overlap further in this work.

\medskip

We finish this overview over existing results for eigenvectors statistics in the GinSE with known results for the large-$N$ limit. In \cite{D21v2} several results are derived, including the angle between eigenvectors.  
We mention in particular \cite[Theorem 3.5]{D21v2}, where Dubach showed that the distribution of the diagonal overlap $\cO_{1,1}/N$ conditioned to have the eigenvalue $z_1=0$ at the origin, converges to $(\gamma_4/2)^{-1}$ in law as $N\to\infty$, where $\gamma_{\alpha}$ is the gamma distribution with parameter $\alpha$. As a corollary of this result, the expectation of the diagonal overlap conditioned at the origin is given by $2/3$.

The method in \cite{D21v2} applied by Dubach is based on the Schur decomposition \eqref{Schur} and beta-gamma algebra, which seems to work only for the case conditioned at the origin. 
The extension of the result for the mean diagonal overlap to real points, including the bulk and edge, will be presented in Theorem~\ref{Thm_CEO} below. Indeed, the origin case is special and can be treated in an easier way, as can be seen from setting $x=0$ in the weight function \eqref{OverlapWeight}. Then, the resulting weight becomes rotationally invariant, the odd SOPs become monomials as in the Gaussian case, and the computation simplifies substantially. In Section~\ref{S_Conditional Origin}, we provide an alternative and short derivation of Dubach's result for the conditional expectation of the diagonal overlap using the SOPs.

Finally, it was heuristically observed in \cite[Eq.(4.12)]{AFK20} that the leading order of the mean diagonal overlap is given by 
\begin{equation}\label{mean diagonal global density}
 \frac{1}{N^2}D_{1,1}^{(N,1)}(z) \approx \frac{1}{\pi}(1-|z|^2), \qquad \textup{for } |z|<1, \textup{ and } z \not \in \R. 
\end{equation}
For the mean off-diagonal overlap the corresponding heuristic result 
reads
\begin{equation}
  \frac{1}{N^2}D_{1,2}^{(N,2)}(x_1,x_2) \approx 
  \frac{1-x_1\overline{x}_2}{\pi^2|x_1-x_2|^4}
\end{equation}
for $x_1\neq x_2$, with $x_1,x_2\notin\mathbb{R}$, and both points inside the unit disc. Both heuristic results are on a macroscopic scale and agree with the corresponding results for the complex Ginibre ensemble \cite{CM98,CM00}, thus being conjectured to be universal. Because we will have to restrict ourselves to the real line, $x_1=a\in\mathbb{R}$, where complex and symplectic Ginibre ensemble differ,  we expect to be in a different universality class here.

\section{Main results}\label{Main}

\subsection{Integrable structure of overlaps}\label{Intstructure}

We begin by presenting the integrable Pfaffian structure of the mean diagonal overlap in the GinSE. Here, we generalise the findings of a corresponding determinantal structure of the overlaps in the GinUE in \cite{ATTZ20}. Let us consider a generalised mean diagonal overlap, by integrating over $N-k$ eigenvalues in  \eqref{O11ev} only:
\begin{align}
\begin{split}
D_{1,1}^{(N,k)}(z_1,\dots,z_k)&:= \frac{N!}{(N-k)!}
 \frac{1}{Z_{N}^{\rm (g)}} |z_1-\overline{z}_1|^2e^{-2|z_1|^2}
\\
&\quad \times  \int_{\mathbb{C}^{N-k}} 
	\prod_{ 2\leq  j<l \leq N } 
	|z_j-z_l|^2|z_j-\overline{z}_l|^2 \prod_{l=2}^{N}|z_l-\overline{z}_l|^2\omega^{ \rm (over) }(z_l,\overline{z}_l|z_1,\overline{z}_1)\prod_{j=k+1}^N	dA(z_j)
	\label{O11evk}
\end{split}
\end{align}
for $k=1,\dots,N$. As mentioned already, for $k=1$ we have that $D_{1,1}^{(N,1)}(z_1)$ is proportional to the partition function 
$Z_{N-1}^{\rm (over)}(z_1)$ of $N-1$ points, as all complex eigenvalues $z_2,\ldots,z_N$ are integrated out. For $k\geq2$ it is a $(k-1)$-point correlation function with respect to the weight $\omega^{ \rm (over) }(z)$, and thus \eqref{def of Pfaff Struc} immediately applies, adjusting the range of $k$, and with $N\to N-1$ therein. This leads to the following lemma.
\begin{lem}[\textbf{Pfaffian structure of the mean diagonal overlap}]
For $z_1,\dots,z_k\in\C$, the generalised mean diagonal overlap \eqref{O11evk}
is given by 
\begin{equation}
 D_{1,1}^{(N,1)}(z_1) = 
 \frac{NZ_{N-1}^{(\rm over)}(z_1)}{Z_N^{(\rm g)}}
 |z_1-\overline{z}_1|^2e^{-2|z_1|^2}
\end{equation}
for $k=1$ and 
\begin{align}
    \begin{split}
D_{1,1}^{(N,k)}(z_1,\dots,z_k)&= 
 \frac{NZ_{N-1}^{(\rm over)}(z_1)}{Z_{N}^{\rm (g)}} |z_1-\overline{z}_1|^2e^{-2|z_1|^2} \prod_{j=2}^{k}(\zbar_j-z_j)\ \omega^{(\rm over)}(z_j,\overline{z}_j|z_1,\overline{z}_1)
\\
&\qquad\times \Pf \bigg[
\begin{pmatrix}
\bfkappa_{N-1}^{(\rm over)}(z_j,z_{\ell}|z_1,\overline{z}_1) & \bfkappa_{N-1}^{(\rm over)}(z_j,\zbar_{\ell}|z_1,\overline{z}_1) 
\smallskip 
\\
\bfkappa_{N-1}^{(\rm over)}(\zbar_j,z_{\ell}|z_1,\overline{z}_1) & \bfkappa_{N-1}^{(\rm over)}(\zbar_j,\zbar_{\ell}|z_1,\overline{z}_1)
\end{pmatrix}
\bigg]_{j,\ell=2}^k ,
\label{O11kPf}
    \end{split}
\end{align}
for $k\geq2$. 
Here, the notation $\bfkappa_{N-1}^{(\rm over)}(z_j,z_{\ell}|z_1,\overline{z}_1)$ indicates that the skew-kernel also depends on the conditioning point $z_1$ (and its complex conjugate).
\end{lem}

We turn to the mean off-diagonal overlap defined in \eqref{O2point}. It can be expressed as an integral over the complex eigenvalues following \cite{AFK20}\footnote{As for the mean diagonal overlap in \eqref{O11ev} there is a certain choice made here in \eqref{O12ev} based on symmetry, for the bracket in the second line of the equation, see details in \cite{AFK20}.}
\newpage
\begin{align}
D_{1,2}^{(N,2)}(z_1,z_2)
&  = -\frac{N(N-1)}{Z_{N}^{\rm (g)}} |z_1-\overline{z}_1|^2 |z_2-\overline{z}_2|^2|z_1-\overline{z}_2|^2  
     e^{ -2|z_1|^2-2|z_2|^2}  \int_{\mathbb{C}^{N-2}} 
	\prod_{ 3\leq k < l \leq N } 
	|z_l-z_k|^2|z_l-\overline{z}_k|^2 
	\nonumber
 \\
&\quad \times  \prod_{j=3}^{N}  |z_1-\overline{z}_j|^2 |z_2-\overline{z}_j|^2\left(|z_1-{z}_j|^2 |z_2-{z}_j|^2+(\overline{z}_1-\overline{z}_j) (z_2-{z}_j)\right)
|z_j-\overline{z}_j|^2  e^{-2|z_j|^2} \, dA(z_j).
	\label{O12ev}
\end{align} 
The second quantity $\widetilde{D}_{1,\overline{2}}^{(N,2)}\left(z_1,z_2\right)$ defined in \eqref{tO2point} is obtained by exchanging $z_2\leftrightarrow \overline{z}_2$ in the above formula \eqref{O12ev}, see \cite{AFK20}. 
In analogy to \eqref{O11evk} we could define the mean off-diagonal overlap conditioned on more than two eigenvalues, yet with another weight function as in the second line in \eqref{O12ev}. This quantity then also has a Pfaffian structure as in \eqref{O11kPf}, with a different kernel. While lacking a more explicit expression for this kernel we shall not pursue this direction further.

We are now in the position to write down a simple relation between the mean diagonal overlap \eqref{O11evk} conditioned to two eigenvalues at $k=2$ and the mean off-diagonal overlap \eqref{O12ev}.
Let $\widehat{T}$ be the transposition operator acting on functions $g$ on $\mathbb{C}^{2k}$, with $k\geq2$. Here, we will treat a complex variable $x$ and its complex conjugate $\overline{x}$ as independent. The function $g$ thus depends on (at least) 
the set of four variables $z_1,\overline{z}_1,z_2,\overline{z}_2$. The action of  $\widehat{T}$ is defined by exchanging the variables 
$\overline{z}_1\leftrightarrow\overline{z}_2$:
\begin{equation}
\widehat{T}g(z_1,\overline{z}_1,z_2,\overline{z}_2,\ldots)= g(z_1,\overline{z}_2,z_2,\overline{z}_1,\ldots).
\label{Tdef}
\end{equation}
In particular $\widehat{T}$ leaves the remaining variables $z_3,\overline{z}_3,\ldots,z_k,\overline{z}_k$ (if present)
unchanged. This allows us to write the following relation, generalising \cite[Lemma 1]{ATTZ20} for the GinUE.

\begin{lem}[\textbf{Relation between mean off-diagonal overlap and mean diagonal overlap}] \label{Lem_D12 D11}
The mean off-diagonal overlap follows from the mean diagonal overlap in \eqref{O11kPf} at $k=2$:
\begin{align}
\begin{split} 
D_{1,2}^{(N,2)}(z_1,z_2)
&=
\frac{(z_1-\overline{z}_1)(z_2-\overline{z}_2)e^{ -2|z_1-z_2|^2} }{(1-|z_1-z_2|^2)|z_1-\overline{z}_2|^2}\widehat{T}
D_{1,1}^{(N,2)}(z_1,z_2)
\\
&=
-
     \frac{NZ_{N-1}^{(\rm over)}(z_1,\overline{z}_2)}{Z_N^{(\rm g)}}
|z_1-\overline{z}_1|^2|z_2-\overline{z}_2|^2(\overline{z}_2-z_1) 
     e^{ -2|z_1|^2-2|z_2|^2}
\bfkappa_{N-1}^{(\rm over)}(z_2,\overline{z}_1|z_1,\overline{z}_2). \label{O11-12rel} 
\end{split}
\end{align} 
\end{lem}


\subsection{Planar skew-orthogonal polynomials} \label{Subsection_Main SOP}

Integrable properties of planar symplectic ensembles can be effectively analysed using the SOP formalism \cite{Kan02}, as extensively discussed in Subsection~\ref{prelim}.
In practice, constructing explicit forms of SOPs is crucial in further asymptotic analysis. 
This is particularly feasible when the underlying measure is radially symmetric, see Section \ref{S_Conditional Origin} and 
\eqref{SOP radially symmetric} below. 
Beyond that case, 
in \cite[Theorem 3.1]{AEP22}, a method to construct SOPs was presented under the assumption that the associated planar orthogonal polynomials satisfy the classical three-term recurrence relation  
\begin{equation} \label{stand 3-term recur}
z \,p_k(z)= p_{k+1}(z) + b_k\,p_k(z)+c_k \,p_{k-1}(z). 
\end{equation} 
However, contrary to orthogonal polynomials on the real line, it is well known that on $\mathbb{C}$ this is not always the case.
Nonetheless, the planar orthogonal polynomials $p_k^{ \rm (pre) }$ associated with the weight $\omega_{({\rm pre})}(z)$ satisfy a non-standard three-term recurrence relation when $p_{k-1}(z)$ is multiplied in addition with argument $z$, see 
\eqref{non-stand 3 term c}. 
Such a recurrence was also introduced in \cite[Remark 1.3]{BLY21} in the context of a point charge insertion of the GinUE. 
This seemingly small difference $p_{k-1}(z)\to z\,p_{k-1}(z)$ significantly impacts the construction of SOPs. Consequently, 
neither the result \cite[Theorem 3.1]{AEP22} nor its underlying idea can be applied to construct SOPs associated with \eqref{OverlapWeight}.
We overcome this difficulty by using a new method to construct the SOPs. 
We mention that a family of orthogonal polynomials satisfying a general recurrence formula, which includes both \eqref{stand 3-term recur} and \eqref{non-stand 3 term c}, is introduced in \cite{IM95} as type \(R_I\), see also \cite{KS23} and references therein.

For $k=0,1,\dots,$ denote the truncated exponential by
\begin{equation}
\label{exponential_sum}
e_k(x):=\sum_{j=0}^{k}\frac{x^j}{j!}
\end{equation} 
and define
\begin{equation}
\label{fsum}
f_k(x):=\sum_{j=0}^k(k+1-j)\frac{x^j}{j!} =(k+1)e_k(x)-xe_{k-1}(x), 
\end{equation}   
which also appears in the GinUE \cite{ATTZ20}. 
Then we have the following. 

\begin{thm}[\textbf{Construction of SOPs associated with pre-overlap weight \eqref{PreOverlapWeight}}]
\label{Thm_OverlapCoe} 
Suppose that $a \in \R$ and define
\begin{equation}
\label{SOPsPreOverlap}
q_{2k}^{({\rm pre})}(z):=\sum_{j=0}^{2k}\alpha_{2k,j}^{({\rm pre})}z^j, \qquad 
q_{2k+1}^{({\rm pre})}(z):=\sum_{j=0}^{2k+1}\beta_{2k+1,j}^{({\rm pre})}z^j,
\end{equation}
where $\alpha_{2k,j}^{ \rm (pre) }$ and $\beta_{2k+1,j}^{ \rm (pre) }$ are given by 
\begin{align}
\label{Coe Alpha Even}
\alpha_{2k,2j}^{(\rm pre)}
&= \bigg[ \sum_{\ell=j}^{k}(\ell+1-j)\frac{(2k+3)!!}{(2\ell+3)!!}\frac{(2a^2)^{\ell}}{(2a^2)^{j}} - \sum_{\ell=j}^{k}(\ell-j)\frac{(2k+1)!!}{(2\ell+1)!!}\frac{(2a^2)^{\ell}}{(2a^2)^{j}} \bigg] \frac{2^jf_j(a^2)}{2^kf_k(a^2)},
\\
\label{Coe Alpha Odd}
\alpha_{2k,2j+1}^{(\rm pre)}
&= 2a \bigg[ \sum_{\ell=j}^{k-1}(\ell+1-j)\frac{(2k+3)!!}{(2\ell+5)!!}\frac{(2a^2)^{\ell}}{(2a^2)^{j}}
- \sum_{\ell=j}^{k-1}(\ell-j)\frac{(2k+1)!!}{(2\ell+3)!!}\frac{(2a^2)^{\ell}}{(2a^2)^{j}} \bigg] \frac{2^{j}f_j(a^2)}{2^kf_k(a^2)},
\\
\label{Coe Beta Even}
\beta_{2k+1,2j}^{(\rm pre)}
&= a \bigg[
\frac{2j+1}{2k+3}
\frac{(2a^2)^k}{(2a^2)^j}+2\sum_{\ell=j}^{k}(\ell+1-j)\frac{(2k+1)!!}{(2\ell+3)!!}\frac{(2a^2)^{\ell}}{(2a^2)^j} \bigg]
\frac{2^jf_j(a^2)}{2^kf_k(a^2)},
\\
\label{Coe Beta Odd}
\beta_{2k+1,2j+1}^{(\rm pre)} &= \bigg[ \frac{2j+3}{2k+3} \frac{(2a^2)^k}{(2a^2)^j}+2\sum_{\ell=j}^{k}(\ell-j)\frac{(2k+1)!!}{(2\ell+3)!!}\frac{(2a^2)^{\ell}}{(2a^2)^j} \bigg] \frac{2^jf_j(a^2)}{2^kf_k(a^2)}. 
\end{align} 
Then, the family $ \{ q_{k}^{ \rm (pre) } \}_{k=0}^\infty$ forms SOP with respect to the weight \eqref{PreOverlapWeight} and skew-product \eqref{inner_product}, with skew-norms 
\begin{equation}
    \label{Skew norm PreWeight}
    r_{k}^{(\rm pre)}\equiv r_k^{(\rm pre)}(a):=\frac{(2k+2)!}{2^{2k+2}}\frac{f_{k+1}(a^2)}{f_k(a^2)}.
\end{equation}
\end{thm}

By employing a method using the moment matrix, we establish a general statement for tri-diagonal moment matrices regarding the recurrence relation of planar SOPs (Theorem~\ref{thm_Recurrence}), providing an alternative method for their construction beyond that presented in \cite{AEP22}. Consequently, Theorem~\ref{Thm_OverlapCoe} follows as a specific application to the weight function given in \eqref{PreOverlapWeight}.
Our construction of SOPs is of independent interest and may find further applications. For our main purpose of applying it to the eigenvector overlap, we defer the details to Section~\ref{S_Construction of skew-orthogonal polynomials}.

By using \cite[Theorem 5.1]{AEP22}, known as the idea of Christoffel perturbation, we can express the SOPs $q_{k}^{({\rm over})}$ as well as the skew-kernel $\bfkappa_{N-1}^{({\rm over})}$ in terms of their (pre)-counterparts.

\begin{prop}[\textbf{SOPs and skew-kernel associated with overlap weight \eqref{OverlapWeight} via Christoffel perturbation \cite[Theorem 5.1]{AEP22}}]  \label{Prop_Christoffel pertubation}
Suppose that $a \in \R$. 
\begin{itemize}
    \item \textup{\textbf{(SOPs)}} The following  family $ \{ q_{k}^{ \rm (over) } \}_{k=0}^\infty$ forms SOP associated with \eqref{OverlapWeight}: 
\begin{equation}
q_{2k}^{(\mathrm{over})}(z) :=\frac{r_k^{({\rm pre})}\bfkappa_{k+1}^{({\rm pre})}(a,z)}{(a-z)q_{2k}^{({\rm pre})}(a)}, 
\qquad 
q_{2k+1}^{(\mathrm{over})}(z):=\frac{q_{2k+2}^{({\rm pre})}(a)q_{2k}^{({\rm pre})}(z)-q_{2k}^{({\rm pre})}(a)q_{2k+2}^{({\rm pre})}(z)}{(a-z)q_{2k}^{({\rm pre})}(a)}, 
\end{equation}
 with skew-norm 
\begin{align} \label{def of skew norm}
r_k^{({\rm over})} \equiv r_k^{({\rm over})}(a):=r_k^{({\rm pre})}\frac{q_{2k+2}^{({\rm pre})}(a)}{q_{2k}^{({\rm pre})}(a)}.
\end{align}
\item  \textup{\textbf{(Skew-kernel)}} Their corresponding skew-kernel reads
\begin{equation} \label{def of skew kernel overlap}
\bfkappa_N^{({\rm over})}(z,w):=\sum_{k=0}^{N-1}
\frac{q_{2k+1}^{({\rm over})}(z)q_{2k}^{({\rm over})}(w)-q_{2k+1}^{({\rm over})}(w)q_{2k}^{({\rm over})}(z)}{r_k^{({\rm over})}}, 
\end{equation}
and  we have 
\begin{equation}
\label{SkewKernelOverlap}
\bfkappa_{N-1}^{({\rm over})}(z,w)
= \frac{\bfkappa_N^{({\rm pre})}(z,w)q_{2N}^{({\rm pre})}(a)
-\bfkappa_N^{({\rm pre})}(z,a)q_{2N}^{({\rm pre})}(w)
+\bfkappa_N^{({\rm pre})}(w,a)q_{2N}^{({\rm pre})}(z)}
{(z-a)(w-a)q_{2N}^{({\rm pre})}(a)}.
\end{equation}
\end{itemize}
\end{prop}

We now discuss the finite-$N$ expression of the  mean diagonal overlap. We write 
\begin{equation}
\label{widehat GinSE}
\widehat{\bfkappa}_{N}^{(\mathrm{g})}(z,a)=e^{-2za}\bfkappa_{N}^{(\mathrm{g})}(z,a),
\end{equation}
where $\bfkappa_{N}^{(\mathrm{g})}$ is given by \eqref{GinSEKernel}. 
Then, as an immediate consequence, we have the following. 

\begin{prop}[\textbf{Conditional mean diagonal overlap at finite-$N$}] \label{Prop_condi exp finite N}
For 
$a\in\C$, we define the conditional mean diagonal overlap at finite-$N$ as
\begin{equation}
\label{hat D11N1 a}
\widehat{D}_{1,1}^{(N,1)}(a)  
:=
\frac{D_{1,1}^{(N,1)}(a)}{\mathbf{R}_{N,1}^{(\rm g)}(a)}. 
\end{equation}
Then, we have  
\begin{equation}   \label{Conditional Expectation Diagonal Overlap Real}
\widehat{D}_{1,1}^{(N,1)}(a)  = \frac{NZ_{N-1}^{(\rm over)}(a)}{Z_{N}^{(\rm g)}} \begin{cases}
\displaystyle \frac{a-\overline{a}}{\bfkappa_N^{(\rm g)}(a,\overline{a})}
 &\textup{if } a \in \C \setminus \R ,
\smallskip 
\\
\displaystyle \frac{e^{-2a^2}}{  \partial_x\widehat{\bfkappa}_{N}^{(\rm g)}(x,a)\bigr|_{x=a}  } &\textup{if } a \in  \R .
\end{cases} 
\end{equation}
\end{prop}

In this article, we follow  \cite{CM98,CM00,AFK20,ATTZ20} in defining the mean conditional overlap, rather than presenting a probabilistic definition based on $\mathcal{O}_{1,1}$ as a random variable that is not used in the following, see e.g  \cite{BD21,D21v2}. It is not difficult to see that the quantity \eqref{hat D11N1 a} normalised by the spectral density is equivalent to the conditional expectation of the diagonal overlap for GinSE similar to the earlier work \cite[Eq. (4.5)]{BD21}, see also \cite{CFW24}. We also refer to the discussion in \cite[Subsection 4.2]{AFK20}.  


Recall that the regularised incomplete gamma function is defined by 
\begin{equation}
Q(a,z) = \frac{1}{ \Gamma(a) }   \int_z^\infty t^{a-1} e^{-t}\,dt, 
\end{equation}
see e.g. \cite[Chapter 8]{NIST}.
We write 
\begin{equation}
\label{ExpLk2}
\widehat{\cL}_{k}(z,a) = (z-a)^2\partial_z \Big[ \frac{\widehat{\bfkappa}_{k}^{(\rm{g})}(z,a)}{z-a} \Big],
\end{equation}
where $\widehat{\bfkappa}_{k}^{(\rm{g})}$ is given by \eqref{widehat GinSE}. 
Our next result provides the most crucial ingredient for the asymptotic analysis.

\begin{thm}[\textbf{Differential equation for the skew-kernel with pre-overlap weight \eqref{PreOverlapWeight}}]
\label{Thm_ODE_PreKernel}
Suppose that $a \in \R$. Define the second order differential operator 
\begin{equation}\label{OP_DG}
\mathfrak{D}_{z,a} := (z-a)\partial_z^2-(2(z-a)^2+2)\partial_z-2(z-a).
\end{equation}
Let  
\begin{equation}
\label{WbfK}
\widetilde{\bfkappa}_N^{(\mathrm{pre})}(z,w) := e^{2a^2-2za-2wa}(z-a)^3(w-a)^3\bfkappa_N^{({\rm pre})}(z,w).
\end{equation}
Then, we have 
\begin{equation}
\label{ODE_WbfK}
\mathfrak{D}_{z,a}\widetilde{\bfkappa}_N^{(\mathrm{pre})}(z,w) = 
\mathrm{I}_N(z,w) - \mathrm{II}_N(z,w) - \mathrm{III}_N(z,w) + \mathrm{IV}_N(z,w).
\end{equation}
Here, the inhomogeneous terms are given by 
\begin{align}
\mathrm{I}_N(z,w) &:= 4(z-a)^3(w-a)^3e^{2(z-a)(w-a)}Q(2N,2zw), \label{def of IN}
\\
\mathrm{II}_N(z,w) &:= 4a(z-a)^3(w-a)^2e^{2a^2-2za-2wa}
\bigg(  \frac{(2zw)^{2N}}{(2N)!} + \bigl(2N+1-2az\bigr)\frac{2^{2N-1}e_N(a^2)(zw)^{2N}}{(2N)!f_N(a^2)} \bigg), \label{def of IIN}
\\
\mathrm{III}_N(z,w) &:= (z-a)^3e^{2a^2-2za} \frac{2N+1}{2} \frac{(2N+3)\widehat{\cL}_{N+1}(w,a)-2a^2\widehat{\cL}_{N}(w,a)}{N!f_N(a^2)}z^{2N},  \label{def of IIIN}
\\
\mathrm{IV}_N(z,w) &:= (z-a)^3e^{2a^2-2za}
za\frac{(2N+1)\widehat{\cL}_{N+1}(w,a)-2a^2\widehat{\cL}_{N}(w,a)}{N!f_N(a^2)}z^{2N}, \label{def of IVN}
\end{align} 
where $e_k$ and $f_k$ are given by \eqref{exponential_sum} and \eqref{fsum}, respectively.  
\end{thm} 

The key feature of finding a differential equation for the skew-kernel is that the inhomogeneous terms make the asymptotic analysis tractable, as they no longer involve double sums. An interesting aspect of this differential equation is that the inhomogeneous terms involve the kernel of the determinantal point process of the GinUE.  
It is highly non-trivial a priori that such a differential operator exists. In particular, the SOP we find in Theorem \ref{Thm_OverlapCoe} are not expressed in terms of classical orthogonal polynomials, such as Hermite or Laguerre polynomials, in contrast to previous works \cite{Kan02,ABK22,BC23,BE23,BES23,BF23a,BF24}. Furthermore, our result is the first example of a differential equation in the case where the associated orthogonal polynomial does not satisfy a standard three term recurrence relation \eqref{stand 3-term recur}.

\subsection{Bulk and edge scaling limits for the diagonal overlap and eigenvalue correlation functions}  
\label{Scalinglim}

We now discuss the scaling limits.
In our next result, we extend Dubach's conditional expectation \cite{D21v2} of the limiting diagonal overlap at $a=0$ to the real axis, distingiushing it bulk and the vicinity of the real edge.

\begin{thm}[\textbf{Scaling limits of the conditional expectation of the diagonal overlap}] \label{Thm_CEO}
For $p \in [-1,1]$, let 
\begin{equation} \label{scaling for bulk/edge}
\alpha= \begin{cases}
\displaystyle 1 &\textup{if } p \in (-1,1)\quad \mbox{bulk case}, 
\smallskip
\\
\displaystyle 1/2  &\textup{if } p = 1 \quad \quad\quad\ \mbox{edge case}.
\end{cases}
\end{equation}
Then as $N \to \infty$, we have 
\begin{equation}
\label{D11lim}
\frac{1}{N^\alpha}\,
\widehat{D}_{1,1}^{(N,1)}(\sqrt{N}p+\chi) \to 
\begin{cases}
\displaystyle  \rho_{\rm b}(p)  &\textup{if } p \in (-1,1) \quad \mbox{bulk case}, 
\smallskip
\\
\displaystyle \rho_{\rm e}(\chi)  &\textup{if } p = 1 \quad \quad\quad\ \mbox{edge case},
\end{cases}
\end{equation}
where the convergence is uniform for $\chi$ in compact subsets of $\R$. 
Here, for the bulk case with subscript \textup{b},
\begin{equation} \label{bulk rho b}
\rho_{\rm b}(p) := \dfrac{2}{3}(1-p^2),  
\end{equation}
and for the edge case with subscript \textup{e}, 
\begin{equation}
\rho_{\rm e}(\chi):= \frac{ \sqrt{2} }{3\sqrt{\pi}}
\frac{e^{-4\chi^2}+\sqrt{2\pi}\chi e^{-2\chi^2}\erfc(\sqrt{2}\chi)-4\sqrt{\pi}\chi\erfc(2\chi)}{ \sqrt{2}\, \erfc(2\chi)-e^{-2\chi^2} \erfc(\sqrt{2}\chi)}. 
\end{equation} 
\end{thm}

\begin{figure}[t!]
   \begin{minipage}[b]{0.48\linewidth}
    \centering
    \includegraphics[width=0.9\textwidth]{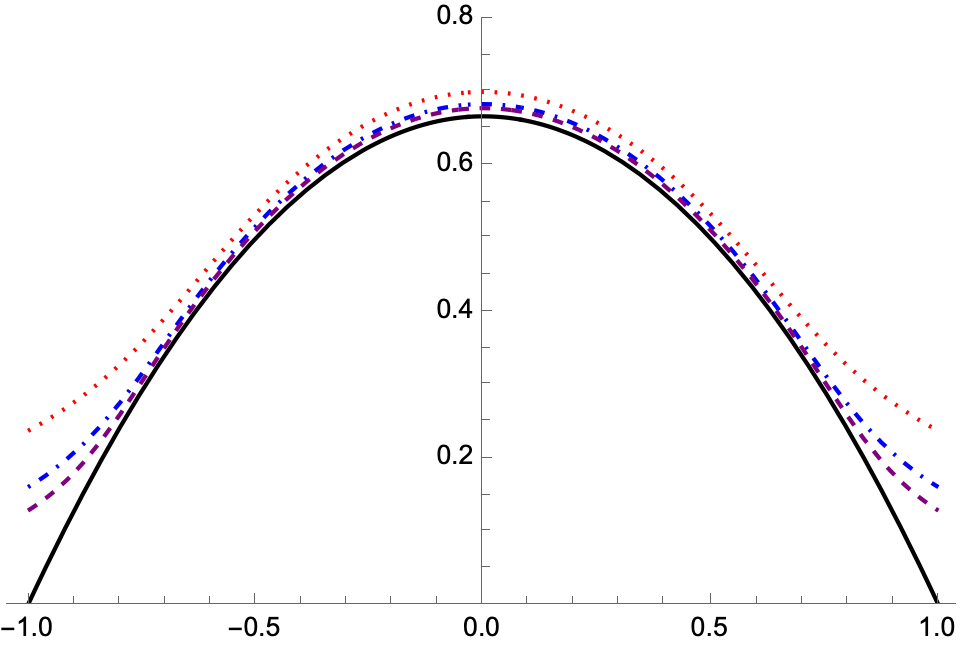}
    \subcaption{Bulk}
  \end{minipage}
   \begin{minipage}[b]{0.48\linewidth}
    \centering
    \includegraphics[width=0.9\textwidth]{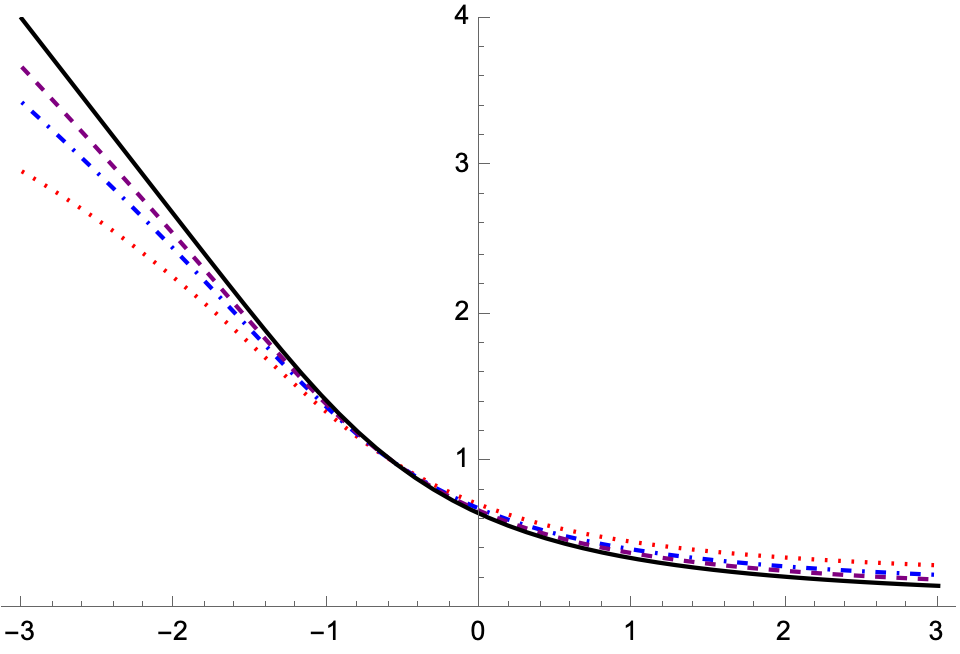}
    \subcaption{Edge}
  \end{minipage}
  \caption{(A) Graph of $p \to \rho_{\rm b}(p)$ and its comparison with $p \mapsto \frac{1}{N}\,\widehat{D}_{1,1}^{(N,1)}(\sqrt{N}p)$, where $N=10, 20$ and $ 30$ (red dotted, blue dot-dashed, and purple dashed lines, respectively). (B) Graph of $\chi \mapsto \rho_{\mathrm{e}}(\chi)$ and its comparison with $\chi \mapsto \frac{1}{\sqrt{N}}\,\widehat{D}_{1,1}^{(N,1)}(\sqrt{N}+\chi)$, where $N=30, 100,$ and $300$.} \label{Fig_CExpectation}
\end{figure}

See Figure~\ref{Fig_CExpectation} for a numerical verification of Theorem~\ref{Thm_CEO}. 

We note that the different scaling in \eqref{scaling for bulk/edge} for the edge is necessary to derive a non-trivial limit. This is also natural, given the conjectural form in \eqref{mean diagonal global density} and the fact that the bulk scaling density \(\rho_{\rm b}\) in \eqref{bulk rho b} vanishes linearly at the edge \(p = \pm 1\). This different scaling for the bulk and edge regimes contrasts with that of the correlation functions in Theorem~\ref{Thm_scaling limits}, where we use the same scaling.

Let us also point out that the edge scaling limit is consistent with the bulk scaling limit in Theorem \ref{Thm_CEO}. Indeed, while in the bulk scaling limit \(\frac{2}{3}(1-p^2) \sim \frac{4}{3}(1-p)\) as \(p \to 1\), the leading constant for the edge scaling limit \(\rho_{\rm e}(\chi) \sim -\frac{4}{3}\chi\) as \(\chi \to -\infty\) coincides with the leading term for the bulk limit. 
We can similarly observe this consistency from the edge to the bulk scaling limit for the limiting skew-kernel of \eqref{GinSEKernel} in the Gaussian case, see \cite[Remark 2.3]{ABK22}.

We now discuss the scaling limits of the correlation functions $\bfR_{N,k}^{(\rm over)}(z_1,z_2,\cdots,z_k)$ associated with \eqref{OverlapWeight}.   
For this purpose, we define the limiting bulk $(\rm b)$ and edge $(\rm e)$ skew-kernel as 
\begin{align}
\begin{split}
\label{Thm Overlap Bulk Kernel}
\kappa_{{\rm b}}^{({\rm over})}(\zeta,\eta)
&:= \frac{1}{2}\frac{(\zeta-\eta)(1+(\zeta-\chi)(\eta-\chi)-e^{2(\zeta-\chi)(\eta-\chi)})}{((\zeta-\chi)(\eta-\chi))^4}
\\
& \quad +\frac{\sqrt{\pi}}{4} \frac{(2(\zeta-\chi)^2-1)(2(\eta-\chi)^2-1)e^{(\zeta-\chi)+(\eta-\chi)^2}\erf(\zeta-\eta)}{((\zeta-\chi)(\eta-\chi))^4}
\\
& \quad +\frac{\sqrt{\pi}}{4} \frac{(2(\eta-\chi)^2-1)((\zeta-\chi)^2-1)e^{(\eta-\chi)^2}\erf(\eta-\chi)  }{((\zeta-\chi)(\eta-\chi))^4}
\\
& \quad -\frac{\sqrt{\pi}}{4} \frac{  (2(\zeta-\chi)^2-1)((\eta-\chi)^2-1)e^{(\zeta-\chi)^2}\erf(\zeta-\chi)  }{((\zeta-\chi)(\eta-\chi))^4} 
\end{split}
\end{align}
and 
\begin{align}
\label{Thm Overlap Edge Kernel}
\kappa_{{\rm e}}^{({\rm over})}(\zeta,\eta)
&:=
\frac{1}{((\zeta-\chi)(\eta-\chi))^4}
\Bigl(
\mathcal{K}(\zeta,\eta|\chi)
-
\frac{\mathcal{A}(\eta,\chi)\mathcal{C}(\zeta,\chi)}{\mathcal{B}(\chi)}
+
\frac{
\mathcal{A}(\zeta,\chi)\mathcal{C}(\eta,\chi)}{\mathcal{B}(\chi)}
\Bigr),
\end{align}
where $\mathcal{A}, \mathcal{B}, \mathcal{C}$ and $\mathcal{K}$ are given by \eqref{cA_Edge}, \eqref{cB_Edge}, \eqref{cC_Edge} and \eqref{cK_Edge}, respectively.

\begin{thm}[\textbf{Scaling limits of eigenvalue correlation functions}] \label{Thm_scaling limits}
Let $p \in [-1,1]$. For $j=1,2,\dots,k$, and $\chi\in\R$, let
\begin{equation} 
z_j=\sqrt{N}p+\zeta_j,\qquad
a=\sqrt{N}p+\chi.
\end{equation}
Then, we have for the $k$-point correlation functions \eqref{def of Pfaff Struc} with respect to the overlap weight \eqref{OverlapWeight} 
\begin{align}
\begin{split}
\label{Lim overlap k point correlation}
\lim_{N\to\infty} \bfR_{N,k}^{({\rm over})}(z_1,\cdots,z_k) & =  R_{k}^{({\rm over})}(\zeta_1,\cdots,\zeta_k) 
\end{split}
\end{align}
uniformly for $\chi$ in a compact subset of $\R$ and for $\zeta_1,\dots,\zeta_k$ in compact subsets of $\C$, where 
\begin{equation}
R_{k}^{({\rm over})}(\zeta_1,\cdots,\zeta_k) 
= \Pf \bigg[
\begin{pmatrix}
\kappa^{({\rm over})}(\zeta_j,\zeta_\ell) 
& \kappa^{({\rm over})}(\zeta_j,\overline{\zeta}_\ell)
\smallskip 
\\
\kappa^{({\rm over})}(\overline{\zeta}_j,\zeta_\ell) 
& \kappa^{({\rm over})}(\overline{\zeta}_j,\overline{\zeta}_\ell)
\end{pmatrix}
\bigg]_{j,\ell=1}^k
\prod_{j=1}^k(\overline{\zeta}_j-\zeta_j) \omega_s(\zeta_j). 
\end{equation}
Here, the weight with shifted exponent (labelled by $s$) is given by 
\begin{equation}
\label{Omega_B}
 \omega_s(\zeta) : =|\zeta-\chi|^2(1+|\zeta-\chi|^2)e^{-2|\zeta-\chi|^2},
\end{equation}
and for the limiting bulk and edge kernel we have to distinguish  
\begin{equation}
\kappa^{({\rm over})}(\zeta,\eta) := 
\begin{cases}
\kappa_{ \rm b }^{({\rm over})}(\zeta,\eta) &\textup{if } p \in (-1,1) \quad \mbox{bulk case}, 
\smallskip
\\
\kappa_{ \rm e }^{({\rm over})}(\zeta,\eta) &\textup{if } p = 1 \quad \quad\quad\ \mbox{edge case}.
\end{cases}
\end{equation}
\end{thm}

  \begin{figure}[t]
   \begin{minipage}[b]{0.32\linewidth}
    \centering
    \includegraphics[width=\textwidth]{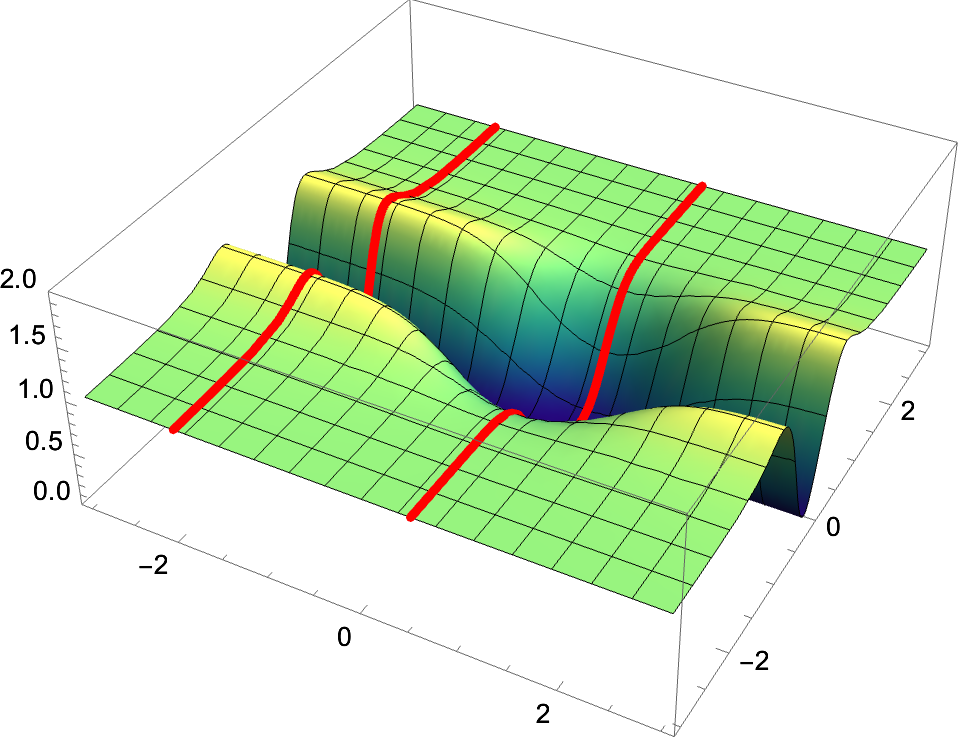}
    \subcaption{bulk  $R_1^{(\mathrm{over})}(x+iy)$ at $\chi=0.5$}
  \end{minipage}
   \begin{minipage}[b]{0.32\linewidth}
    \centering
    \includegraphics[width=\textwidth]{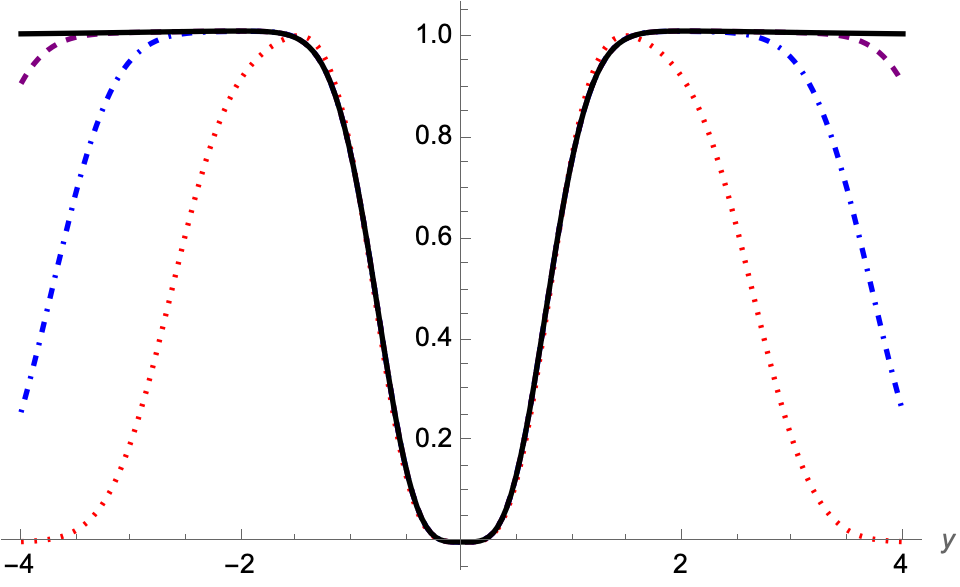}
    \subcaption{bulk  $R_{N,1}^{(\mathrm{over})}(0.5+iy)$ at $\chi=0.5$ }
  \end{minipage}
      \begin{minipage}[b]{0.32\linewidth}
    \centering
    \includegraphics[width=\textwidth]{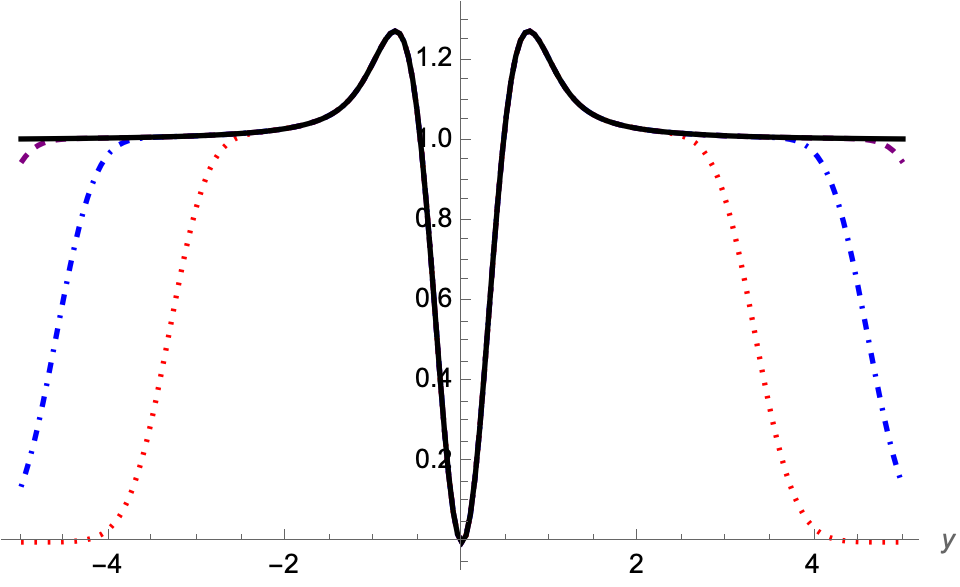}
    \subcaption{bulk  $R_{N,1}^{(\mathrm{over})}(-2+iy)$ at $\chi=0.5$ }
  \end{minipage}
  
      \begin{minipage}[b]{0.32\linewidth}
    \centering
    \includegraphics[width=\textwidth]{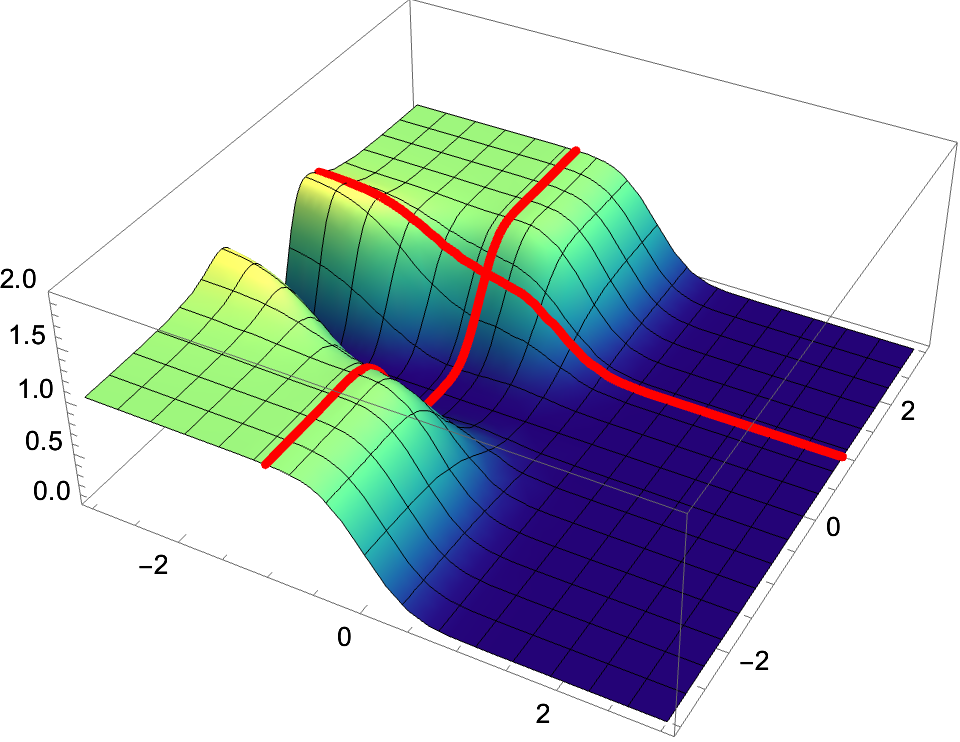}
    \subcaption{edge  $R_1^{(\mathrm{over})}(x+iy)$ $\chi=-1$ }
  \end{minipage}
  \begin{minipage}[b]{0.32\linewidth}
    \centering
    \includegraphics[width=\textwidth]{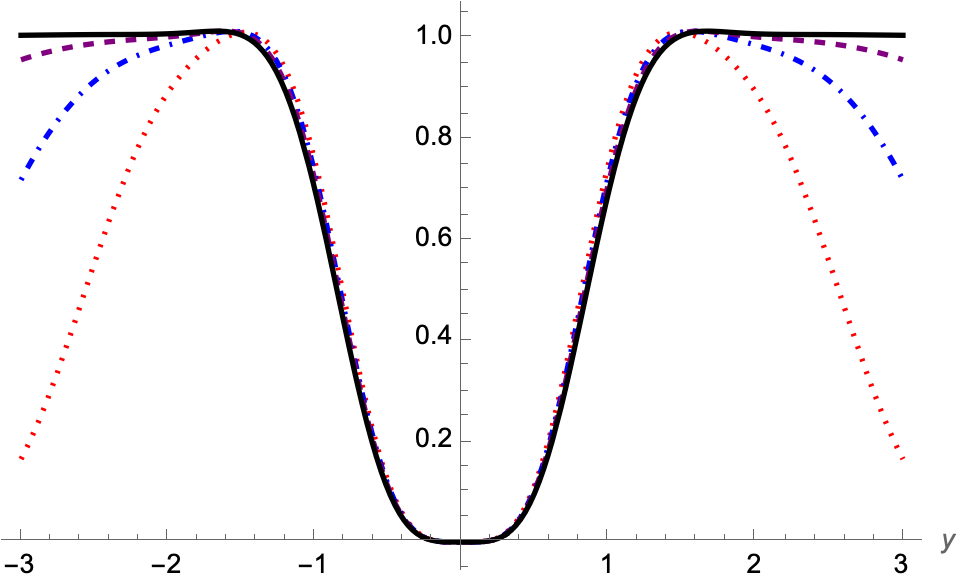}
    \subcaption{edge  $R_{N,1}^{(\mathrm{over})}(-1+iy)$ at $\chi=-1$ }
  \end{minipage}
    \begin{minipage}[b]{0.32\linewidth}
    \centering
    \includegraphics[width=\textwidth]{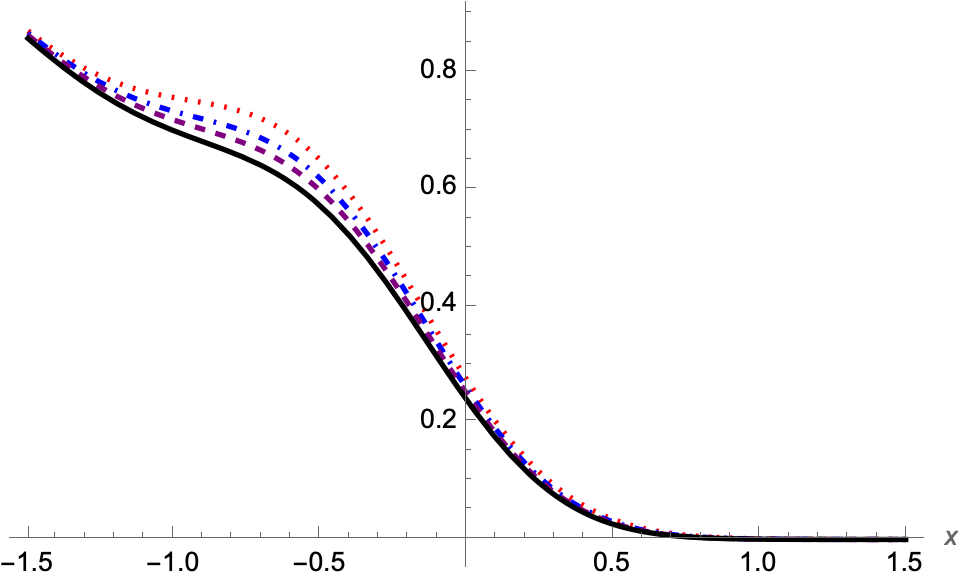}
    \subcaption{edge  $R_{N,1}^{(\mathrm{over})}(x+i)$ at $\chi=-1$ }
  \end{minipage}
  \caption{(A) Graph of the bulk one point density $(x,y)\mapsto R_1^{(\mathrm{over},\mathrm{b})}(x+iy)$ at $\chi=0.5$. (B) and (C) show the graphs of the cross-sections $x=0.5$ and $x=-2$, together with their comparison with $R_{N,1}^{(\mathrm{over})}$ at $\chi=0.5$ with $N=10,20$, and $30$ (red dotted, blue dot-dashed, and purple dashed curves, respectively). (D), (E) and (F) are analogous figures for the edge case at $\chi=-1$, where $N=10, 40$ and $100$.} \label{Fig_scaling limits}
\end{figure}

Note that as an example the limiting local spectral density or $1$-point function is given by 
\begin{equation}
R_1^{ \rm (over) }(\zeta) = \kappa^{ \rm (over) }(\zeta,\overline{\zeta}) ( \overline{\zeta}-\zeta )  \omega_s(\zeta). 
\end{equation}
See Figure~\ref{Fig_scaling limits} for the graph of $R_1^{ \rm (over) }$, as well as a numerical verification of Theorem~\ref{Thm_scaling limits}. 
In particular, one can see the additional repulsion at the point $\chi$.

The rest of this paper is organised as follows. Section~\ref{S_Conditional Origin} is a warm-up case where we demonstrate the conditional origin case. If \(a = 0\), then all the computations become significantly simpler as they can be easily made explicit. This section is for instructive purposes, so that the reader can easily follow the overall strategy of the proofs. 
Section~\ref{S_Construction of skew-orthogonal polynomials} is devoted to the finite-\(N\) analysis. In this section, after demonstrating the transposition lemma (Lemma~\ref{Lem_D12 D11}), we introduce a way to construct the SOPs and prove Theorem~\ref{Thm_OverlapCoe}. Then, we derive a second-order differential equation for the skew-kernel with respect to weight \eqref{PreOverlapWeight}, Theorem~\ref{Thm_ODE_PreKernel}.
Section~\ref{S_Scling limits} is devoted to the large-\(N\) asymptotic analysis. In particular, we prove Theorems~\ref{Thm_CEO} and ~\ref{Thm_scaling limits}.

\section{Conditional origin case and rotationally invariant weight}
\label{S_Conditional Origin}
 
In this section, for pedagogical reasons we first derive results in the much simpler case conditioned at the origin, $a=0$.  
In particular, we present the proof of Theorem~\ref{Thm_scaling limits} for that case. Conditioning a point at the origin simplifies the computations significantly due to the radial symmetry of the weight function \eqref{OverlapWeight}. From this special case, readers can more easily grasp the overall strategy and underlying ideas for subsequent proofs.

Let us first recall the general construction of SOPs for the radially symmetric case. 
Suppose that the underlying measure $\mu$ is radially symmetric, i.e. $\mu(z)=\mu(|z|)$. 
It is obvious that the associated planar orthogonal polynomials defined in Appendix \ref{Appendix_Planar OP} are monomials with squared norms $h_k$
\begin{equation}
p_k(z)= z^k, \qquad h_k= \int_{\mathbb{C}} |z|^{2k} \,d\mu(z). 
\end{equation}
Because of that the associated SOP simply follow, see e.g. \cite[Corollary 3.3]{AEP22}:
\begin{equation} \label{SOP radially symmetric}
q_{2k+1}(z)=z^{2k+1},\qquad q_{2k}(z)=z^{2k}+\sum_{\ell=0}^{k-1} z^{2\ell} \prod_{j=0}^{k-\ell-1} \frac{h_{2\ell+2j+2}}{h_{2\ell+2j+1}}, \qquad r_k=2\,h_{2k+1}. 
\end{equation}
In particular, if the weight is Gaussian, $d\mu^{ \rm (g) }(z)= e^{-2|z|^2}\,dA(z)$, the associated SOP and its norms are given by \eqref{GinSESOP}.

Note that for the case $a=0$, the overlap weight function \eqref{OverlapWeight} and \eqref{PreOverlapWeight} are radially symmetric. 
In this case, we can directly derive the finite-$N$ skew-kernel associated with  both \eqref{OverlapWeight} and \eqref{PreOverlapWeight}. 
However, to prepare the next section we will derive the finite $N$ skew-kernel associated with \eqref{OverlapWeight} from \eqref{PreOverlapWeight} as in Proposition~\ref{Prop_Christoffel pertubation} and analyze the large $N$-limit. 
Therefore, we first demonstrate the case $a=0$ for \eqref{PreOverlapWeight}. Then, by using \eqref{SOP radially symmetric}, it follows that the skew-kernel associated with the weight function \eqref{PreOverlapWeight} is given by
\begin{equation}
    \label{Pre Origin Overlap Kernel}
\bfkappa_N^{(\mathrm{pre})}(z,w)
\equiv
\bfkappa_N^{(\mathrm{pre})}(z,w|0,0)
:=
G_N^{(\mathrm{pre})}(z,w)-G_N^{(\mathrm{pre})}(w,z),
\end{equation}
where
\[
G_N^{(\mathrm{pre})}(z,w)=
\sqrt{2}^3\sum_{k=0}^{N-1}\sum_{\ell=0}^{k}\frac{(2k+3)(2\ell+2)}{(2k+4)!!(2\ell+3)!!}(\sqrt{2}z)^{2k+1}(\sqrt{2}w)^{2\ell}.
\]
Here, we have used 
\begin{equation}
h_k^{\rm (pre)}(0)=\int_{\mathbb{C}} |z|^{2k} \omega^{(\rm pre)}(z,\overline{z}|0,0)dA(z)=
2\int_0^\infty r^{2k+1}(1+r^2)e^{-2r^2}dr=\frac{\Gamma(k+1)(k+3)}{2^{k+2}},
\end{equation}
leading to 
\begin{equation}
\label{Origin Pre Skew Norm}
    r_k^{(\mathrm{pre})}(0) =  (k+2)\frac{(2k+1)!}{2^{2k+1}}. 
\end{equation}
We also write 
\begin{equation} \label{Pre origin hat GN}
\widehat{\bfkappa}_N^{(\mathrm{pre})}(z,w):=(zw)^3\bfkappa_N^{(\mathrm{pre})}(z,w), \qquad 
\widehat{G}_N^{(\mathrm{pre})}(z,w):=(zw)^3G_N^{(\mathrm{pre})}(z,w).  
\end{equation}  
We then show the following. 
\begin{prop}
\label{prop_pre_originkernel_Limit}
We have 
\begin{align}
\begin{split}
\label{OriginODE}
&\quad\Bigl[
z\partial_z^2-(2z^2+2)\partial_z-2z
\Bigr]\widehat{\bfkappa}_N^{(\mathrm{pre})}(z,w)\\
&=
4(zw)^3\sum_{k=0}^{2N-1}\frac{(2zw)^{k}}{k!}
-\frac{1}{2}\frac{(2N+1)(2N+3)}{(2N+2)!!}(\sqrt{2}z)^{2N+3}\sum_{\ell=0}^{N-1}\frac{2\ell+2}{(2\ell+3)!!}(\sqrt{2}w)^{2\ell+3}.
\end{split}
\end{align}
 Furthermore,  as $N\to\infty$, the skew-kernel $\bfkappa_N^{(\rm pre)}(z,w)$ converges to $\varkappa_{\,\rm o}^{(\mathrm{pre})}(z,w)$, uniformly for $z,w$ in compact subsets of $\C$, where 
 \begin{align}
\begin{split}
\label{Limit Pre Bulk Origin}
\varkappa_{\,\rm o}^{(\mathrm{pre})}(z,w)
&= \frac{\sqrt{\pi}}{4}\frac{(2z^2-1)(2w^2-1)}{ (zw)^3 }e^{z^2+w^2}\Big( \mathrm{erf}(z-w) + \mathrm{erf}(w)-\mathrm{erf}(z) \Big) 
\\
& \quad -\frac{z(2w^2-1)}{2(zw)^3} e^{w^2} +\frac{ w(2z^2-1) }{2(zw)^3} e^{z^2} -\frac{z-w}{2(zw)^3}e^{2zw}.
\end{split}
\end{align}
\end{prop}

Here, the subscript o stands for the origin limit. 
We mention that indeed, the differential equation \eqref{Origin_Limit_ODE} and the expression \eqref{Limit Pre Bulk Origin} will play a role in the proof of the bulk case of  Theorem~\ref{Thm_scaling limits} for general real $a\neq0$.

\begin{proof}[Proof of Proposition~\ref{prop_pre_originkernel_Limit}]
We begin with the proof of \eqref{OriginODE}.
By differentiating \eqref{Pre origin hat GN} and shifting the index, we have 
\[
\partial_z\widehat{G}_N^{(\rm pre)}(z,w)
= 4(zw)^3 + \frac12 \sum_{k=0}^{N-2}\sum_{\ell=0}^{k+1}
\frac{(2k+5)(2\ell+2)}{(2k+4)!!(2\ell+3)!!}
(\sqrt{2}z)^{2k+5}(\sqrt{2}w)^{2\ell+3}. 
\]
By rearranging the last term, we have 
\begin{align*}
&\quad \frac12\sum_{k=0}^{N-2}\sum_{\ell=0}^{k+1}\frac{(2k+5)(2\ell+2)}{(2k+4)!!(2\ell+3)!!}(\sqrt{2}z)^{2k+5}(\sqrt{2}w)^{2\ell+3}
\\
&= \frac12\sum_{k=0}^{N-1}\sum_{\ell=0}^{k}\frac{(2k+5)(2\ell+2)}{(2k+4)!!(2\ell+3)!!}(\sqrt{2}z)^{2k+5}(\sqrt{2}w)^{2\ell+3} + \frac12 \sum_{k=0}^{N-2}\frac{(2zw)^{2k+5}}{(2k+2)!!(2k+3)!!}
\\
& \quad - \frac12 \frac{2N+3}{(2N+2)!!}(\sqrt{2}z)^{2N+3}\sum_{\ell=0}^{N-1}\frac{2\ell+2}{(2\ell+3)!!}(\sqrt{2}w)^{2\ell+3},
\end{align*}
which leads to 
\begin{align*}
\partial_z\widehat{G}_N^{(\mathrm{pre})}(z,w)
&= 2z\, \widehat{G}_N^{(\mathrm{pre})}(z,w)
+ 2z^2\sum_{k=0}^{N-1}\sum_{\ell=0}^{k}\frac{(2\ell+2)}{(2k+4)!!(2\ell+3)!!}(\sqrt{2}z)^{2k+3}(\sqrt{2}w)^{2\ell+3}
\\
&\quad + \frac12 \sum_{k=0}^{N-1}\frac{(2zw)^{2k+3}}{(2k+1)!} -\frac12 \frac{2N+3}{(2N+2)!!}(\sqrt{2}z)^{2N+3}\sum_{\ell=0}^{N-1}\frac{2\ell+2}{(2\ell+3)!!}(\sqrt{2}w)^{2\ell+3}.
\end{align*}
Similarly, we have 
\begin{align*}
\partial_z\widehat{G}_N^{(\mathrm{pre})}(w,z)
&= 2z\, \widehat{G}_N^{(\mathrm{pre})}(w,z)
+ 2z^2\sum_{k=0}^{N-1}\sum_{\ell=0}^{k}\frac{2k+3}{(2k+4)!!(2\ell+3)!!}(\sqrt{2}w)^{2k+4}(\sqrt{2}z)^{2\ell+2}
\\
& \quad +2z^2\sum_{k=0}^{N-1}\frac{2k+3}{(2k+4)!!}(\sqrt{2}w)^{2k+4}
-\frac12\sum_{k=0}^{N-1}\frac{(2zw)^{2k+4}}{(2k+2)!}.
\end{align*}
Here, one can express the double summations in the inhomogenous terms by recognising 
\begin{align*}
\int_0^z\frac{1}{t^2}\widehat{G}_N^{(\mathrm{pre})}(t,w) \, dt
&= \frac12 \sum_{k=0}^{N-1}\sum_{\ell=0}^{k}\frac{2\ell+2}{(2k+4)!!(2\ell+3)!!}(\sqrt{2}z)^{2k+3}(\sqrt{2}w)^{2\ell+3},
\\
\int_0^z\frac{1}{t^2}\widehat{G}_N^{(\mathrm{pre})}(w,t)\, dt
&= \frac12 \sum_{k=0}^{N-1}\sum_{\ell=0}^{k}\frac{2k+3}{(2k+4)!!(2\ell+3)!!}(\sqrt{2}w)^{2k+4}(\sqrt{2}z)^{2\ell+2}.
\end{align*}
Then it follows that 
\begin{align*}
\partial_z\widehat{G}_N^{(\mathrm{pre})}(z,w)
&= 2z\, \widehat{G}_N^{(\mathrm{pre})}(z,w)+4z^2\int_{0}^z\frac{\widehat{G}_N^{(\mathrm{pre})}(t,w)}{t^2}\,dt
\\
& \quad + \frac12 \sum_{k=0}^{N-1}\frac{(2zw)^{2k+3}}{(2k+1)!} - \frac12 \frac{2N+3}{(2N+2)!!}(\sqrt{2}z)^{2N+3}\sum_{\ell=0}^{N-1}\frac{2\ell+2}{(2\ell+3)!!}(\sqrt{2}w)^{2\ell+3},
\\
\partial_z\widehat{G}_N^{(\mathrm{pre})}(w,z) &= 2z\, \widehat{G}_N^{(\mathrm{pre})}(w,z)+4z^2\int_0^z\frac{\widehat{G}_N^{(\mathrm{pre})}(w,t)}{t^2} \, dt +2z^2\sum_{k=0}^{N-1}\frac{2k+3}{(2k+4)!!}(\sqrt{2}w)^{2k+4} -\frac12 \sum_{k=0}^{N-1}\frac{(2zw)^{2k+4}}{(2k+2)!}.
\nonumber
\end{align*}
Combining the above, we obtain  
\begin{align*}
\frac{1}{z^2}\partial_z\widehat{\bfkappa}_N^{(\mathrm{pre})}(z,w)
&= \frac{2}{z}\widehat{\bfkappa}_N^{(\mathrm{pre})}(z,w) + 4\int_{0}^{z}\frac{\widehat{\bfkappa}_N^{(\mathrm{pre})}(t,w)}{t^2} \, dt + \frac12 \sum_{k=0}^{2N-1}\frac{2^{k+3}z^{k+1}w^{k+3}}{(k+1)!}
\\
& \quad -\frac{2N+3}{(2N+2)!!}2^{N+\frac12}z^{2N+1}\sum_{\ell=0}^{N-1}\frac{2\ell+2}{(2\ell+3)!!}(\sqrt{2}w)^{2\ell+3} -2\sum_{k=0}^{N-1}\frac{2k+3}{(2k+4)!!}(\sqrt{2}w)^{2k+4}. 
\end{align*}
By differentiating the above with respect to $z$ and rearranging terms, we obtain the desired equation \eqref{OriginODE}.

Next, we shall prove \eqref{Limit Pre Bulk Origin}. 
We first claim that $\widehat{\varkappa}_{\,\rm o}^{(\mathrm{pre})}(z,w)=\lim_{N\to\infty}\widehat{\bfkappa}_{N}^{(\mathrm{pre})}(z,w)$ exists and satisfies the differential equation
\begin{equation}
\Bigl[
z\partial_z^2-(2z^2+2)\partial_z-2z
\Bigr]\widehat{\varkappa}_{\,\rm o}^{(\mathrm{pre})}(z,w)
=
4(zw)^3e^{2zw}. \label{Origin_Limit_ODE}
\end{equation} 
To verify this, our approach is to derive an integral representation of \(\widehat{\bfkappa}_{N}^{(\mathrm{pre})}\) by explicitly solving the differential equation \eqref{OriginODE}, see also \cite[Appendix]{BN24} for a similar argument. For this, we first note that two independent solutions to the homogeneous equation of \eqref{OriginODE} are given by 
\begin{equation*}
f_1(z):=\frac{1}{2}(2z^2-1)e^{z^2}, \qquad
f_2(z):=\frac{1}{4} \Bigl(  \sqrt{\pi}(2z^2-1)e^{z^2}\erf(z)
+2z\Bigr).
\end{equation*}
We denote by $F_N(z,w)$ the right-hand side of \eqref{OriginODE}.  
Then it follows from \eqref{OriginODE} that $\widehat{\bfkappa}_N^{(\mathrm{pre})}$ is of the form 
\[ 
\widehat{\bfkappa}^{(\mathrm{pre})}_N(z,w)=a_N(w)\,f_1(z)+b_N(w)\,f_2(z)+\int_{w}^{z}\frac{f_1(t)f_2(z)-f_1(z)f_2(t)}{\mathcal{W}(f_1,f_2)(t)}F_N(t,w)\,dt  
\]
for some functions $a_N(w)$ and $b_N(w)$. Here $\mathcal{W}(f_1,f_2)(z) = -z^2e^{z^2} $ is the Wronskian of $f_1$ and $f_2$.  
Note that, due to skew-symmetry, we have \(\widehat{\bfkappa}_N^{(\mathrm{pre})}(w,w) = 0\). Moreover, from \eqref{Pre origin hat GN}, it follows that \(\widehat{\bfkappa}_N^{(\mathrm{pre})}(z,w) = O(z^3)\) as \(z \to 0\). These serve as initial conditions which, after straightforward computations, lead to  
$$
a_N(w)=-\int_0^{w}\frac{f_2(t)F_N(t,w)}{\mathcal{W}(f_1,f_2)(t)}\,dt, \qquad
b_N(w)=\frac{f_1(w)}{f_2(w)}\int_0^{w}\frac{f_2(t)F_N(t,w)}{\mathcal{W}(f_1,f_2)(t)}\,dt.
$$  
Note that for \(z, w\) in compact subsets of \(\mathbb{C}\), the function \(F_N(z,w)\) converges absolutely to \(4(zw)^3 e^{2zw}\) as \(N \to \infty\). Then, by the dominated convergence theorem, in the integral expression of $\widehat{\bfkappa}_N^{(\rm pre)}(z,w)$, one can exchange the limit and the integral for the terms involving \(F_N(t,w)\) as \(N \to \infty\). Moreover, since \(\widehat{\bfkappa}_N^{(\mathrm{pre})}(z,w)\) is analytic with respect to \(z\) and $w$, the limit \(\widehat{\varkappa}_{\rm o}(z,w) = \lim_{N\to\infty} \widehat{\bfkappa}_N^{(\mathrm{pre})}(z,w)\) exists.  
This establishes the desired equation \eqref{Origin_Limit_ODE}.

Note that by the skew-symmetry and the pre-factor $(zw)^3$ in \eqref{Pre origin hat GN}, the limiting kernel $\widehat{\varkappa}_{\,\rm o}^{(\mathrm{pre})}$ satisfies  
\begin{equation*}
\widehat{\varkappa}_{\,\rm o}^{(\mathrm{pre})}(w,w)=0,\quad
\partial_z\widehat{\varkappa}_{\,\rm o}^{(\mathrm{pre})}(z,w)|_{z=0}=0.
\end{equation*}
By solving the above ODE with these initial conditions and then dividing by $(zw)^3$, we obtain \eqref{Limit Pre Bulk Origin}.
\end{proof}

Similar to \eqref{OverlapWeight}, 
by using \eqref{SOP radially symmetric} again, it follows that the skew-kernel associated with the weight function \eqref{OverlapWeight} is given by
\begin{equation}
\label{Origin Overlap Kernel}
    \bfkappa_N^{(\rm over)}(z,w)
    :=
    \bfkappa_N^{(\rm over)}(z,w|0,0)
    =
    G_N^{(\rm over)}(z,w)-G_N^{(\rm over)}(w,z),
\end{equation}
where 
\[
G_N^{(\rm over)}(z,w)
:=
\sqrt{\pi}\sum_{k=0}^{N-1}\sum_{\ell=0}^{k}
\frac{(k+2)(\ell+\frac{3}{2})}{\Gamma(k+\frac{7}{2})\Gamma(\ell+3)}
z^{2k+1}w^{2\ell}.
\]
We also have 
\[
Z_{N-1}^{(\rm over)}(0):=(N-1)!\prod_{k=0}^{N-2}\frac{(2k+5)(2k+2)!}{2^{2k+3}}
=
\frac{2}{3}\frac{2N+1}{N}Z_{N}^{(\rm g)}. 
\]
Then by \eqref{Conditional Expectation Diagonal Overlap Real} with $a=0$ 
one can easily verify the $a=0$ case in Theorem~\ref{Thm_CEO}. 
We define 
\begin{equation} 
\widehat{\bfkappa}_N^{(\rm over)}(z,w) := (zw)^4\bfkappa_N^{(\rm over)}(z,w), \qquad  \widehat{G}_N^{(\rm over)}(z,w) := (zw)^4G_N^{(\rm over)}(z,w). 
\end{equation}
For the overlap kernels, we have the analogous results. 

\begin{prop}\label{prop_ODE_OriginODE}
 We have 
 \begin{align}
    \begin{split}
     &\quad
\Bigl[z\partial_z^2-(2z^2+2)\partial_z-2z\Bigr]
\widehat{\bfkappa}_N^{(\rm over)}(z,w)\\
&=
\frac{1}{2}\sum_{k=0}^{2N-1}\frac{(2zw)^{k+4}}{(k+1)!}
-4\sqrt{\pi}\frac{(N+1)(N+2)}{\Gamma(N+\frac{5}{2})}z^{2N+4}
\sum_{\ell=0}^{N-1}\frac{\ell+\frac{3}{2}}{(\ell+2)!}w^{2\ell+4}
-3\sqrt{\pi}z^{3}\sum_{k=0}^{N-1}\frac{k+2}{\Gamma(k+\frac{7}{2})}w^{2k+5}. 
     \end{split}
 \end{align}
 Furthermore,  as $N\to\infty$, the skew-kernel $\bfkappa_N^{(\rm over)}(z,w)$ converges to $\varkappa_{\,\rm o}^{(\mathrm{over})}(z,w)$, uniformly for $z,w$ in compact subsets of $\C$, where 
    \begin{align}
        \begin{split}
        \label{overlap kernel origin limit}
            \varkappa_{\,\rm o}^{(\mathrm{over})}(z,w)
            &=
            \frac{1}{2} \frac{(z-w)(1+zw-e^{2zw})}{(zw)^4}
             +
             \frac{\sqrt{\pi}}{4}\frac{e^{z^2+w^2}(2z^2-1)(2w^2-1)\erf(z-w)}{(zw)^{4}}
            \\
            &\quad
            +
            \frac{\sqrt{\pi}}{4}
            \frac{(z^2-1)(2w^2-1)e^{w^2}\erf(w)-(w^2-1)(2z^2-1)e^{z^2}\erf(z)}{(zw)^4}.
        \end{split}
    \end{align}
\end{prop}
It follows from the same strategy of the proof of Proposition~\ref{prop_pre_originkernel_Limit} and we omit the details. 
Note that \eqref{overlap kernel origin limit} is coincident with \eqref{Thm Overlap Bulk Kernel} for $\chi=0$.
Let us stress that \eqref{overlap kernel origin limit} can also be obtained using \eqref{Limit Pre Bulk Origin} and the Christoffel perturbation, Proposition~\ref{Prop_Christoffel pertubation}. This approach will be used for general real-valued $a$.

\section{Finite-$N$ analysis}
\label{S_Construction of skew-orthogonal polynomials}

In this section, we provide the proofs for finite-$N$ results.

\subsection{Relation between mean off-diagonal overlap and mean diagonal overlap}

In this subsection, we prove Lemma~\ref{Lem_D12 D11}. 

\begin{proof}[Proof of Lemma~\ref{Lem_D12 D11}]
The short proof is a simple matter of spelling out the dependence on all variables $z_1,\overline{z}_1,z_2,\overline{z}_2$ in \eqref{O11evk} at $k=2$, and applying the transposition 
$\widehat{T}$  to $D_{1,1}^{(N,2)}(z_1,z_2)$.
From the definition \eqref{O11evk} and from the result \eqref{O11kPf} for $k=2$ we have, respectively,
\begin{align*}
D_{1,1}^{(N,2)}(z_1,z_2)
&=
\frac{N(N-1)}{Z_N^{(\rm g)}}
(z_1-\overline{z}_1)(\overline{z}_1-z_1)
(z_2-\overline{z}_2)(\overline{z}_2-z_2)
(z_2-\overline{z}_1)(\overline{z}_2-z_1)
\\
&\quad\times
(1+(z_2-z_1)(\overline{z}_2-\overline{z}_1))
e^{-2z_1\overline{z}_1-2z_2\overline{z}_2}
\int_{\C^{N-2}}
\prod_{3\leq j<k \leq N}
|z_j-z_k|^2|z_j-\overline{z}_k|^2
\\
&\quad\times
\prod_{j=3}^N
|z_j-\overline{z}_j|^2
(z_2-z_j)
(\overline{z}_2-\overline{z}_j)
(z_2-\overline{z}_j)
(\overline{z}_2-z_j)
\\
&\quad\times
(z_j-\overline{z}_1)(\overline{z}_j-z_1)
(1+(z_j-z_1)(\overline{z}_j-\overline{z}_1))e^{-2|z_j|^2}
dA(z_j),
\\
&=
\frac{NZ_{N-1}^{(\rm over)}}{Z_N^{(\rm g)}}(z_1-\overline{z}_1)(\overline{z}_1-z_1)
e^{-2z_1\overline{z}_1}
\bfkappa_{N-1}^{(\rm over)}(z_2,\overline{z}_2|z_1,\overline{z}_1)(\overline{z}_2-z_2)\omega^{(\rm over)}(z_2,\overline{z}_2|z_1,\overline{z}_1).
\end{align*}
By noting that the normalization constant $Z_{N-1}^{(\rm over)}\equiv Z_{N-1}^{(\rm over)}(z_1,\overline{z}_1)$ associated with \eqref{OverlapWeight} depends on $z_1,\overline{z}_1$ by definition, and
by a comparison to \eqref{O12ev}, we obtain
\begin{align*}
\widehat{T}D_{1,1}^{(N,2)}(z_1,z_2)
&=
|z_1-\overline{z}_2|^2
(z_2-\overline{z}_2)
(z_1-\overline{z}_1)
(1-|z_1-z_2|^2)
e^{-2z_1\overline{z}_2-2z_2\overline{z}_1}
\\
&\quad\times
\frac{-N(N-1)}{Z_N^{(\rm g)}}
|z_1-\overline{z}_2|^2
\int_{\C^{N-2}}
\prod_{3\leq j<k\leq N}
|z_j-z_k|^2|z_j-\overline{z}_k|^2
\\
&\quad\times
\prod_{j=3}^N
|z_j-\overline{z}_j|^2
|z_j-\overline{z}_1|^2
|z_j-\overline{z}_2|^2
(
|z_j-z_1|^2|z_j-z_2|^2
+
(z_j-z_2)
(\overline{z}_j-\overline{z}_1)
)
e^{-2|z_j|^2}
dA(z_j)
\\
&=
|z_1-\overline{z}_2|^2
(1-|z_1-z_2|^2)
\frac{e^{2|z_1-z_2|^2}
D_{1,2}^{(N,2)}(z_1,z_2)}{(z_2-\overline{z}_2)(z_1-\overline{z}_1)} 
\end{align*}
on the one hand, and 
\begin{align*}
\widehat{T}D_{1,1}^{(N,2)}(z_1,z_2)
&=
-
\frac{NZ_{N-1}^{(\rm over)}(z_1,\overline{z}_2)}{Z_N^{(\rm g)}}(z_1-\overline{z}_2)^2
e^{-2z_1\overline{z}_2}
\bfkappa_{N-1}^{(\rm over)}(z_2,\overline{z}_1|z_1,\overline{z}_2)(\overline{z}_1-z_2)\omega^{(\rm over)}(z_2,\overline{z}_1|z_1,\overline{z}_2)
\end{align*}
on the other hand.
This completes the proof, after inserting the definition of the weight function \eqref{OverlapWeight} in the last line.
\end{proof}

\subsection{Recurrence relation for planar SOPs and derivation of $q_k^{\rm (pre)}(z)$} \label{Subsection_recurrence}

In this section we first derive a general theorem that gives the coefficients of general SOPs when the underlying moment matrix is 
tri-diagonal. It is based on showing certain recurrence relations amongst its coefficients. Because the pre-overlap weight function \eqref{PreOverlapWeight} satisfies this condition we can then apply this theorem to determine the $q_k^{\rm (pre)}(z)$ including their skew-norms.

We first define the moment matrix of a positive Borel measure on $\mathbb{C}$ and its real part as
\begin{equation}
\label{Moment}
m_{j,k}:=
\int_{\C}z^j \overline{z}^k d\mu(z),
\qquad \widehat{m}_{j,k}:=\re(m_{j,k}). 
\end{equation}
We denote the skew-moment matrix with respect to the corresponding skew-product \eqref{inner_product} as
\begin{equation}
\label{SkewMoment1}
g_{j,k}:=\langle z^j, z^k\rangle_s=\int_{\C}\left(z^j \overline{z}^k-z^k\overline{z}^j\right) (z-\overline{z})d\mu(z)
=2(\widehat{m}_{j+1,k}-\widehat{m}_{j,k+1}),\qquad G_k:=(g_{i,j})_{i,j=0}^{2k-1},
\end{equation}
and we define 
\begin{equation}
\label{SkewMoment2}
\Delta_{-1}:=1,\quad \Delta_k:=\Pf(G_{k+1}),\quad \mathcal{Z}_{k+1}:=\frac{1}{2}\frac{\Delta_k}{\Delta_{k-1}}. 
\end{equation}

\begin{thm}[\textbf{Recurrence relation of planar SOPs}] \label{thm_Recurrence}
We assume that the moment matrix \eqref{Moment} is tri-diagonal. Then, the monic SOPs associated with the corresponding skew-product \eqref{inner_product} are obtained as 
\begin{equation}
q_{2k}(z):= \sum_{j=0}^{2k}\alpha_{2k,j}z^j,\quad q_{2k+1}(z):=\sum_{j=0}^{2k+1}\beta_{2k+1,j}z^j,
\end{equation}
with $\alpha_{2k,2k}=1$, $\beta_{2k+1,2k+1}=1$ and $\beta_{2k+1,2k}=\beta_{2k+1,2k-1}=0$. 
The coefficients are determined by the following recurrence relationships: for $j=1,2,...,k$, 
\begin{align} 
\label{EvenCoefficients}
&\cZ_j\alpha_{2k,2j-1} = - \widehat{m}_{2j-1,2j}\alpha_{2k,2j}, 
\qquad \qquad \cZ_j\alpha_{2k,2j-2} = \widehat{m}_{2j,2j+1}\alpha_{2k,2j+1} + \widehat{m}_{2j,2j}\alpha_{2k,2j},  
\\
& \label{OddCoefficients}
\cZ_j\beta_{2k+1,2j-1}  = -\widehat{m}_{2j-1,2j}\beta_{2k+1,2j}, 
\qquad  \cZ_j\beta_{2k+1,2j-2} =\widehat{m}_{2j,2j+1}\beta_{2k+1,2j+1}+\widehat{m}_{2j,2j}\beta_{2k+1,2j}, 
\end{align}
where $\mathcal{Z}_j$ is given by \eqref{SkewMoment2} and we use the convention $\alpha_{2k,2k+1}=0$. 
\end{thm}

We mention that the initial condition $\beta_{2k+1,2k}=0$ in Theorem~\ref{thm_Recurrence} uniquely determines the odd SOPs, see \cite[Lemma 2.2]{AEP22}.

\begin{proof}[Proof of Theorem~\ref{thm_Recurrence}]
We shall use an induction argument.  
Suppose that the coefficients of $q_{2k}$ satisfy \eqref{EvenCoefficients} and those of $q_{2k+1}$ satisfy \eqref{OddCoefficients}.  

Note that for $1\leq j\leq k-1$, we have 
\begin{align*}
\langle q_{2j-2},q_{2k+2} \rangle_s & =\langle z^{2j-2}+\cdots, \alpha_{2k+2,2k-1}z^{2k-1}+\cdots\rangle_s,
\\
\langle q_{2j-1},q_{2k+2} \rangle_s & = \langle z^{2j-1}+\cdots, \alpha_{2k+2,2k-1}z^{2k-1}+\cdots\rangle_s.
\end{align*}
By the induction assumption and replacing $2k$ with $2k+2$, 
it follows that for $1\leq j\leq k-1$, 
\[
\cZ_j\alpha_{2k+2,2j-1}=-\widehat{m}_{2j-1,2j}\alpha_{2k+2,2j},\qquad 
\cZ_j\alpha_{2k+2,2j-2}=\widehat{m}_{2j,2j+1}\alpha_{2k+2,2j+1}+\widehat{m}_{2j,2j}\alpha_{2k+2,2j}. 
\]
Therefore, it suffices to consider $q_{2j-2},q_{2j-1}$ for $j=k,k+1$. 
First, note that 
\begin{align*}
\langle q_{2k-2},q_{2k+2}\rangle_s 
&= 2\alpha_{2k+2,2k}\widehat{m}_{2k-1,2k} +2\alpha_{2k+2,2k-1} \widehat{m}_{2k-1,2k-1}
\\
&\quad -2\alpha_{2k+2,2k-3} \widehat{m}_{2k-2,2k-2}-2\alpha_{2k+2,2k-4}\widehat{m}_{2k-2,2k-3} +\Bigl\langle \sum_{j=0}^{2k-3}\alpha_{2k-2,j}z^j, \sum_{j=0}^{2k-1}\alpha_{2k+2,j}z^j \Bigr\rangle_s.
\end{align*}
To show \eqref{EvenCoefficients}, we need to check 
\begin{align*}
\cZ_{k}\alpha_{2k+2,2k-1} &= \frac12 \bigl\langle \sum_{j=0}^{2k-3}\alpha_{2k-2,j}z^j, \sum_{j=0}^{2k-1}\alpha_{2k+2,j}z^j \bigr\rangle_s
\\
&\quad + \alpha_{2k+2,2k-1}\widehat{m}_{2k-1,2k-1} -\alpha_{2k+2,2k-3}\widehat{m}_{2k-2,2k-2} - \alpha_{2k+2,2k-4}\widehat{m}_{2k-2,2k-3}. 
\end{align*} 
On the other hand, since 
\begin{align*}
&\quad \alpha_{2k+2,2k-1}\widehat{m}_{2k-1,2k-1} -\alpha_{2k+2,2k-3}\widehat{m}_{2k-2,2k-2} -\alpha_{2k+2,2k-4}\widehat{m}_{2k-2,2k-3}
\\
&= \alpha_{2k+2,2k-1}\Bigl(\widehat{m}_{2k-1,2k-1}-\frac{\widehat{m}_{2k-3,2k-2}\widehat{m}_{2k-2,2k-1}}{\cZ_{k-1}} \Bigr)= \cZ_{k}\alpha_{2k+2,2k-1},
\end{align*}
it is enough to show that 
\[
\Bigl\langle \sum_{j=0}^{2k-3}\alpha_{2k-2,j}z^j, \sum_{j=0}^{2k-1}\alpha_{2k+2,j}z^j \Bigr\rangle_s=0. 
\]
For $j=3,4,\dots, 2k-2$, we write 
\begin{align}
\begin{split}
u_{k,j}& =
\alpha_{2k+2,2k+2-j}\widehat{m}_{2k+1-j,2k+2-j}+\alpha_{2k+2,2k+1-j}\widehat{m}_{2k+1-j,2k+1-j}
\\
&\quad  -\alpha_{2k+2,2k-j-1}\widehat{m}_{2k-j,2k-j}-\alpha_{2k+2,2k-j-2}\widehat{m}_{2k-j,2k-j-1}.
\end{split}
\end{align}
Then we have  
\begin{align*}
\frac12 \bigl\langle \sum_{j=0}^{2k-3}\alpha_{2k-2,j}z^j, \sum_{j=0}^{2k-1}\alpha_{2k+2,j}z^j \bigr\rangle_s
&= \sum_{j=3}^{2k-2}\alpha_{2k-2,2k-j}u_{k,j}  +\alpha_{2k-2,0}(\alpha_{2k+2,2}\widehat{m}_{1,2}+\alpha_{2k+2,1}\widehat{m}_{1,1})
\\
&\quad +\alpha_{2k-2,1}(\alpha_{2k+2,3}\widehat{m}_{2,3}+\alpha_{2k+2,2}\widehat{m}_{2,2}-\alpha_{2k+2,0}\widehat{m}_{1,1}).
\end{align*} 
By the assumption of the induction, 
\begin{align*}
u_{k,2j}&= \alpha_{2k+2,2k-2j+2}\widehat{m}_{2k-2j+1,2k-2j+2}+\alpha_{2k+2,2k-2j+1}\widehat{m}_{2k-2j+1,2k-2j+1}
\\
&\quad +\widehat{m}_{2k-2j,2k-2j}\frac{\widehat{m}_{2k-2j-1,2k-2j}\alpha_{2k+2,2k-2j}}{\cZ_{k-j}}
\\
& \quad  -\widehat{m}_{2k-2j-1,2k-2j}\frac{\widehat{m}_{2k-2j+1,2k-2j}\alpha_{2k+2,2k-2j+1}+\widehat{m}_{2k-2j,2k-2j}\alpha_{2k+2,2k-2j}}{\cZ_{k-j}}
\\
&= \widehat{m}_{2k-2j+1,2k-2j+2}\alpha_{2k+2,2k+2-2j}+\cZ_{k-j+1}\alpha_{2k+2,2k-2j+1}=0.
\end{align*}
Similarly, one can show that $u_{k,2j+1}=0$. 
Hence, we have shown that $u_{k,j}=0$ for all $3\leq j\leq 2k-2$, which leads to the desired conclusion. 

Next, to show $\langle q_{2k-1},q_{2k+2}\rangle _s=0$, first observe that 
\[
\langle q_{2k-1},q_{2k+2}\rangle_s
= 2\alpha_{2k+2,2k+1}\widehat{m}_{2k,2k+1}+2\alpha_{2k+2,2k}\widehat{m}_{2k,2k}-2\cZ_{k}\alpha_{2k+2,2k-2}
+
\Bigl\langle 
\sum_{j=0}^{2k-4}\beta_{2k-1,j}z^j,\sum_{j=0}^{2k-2}\alpha_{2k+2,j}z^{j}
\Bigr\rangle_s.
\]
Therefore, it suffices to show that 
\[
\Bigl\langle 
\sum_{j=0}^{2k-4}\beta_{2k-1,j}z^j,\sum_{j=0}^{2k-2}\alpha_{2k+2,j}z^{j}
\Bigr\rangle_s=0.
\]
Similar to the even coefficient case, this follows from \( u_{k,j}=0 \) for \( 4 \leq j \leq 2k-2 \), which has already been made. Hence, we have verified the induction argument for \( q_{2j-2}, q_{2j-1} \) with \( j=k \).

The remaining task is to verify the induction argument for \( q_{2j-2}, q_{2j-1} \) with \( j=k+1 \). For \eqref{EvenCoefficients}, we can reuse computations since the highest degree term of \( q_{2k+2} \) is \( z^{2k+2} \). Hence, there is no additional task to show the induction argument for \eqref{EvenCoefficients}. For \eqref{OddCoefficients}, similarly, we can reuse the computations made so far. 
This completes the proof.
\end{proof}

We now turn our attention to Theorem~\ref{Thm_OverlapCoe}.  
For $a\in\R$, specifying to the pre-overlap weight \eqref{PreOverlapWeight}, we have
\begin{equation}
\label{Overlap_Moment}
m_{i,j}^{(\mathrm{pre})}:=\int_{\C}z^i\overline{z}^j\omega^{(\mathrm{pre})}(z) \, dA(z) =
\begin{cases} 
-a\,j!\,2^{-j-1} & \text{if $i=j-1$}, 
\smallskip 
\\
(2a^2+j+3) \, j!\, 2^{-j-2}& \text{if $i=j$}, 
\smallskip 
\\
-a \, (j+1)!\,2^{-j-2} & \text{if $i=j+1$}, 
\smallskip 
\\
0 & \text{if $i<j-1$ or $i>j+1$}, 
\end{cases}
\end{equation}
which is obviously real and Hermitian.
We also write the corresponding skew-moment matrix as 
\begin{equation}
\label{Overlap_Skew_Moment}
g_{i,j}^{(\mathrm{pre})}:=\int_{\C}(z^i\overline{z}^j-\overline{z}^jz^j)(z-\overline{z})\omega^{(\mathrm{pre})}(z)dA(z),
\qquad 
G_k^{(\mathrm{pre})}:=(g_{i,j}^{(\mathrm{pre})})_{i,j=0}^{2k-1}.
\end{equation}
We define the coefficients $\{\alpha_{2n,j}^{(\mathrm{pre})}\}_{j=0}^{2n}$ and $\{\widetilde{\beta}_{2n+1,j}^{(\mathrm{pre})}\}_{j=0}^{2n+1}$ associated with \eqref{Overlap_Skew_Moment} obtained from Theorem~\ref{thm_Recurrence}.

In our setting, it is convenient to consider $q_{2k+1}(z)\mapsto q_{2k+1}(z)+aq_{2k}(z)$ rather than the original odd SOPs, which does not change the value of the skew-inner products by \cite[Lemma 2.2]{AEP22}. Then by definition, 
\begin{equation}
\beta_{2n+1,j}^{(\mathrm{pre})}=\widetilde{\beta}_{2n+1,j}^{(\mathrm{pre})}+a\alpha_{2n,j}^{(\mathrm{pre})}, 
\end{equation}
for $j=0,1,\dots,2n$. We also have 
$$
\beta_{2n+1,2n+1}^{(\mathrm{pre})}=1, \qquad \beta_{2n+1,2n}^{(\mathrm{pre})}=a. 
$$ 
By  Theorem~\ref{thm_Recurrence}, we have the following recurrence relationships:
\begin{align}
\begin{split} 
&\label{RS_RC} \cZ_j\alpha_{2n,2j-1}^{(\mathrm{pre})}  =-m_{2j-1,2j}^{(\mathrm{pre})}\alpha_{2n,2j}^{(\mathrm{pre})},\qquad  \qquad 
\cZ_j\alpha_{2n,2j-2}^{(\mathrm{pre})} = m_{2j,2j+1}^{(\mathrm{pre})}\alpha_{2n,2j+1}^{(\mathrm{pre})}+m_{2j,2j}^{(\mathrm{pre})}\alpha_{2n,2j}^{(\mathrm{pre})},
\\
&\cZ_j\beta_{2n+1,2j-1}^{(\mathrm{pre})}=-m_{2j-1,2j}^{(\mathrm{pre})}\beta_{2n+1,2j}^{(\mathrm{pre})},\qquad 
\cZ_j\beta_{2n+1,2j-2}^{(\mathrm{pre})}=m_{2j,2j+1}^{(\mathrm{pre})}\beta_{2n+1,2j+1}^{(\mathrm{pre})}+m_{2j,2j}^{(\mathrm{pre})}\beta_{2n+1,2j}^{(\mathrm{pre})}.
\end{split}
\end{align}

We are now ready to prove Theorem~\ref{Thm_OverlapCoe}. 
\begin{proof}[Proof of Theorem~\ref{Thm_OverlapCoe}]
By using \eqref{Overlap_Moment} and \eqref{Overlap_Skew_Moment}, we define
\begin{equation}
\label{Overlap_Delta}
\Delta_{-1}^{(\mathrm{pre})}:=1,\qquad 
\Delta_{k}^{(\mathrm{pre})}:=\Pf(G_{k+1}^{(\mathrm{pre})}),\qquad 
\cZ_{k+1}^{(\mathrm{pre})}:=\frac{1}{2}\frac{\Delta_{k}^{(\mathrm{pre})}}{\Delta_{k-1}^{(\mathrm{pre})}}.
\end{equation}
Then, by \eqref{Overlap_Moment} and \eqref{Overlap_Skew_Moment}, we have
\begin{equation}
\cZ_{k+1}^{(\mathrm{pre})}=m_{2k+1,2k+1}^{(\mathrm{pre})}-\frac{m_{2k-1,2k}^{(\mathrm{pre})}m_{2k,2k+1}^{(\mathrm{pre})}}{\cZ_{k}^{(\mathrm{pre})}},
\end{equation}
or equivalently, 
\begin{equation}
\label{Overlap_Delta2}
\Delta_{k}^{(\mathrm{pre})}=2\Delta_{k-1}^{(\mathrm{pre})}m_{2k+1,2k+1}^{(\mathrm{pre})}-4m_{2k-1,2k}^{(\mathrm{pre})}m_{2k,2k+1}^{(\mathrm{pre})}\Delta_{k-2}^{(\mathrm{pre})}.
\end{equation}
Then, by induction, the unique solution of \eqref{Overlap_Delta2} is given by  
\[
\Delta_k^{(\mathrm{pre})}=(k+1)!\frac{2^{k+1}\prod_{i=1}^{k+1}\Gamma(2i)}{2^{(k+1)(k+2)}}f_{k+1}(a^2),
\]
where $f_k$ is given by \eqref{fsum}.
Hence, we obtain 
\begin{equation}
\label{cZ pre}
    \cZ_k^{(\rm pre)}
    =
    \frac{\Gamma(2k+1)}{2^{2k+1}}
    \frac{f_k(a^2)}{f_{k-1}(a^2)}.
\end{equation}
Notice also that the skew-norm associated with \eqref{PreOverlapWeight} is given by \eqref{Skew norm PreWeight}.

Let us fix $k\in\mathbb{N}$, and assume that for some $j\in\mathbb{Z}_{\geq 0}$,
\begin{align}
\begin{split}
\label{Induction even}
\alpha_{2k,2(k-j)}^{(\rm pre)}
&= \frac{1}{2^j}\frac{f_{k-j}(a^2)}{f_k(a^2)} \sum_{\ell=0}^{j}  \bigg( \frac{(2k+3)!!\,(\ell+1)(2a^2)^{\ell}}{(2\ell+2k-2j+3)!!}- \frac{(2k+1)!!\,\ell(2a^2)^{\ell}}{(2\ell+2k-2j+1)!!}
\bigg),  
\end{split}\\
\begin{split}
\label{Induction odd}
\alpha_{2k,2(k-j)+1}^{(\rm pre)}
&= \frac{a}{2^{j-1} }\frac{f_{k-j}(a^2)}{f_k(a^2)}  
\sum_{\ell=0}^{j-1} \bigg( \frac{(2k+3)!!\,(\ell+1)(2a^2)^{\ell}}{(2\ell+2k-2j+5)!!}-\frac{(2k+1)!!\,\ell(2a^2)^{\ell}}{(2\ell+2k-2j+3)!!}
\bigg).
\end{split}
\end{align}
It is enough to show that $\alpha_{2k,2(k-j-1)}^{(\rm pre)}$ and $\alpha_{2k,2(k-j-1)+1}^{(\rm pre)}$ also satisfy \eqref{Induction even} and \eqref{Induction odd}, respectively.  
By Theorem~\ref{thm_Recurrence}, we have
\[
\cZ_{k-j}^{(\rm pre)}\alpha_{2k,2(k-j-1)+1}^{(\rm pre)}
=
a\frac{\Gamma(2(k-j)+1)}{2^{2(k-j)+1}}
\alpha_{2k,2(k-j)}.
\]
By \eqref{Induction even} and \eqref{cZ pre}, one can observe that $\alpha_{2k,2(k-j-1)+1}^{(\rm pre)}$ also satisfies \eqref{Induction odd}. 
Next, by Theorem~\ref{thm_Recurrence}, \eqref{Induction even}, and \eqref{Induction odd}, we have 
\[
\cZ_{k-j}^{(\rm pre)}\alpha_{2k,2(k-j-1)}^{(\rm pre)}
=
-a\frac{(2k-2j+1)!}{2^{2k-2j+2}}\alpha_{2k,2(k-j)+1}^{(\rm pre)}
+
\frac{(2k-2j)! \, (2a^2+2k-2j+3)}{2^{2k-2j+2}}\alpha_{2k,2(k-j)}^{(\rm pre)}. 
\]
Hence, we have 
\[
\alpha_{2k,2(k-j-1)}^{(\rm pre)}
=
\frac{1}{2^{j+1}}\frac{f_{k-j-1}(a^2)}{f_k(a^2)}
(\Psi_1-\Psi_2),
\]
where
\begin{align*}
\Psi_1
&:=
(2a^2+2k-2j+3)\Bigl( 
\sum_{\ell=0}^{j}\frac{(2k+3)!!(\ell+1)(2a^2)^{\ell}}{(2\ell+2k-2j+3)!!}
-
\sum_{\ell=0}^{j}\frac{(2k+1)!!\ell(2a^2)^{\ell}}{(2\ell+2k-2j+1)!!}
\Bigr),
\\
\Psi_2
&:=
2a^2(2k-2j+1)\Bigl( 
\sum_{\ell=0}^{j-1}\frac{(2k+3)!!(\ell+1)(2a^2)^{\ell}}{(2\ell+2k-2j+5)!!}
-
\sum_{\ell=0}^{j-1}\frac{(2k+1)!!\ell(2a^2)^{\ell}}{(2\ell+2k-2j+3)!!}
\Bigr).
\end{align*}
Note that 
\begin{align*}
\Psi_1&=\sum_{\ell=0}^{j+1} \bigg(  \frac{(2k+3)!!(\ell+1)(2a^2)^{\ell}}{(2\ell+2k-2j+1)!!} -\frac{(2k+1)!!\ell(2a^2)^{\ell}}{(2\ell+2k-2j-1)!!} \bigg)
-\sum_{\ell=1}^{j+1}\bigg( \frac{(2k+3)!!(2a^2)^{\ell}}{(2\ell+2k-2j+1)!!}-\frac{(2k+1)!!(2a^2)^{\ell}}{(2\ell+2k-2j-1)!!} \bigg)
\\
&\quad +(2k-2j+3) \sum_{\ell=1}^{j}\bigg( \frac{(2k+3)!!(\ell+1)(2a^2)^{\ell}}{(2\ell+2k-2j+3)!!}-\frac{(2k+1)!!\ell(2a^2)^{\ell}}{(2\ell+2k-2j+1)!!} \bigg).
\end{align*} 
Then, after some straightforward simplifications, we obtain  
\[
\Psi_1-\Psi_2
=
\sum_{\ell=0}^{j+1}\frac{(2k+3)!!(\ell+1)(2a^2)^{\ell}}{(2\ell+2k-2j+1)!!}-\sum_{\ell=0}^{j+1}\frac{(2k+1)!!\ell (2a^2)^{\ell}}{(2\ell+2k-2j-1)!!}, 
\]
which shows that $\alpha_{2k,2(k-j-1)}^{(\rm pre)}$ also satisfies \eqref{Induction even}. 
A similar argument for \eqref{Coe Beta Even} and \eqref{Coe Beta Odd} can be made, and we leave it to the interested reader to verify this.
\end{proof}

\subsection{Differential equation for the pre-overlap kernel}

In this subsection, we prove Theorem~\ref{Thm_ODE_PreKernel}. 
Let us write 
\begin{equation}
\cL_{k}(z,a)
:= \Bigl(2(z-a)^2-1\Bigr)\bfkappa_{k}^{(\rm g)}(z,a) + 2(z-a)e_{2k-1}(2za) -(z-a)\frac{2^{k+1}z^{2k}}{(2k-1)!!}e_{k-1}(a^2),
\label{Lfunc}
\end{equation}
where $\bfkappa_k^{ (\rm g) }$ and $e_k$ are given by \eqref{GinSEKernel} and \eqref{exponential_sum}. 
We first rewrite the SOPs in the previous subsection in terms of $\cL_k$. 

\begin{lem}
\label{RewriteSOP}
We have 
\begin{align}
\label{qodd3}
(z-a)^3q_{2k+1}^{(\mathrm{pre})}(z) &= (z-a)^2z^{2k+2} + a\frac{(2k+1)!!}{2^{k+2}f_k(a^2)} \cL_{k+1}(z,a),
\\
\label{qeven3}
(z-a)^3q_{2k}^{(\mathrm{pre})}(z) &= a(z-a)^2\frac{e_{k+1}(a^2)}{f_k(a^2)}z^{2k+2} + \frac{(2k+3)!!}{2^{k+3}f_k(a^2)}\cL_{k+2}(z,a) - a^2\frac{(2k+1)!!}{2^{k+2}f_k(a^2)}\cL_{k+1}(z,a).
\end{align} 
\end{lem}

\begin{proof}
We present only the proof of \eqref{qeven3}, as that of \eqref{qodd3} follows in a similar manner with minor modifications.
First, note that 
\begin{align*}
\begin{split}
&\quad f_k(a^2)q_{2k}^{(\mathrm{pre})}(z)
\\
&= \sum_{\ell=0}^{k}\sum_{p=0}^{\ell}
\frac{(2k+3)!!}{(2\ell+3)!!}\frac{2^{\ell-k}a^{2\ell+2p}}{p!}
\bigg( \sum_{j=p}^{\ell}(\ell+1-j)(j+1-p)\left(\frac{z}{a}\right)^{2j} + \sum_{j=p}^{\ell-1}(\ell-j)(j+1-p)\left(\frac{z}{a}\right)^{2j+1} \bigg) 
\\
&\quad -\sum_{\ell=0}^{k}\sum_{p=0}^{\ell}
\frac{(2k+1)!!}{(2\ell+1)!!}\frac{2^{\ell-k}a^{2\ell+2p}}{p!}
\bigg( \sum_{j=p}^{\ell}(\ell-j)(j+1-p)\left(\frac{z}{a}\right)^{2j} + \sum_{j=p}^{\ell-1}(\ell-1-j)(j+1-p)\left(\frac{z}{a}\right)^{2j+1}
\bigg).
\end{split}
\end{align*} 
By elementary geometric series identities, the above can be rewritten as 
\begin{equation}
\label{q2k}
2^k(z+a)(z-a)^3f_k(a^2)q_{2k}^{(\mathrm{pre})}(z)
=
(2k+3)!!\mathcal{M}_1
+
(2k+1)!!\mathcal{M}_2,
\end{equation}
where 
\begin{align*}
\mathcal{M}_1
=&
\sum_{\ell=0}^{k}\frac{2^{\ell}a^{2\ell+3}}{(2\ell+3)!!}
\sum_{p=0}^{\ell}\frac{(\ell+2-p)z^{2p+1}}{p!}
-
\sum_{\ell=0}^{k}\frac{2^{\ell}z^{2\ell+3}}{(2\ell+3)!!}
\sum_{p=0}^{\ell}\frac{(\ell+2-p)a^{2p+1}}{p!}
\\
&
-\sum_{\ell=0}^{k}\frac{2^{\ell}a^{2\ell+4}}{(2\ell+3)!!}
\sum_{p=0}^{\ell}\frac{(\ell+1-p)z^{2p}}{p!}
+
\sum_{\ell=0}^{k}\frac{2^{\ell}z^{2\ell+4}}{(2\ell+3)!!}
\sum_{p=0}^{\ell}\frac{(\ell+1-p)a^{2p}}{p!},
\\
\mathcal{M}_2
=&
-
\sum_{\ell=0}^{k}\frac{2^{\ell}a^{2\ell+3}}{(2\ell+1)!!}
\sum_{p=0}^{\ell}\frac{(\ell+1-p)z^{2p+1}}{p!}
+
\sum_{\ell=0}^{k}\frac{2^{\ell}z^{2\ell+1}}{(2\ell+1)!!}
\sum_{p=0}^{\ell}\frac{(\ell+1-p)a^{2p+3}}{p!}
\\
&
+
\sum_{\ell=0}^{k}\frac{2^{\ell}a^{2\ell+4}}{(2\ell+1)!!}
\sum_{p=0}^{\ell}\frac{(\ell-p)z^{2p}}{p!}
-
\sum_{\ell=0}^{k}\frac{2^{\ell}z^{2\ell+2}}{(2\ell+1)!!}
\sum_{p=0}^{\ell}\frac{(\ell-p)a^{2p+2}}{p!}.
\end{align*}
After some straightforward manipulations, $\mathcal{M}_1$ can be expressed in terms of $ \bfkappa_{k+1}^{(\rm g )}(z,a)$ as 
\begin{align*}
\mathcal{M}_1
&= \frac{(2(z-a)^2-1)(z+a)}{8} \bfkappa_{k+1}^{(\rm g )}(z,a) -(z-a)z^2\frac{2^{k}a^{2k+3}}{(2k+3)!!} e_k(z^2) -(z-a)a^2\frac{2^{k}z^{2k+3}}{(2k+3)!!} e_k(a^2) 
\\
& \quad + \frac{(z+a)}{2} \Big(\frac{2^{k}a^{2k+3}}{(2k+3)!!} e_k(z^2) -\frac{2^{k}z^{2k+3}}{(2k+3)!!} e_k(a^2) \Big)  - \frac{(z-a)(z+a)}{4} \Big(\frac{2^{2k+2}(za)^{2k+2}}{(2k+3)!} - e_{2k+2}(2za) \Big).
\end{align*}
Similarly, we have 
\begin{align*}
\mathcal{M}_2
&= -\frac{a^2(2(z-a)^2-1)(z+a)}{4}\bfkappa_{k+1}^{(\rm g)}(z,a) - \frac{a^2(z-a)(z+a)}{2} e_{2k+1}(2za)
\\
& \quad + (z-a)\frac{2^{k}a^{2k+5}}{(2k+1)!!} e_k(z^2) + (z-a)\frac{2^{k}z^{2k+3}a^2}{(2k+1)!!} e_k(a^2).
\end{align*}
Combining all of the above, we obtain \eqref{qeven3}. 
\end{proof}

By using \eqref{qeven3} and \eqref{qodd3}, we write 
\begin{align}
\begin{split}
\label{Skeq Odd Hat}
\widehat{q}^{\,(\mathrm{pre})}_{2k+1}(z)
&:= 2^{k+2} e^{-2za} (z-a)^3 f_{k}(a^2) q_{2k+1}^{\,(\mathrm{pre})}(z)
\\
&\,= \Big( 
2^{k+2}f_k(a^2)(z-a)^2z^{2k+2}+a(2k+1)!!\cL_{k+1}(z,a)
\Big)e^{-2za},
\end{split}
\\
\begin{split}
\label{Skeq Even Hat}
\widehat{q}^{\,(\mathrm{pre})}_{2k}(z)
&:= 2^{k+3} e^{-2za} (z-a)^3 f_{k}(a^2) q_{2k}^{\,(\mathrm{pre})}(z)
\\
&\, = \Big( 2^{k+3}ae_{k+1}(a^2)(z-a)^2z^{2k+2}+(2k+3)!!\cL_{k+2}(z,a)-2a^2(2k+1)!!\cL_{k+1}(z,a) \Big) e^{-2za}.
\end{split}
\end{align}
Let us also denote 
\begin{equation}
\label{ExpLk0}
\widehat{\cL}_{k}(z,a) := e^{-2za}\cL_{k}(z,a).
\end{equation}
Note that by \cite[Eq (10.29)]{BF24}, we have 
\begin{equation}
\label{RGinK}
\partial_z\widehat{\bfkappa}_{k}^{(\rm{g})}(z,a) = 2(z-a)\widehat{\bfkappa}_{k}^{(\rm{g})}(z,a) + 2e_{2k-1}(2za)e^{-2za} -2\frac{(2z^2)^k}{(2k-1)!!}e_{k-1}(a^2)e^{-2za}.
\end{equation}
By using \eqref{RGinK}, one can check that \eqref{ExpLk0} is equivalent to \eqref{ExpLk2}. 
We further define 
\[
\widehat{\bfkappa}_N^{(\mathrm{pre})}(z,w)
:=\widehat{G}_N^{\,(\mathrm{pre})}(z,w)-\widehat{G}_N^{\,(\mathrm{pre})}(w,z),
\]
where we set 
\begin{align*} 
\widehat{G}_{N}^{\,(\mathrm{pre})}(z,w)
&:= ((z-a)(w-a))^3 e^{-2(z+w)a} \sum_{k=0}^{N-1}\frac{q_{2k+1}^{(\mathrm{pre})}(z)q_{2k}^{\,(\mathrm{pre})}(w)}{r_k^{(\mathrm{pre})}} = \frac{1}{8} \sum_{k=0}^{N-1} \frac{\widehat{q}_{2k+1}^{(\mathrm{pre})}(z) \widehat{q}_{2k}^{\,(\mathrm{pre})}(w)}{(2k+2)!f_{k+1}(a^2)f_{k}(a^2)}. 
\end{align*}

For the proof of Theorem~\ref{Thm_ODE_PreKernel}, we first prove the following lemma. Recall that $\mathfrak{D}_{z,a}$ is given by \eqref{OP_DG}. 

\begin{lem}\label{Lem_Part1}
We have  
\begin{align*}
\begin{split}
&\quad \frac{8\,  e^{2wa} }{(z-a)^3(w-a)^3} \mathfrak{D}_{z,a}\widehat{G}_N^{(\rm pre)}(z,w)
\\
&=\partial_z
\bigg[ e^{-2za} \sum_{k=0}^{N-1}\frac{2^{2k+5}q_{2k}^{(\rm pre)}(w)}{(2k+2)!f_{k+1}(a^2)} \Bigl( 4f_k(a^2)\frac{(2k+2)!e_{2k+2}(2za)}{(2a)^{2k+3}} + (2k+2)f_{k+1}(a^2)z^{2k+1} - 2f_k(a^2)z^{2k+3} \Bigr) \bigg] 
\end{split}
\end{align*}
and
\begin{align*}
\begin{split}
&\quad \frac{8\,  e^{2wa} }{(z-a)^3(w-a)^3} \mathfrak{D}_{z,a}\widehat{G}_N^{(\rm pre)}(w,z)
\\
&= \partial_z \bigg[
e^{-2za}
\sum_{k=0}^{N-1}\frac{2^{2k+4}q_{2k+1}^{(\rm pre)}(w)}{(2k+2)!f_{k+1}(a^2)}
\Bigl( 8\frac{(2k+2)!ae_{k+1}(a^2)}{(2a)^{2k+3}}e_{2k+2}(2za) - 4f_{k+1}(a^2)z^{2k+2}
\Bigr) \bigg].
    \end{split}
\end{align*}
\end{lem}

\begin{proof}
First, note that by \eqref{ExpLk2} and \eqref{Skeq Odd Hat}, 
we have
\begin{equation}
\label{q odd integrate}
    a(2k+1)!!\frac{\widehat{\bfkappa}_{k+1}^{(\rm g)}(z,a)}{z-a}
    =  \int_a^z\Big(
\frac{\widehat{q}_{2k+1}^{\,(\rm pre)}(t)}{(t-a)^2} - 2^{k+2}f_k(a^2)t^{2k+2}e^{-2ta}
    \Big)\,dt.
\end{equation}
By differentiating \eqref{Skeq Odd Hat} with respect to $z$ and by some computations using \eqref{q odd integrate}, 
we have
\begin{align}
\begin{split}
\label{q odd identity 1}
   \partial_z\widehat{q}_{2k+1}^{\,(\rm pre)}(z)
   &=
   2(z-a)\widehat{q}_{2k+1}^{\,(\rm pre)}(z)
   +
   4(z-a)^2\int_a^z\Bigl( \frac{\widehat{q}_{2k+1}^{\,(\rm pre)}(t)}{(t-a)^2} -
2^{k+2}f_k(a^2)t^{2k+2}e^{-2ta} \Bigr)\,dt
    \\
    &\quad
    +(z-a)^2\Big( 2^{k+3}(k+1)f_{k+1}(a^2)z^{2k+1}e^{-2za} - 2^{k+3}f_k(a^2)z^{2k+3}e^{-2za}  \Big).
\end{split}    
\end{align}
Note also that 
\begin{equation}
\label{epoly identity}
z^{2k+2}e^{-2za}
=
-\frac{(2k+2)!}{(2a)^{2k+3}}
\partial_z\bigl(e^{-2za}e_{2k+2}(2za)\bigr). 
\end{equation}
By \eqref{Skeq Odd Hat}, \eqref{q odd identity 1}, and \eqref{epoly identity}, we have 
\begin{align}
\begin{split}
\label{Tilde D Odd q}
&\quad
\mathfrak{D}_{z,a}\widehat{q}_{2k+1}^{\,(\rm pre)}(z)
\\
&=
(z-a)^3\partial_z
\Bigl(
2^{k+4}f_k(a^2)\frac{(2k+2)!}{(2a)^{2k+3}}e_{2k+2}(2za)
+2^{k+3}(k+1)f_{k+1}(a^2)z^{2k+1}
-2^{k+3}f_k(a^2)z^{2k+3}
\Bigr)e^{-2za}.
\end{split}
\end{align}
Hence, we obtain the first assertion of Lemma~\ref{Lem_Part1}. 

Next, we show the second assertion of the lemma. 
First, note that by \eqref{Skeq Even Hat} and \eqref{ExpLk2}, 
we have 
\begin{equation}
\label{q even integrate 1}
       (2k+3)!!\frac{\widehat{\bfkappa}_{k+2}^{(\rm g)}(z,a)}{z-a}
        -
        2a^2(2k+1)!!\frac{\widehat{\bfkappa}_{k+1}^{(\rm g)}(z,a)}{z-a}
        =        \int_a^z
        \Big( \frac{\widehat{q}_{2k}^{\,(\rm pre)}(t)}{(t-a)^2}-2^{k+3}a e_{k+1}(a^2)t^{2k+2}e^{-2ta}
        \Big)\, dt.
\end{equation}
By differentiating \eqref{Skeq Even Hat} with respect to $z$ and by some computations using \eqref{q even integrate 1}, 
we have 
\begin{align}
    \begin{split}\label{q even hat identity 1}
       \partial_z\widehat{q}_{2k}^{\,(\rm pre)}(z)
       &=2(z-a)\widehat{q}_{2k}^{\,(\rm pre)}(z)
       +
       4(z-a)^2\int_a^{z}\Big(
\frac{\widehat{q}_{2k}^{\,(\rm pre)}(t)}{(t-a)^2}
       -
       2^{k+3}ae_{k+1}(a^2)t^{2k+2}e^{-2ta}
       \Bigr)\, dt
       \\
       &\quad
       -2^{k+4}(z-a)^2f_{k+1}(a^2)z^{2k+2}e^{-2za}.
    \end{split}
\end{align}
By \eqref{Skeq Even Hat}, \eqref{q even hat identity 1}, and \eqref{epoly identity}, we have 
\begin{equation}
\label{Tilde D Even q}
\mathfrak{D}_{z,a}\widehat{q}_{2k}^{\,(\rm pre)}(z)
    =
  (z-a)^3\partial_z
  \Bigl(
2^{k+5}ae_{k+1}(a^2)\frac{(2k+2)!}{(2a)^{2k+3}}e_{2k+2}(2za)
-
2^{k+4}f_{k+1}(a^2)z^{2k+2}
  \Bigr)e^{-2za}. 
\end{equation}
This completes the proof. 
\end{proof}

We now prove Theorem~\ref{Thm_ODE_PreKernel}. 

\begin{proof}[Proof of Theorem~\ref{Thm_ODE_PreKernel}]

We claim that 
\begin{align}
    \begin{split} \label{Claim Lem Part 3}
\frac{e^{2(z+w)a}}{(z-a)^3} \mathfrak{D}_{z,a}\widehat{\bfkappa}_N^{(\rm pre)}(z,w)
&=  4(w-a)^3e^{2zw}Q(2N,2zw) 
\\
&\quad - \frac{2^{2N+1}a(w-a)^2(zw)^{2N}}{(2N)!f_N(a^2)}
\Big((2N+1-2za)e_N(a^2)+2f_N(a^2)\Big)
\\
&\quad -2^{N}
\frac{(2N+3)!!\cL_{N+1}(w,a)-2a^2(2N+1)!!\cL_{N}(w,a)}{2(2N)!f_N(a^2)}z^{2N}
\\
&\quad
+2^{N}a \frac{(2N+1)!!\cL_{N+1}(w,a)-2a^2(2N-1)!!\cL_{N}(w,a)}{(2N)!f_N(a^2)}z^{2N+1}. 
\end{split}
\end{align}
Then Theorem~\ref{Thm_ODE_PreKernel} immediately follows after straightforward transformations.   

For the proof of \eqref{Claim Lem Part 3}, we use Lemma~\ref{Lem_Part1} and consider the decomposition  
\begin{align}  \label{Claim Lem Part 30}
\begin{split}
\mathfrak{D}_{z,a}\widehat{\bfkappa}_N^{(\rm pre)}(z,w)
&= \mathfrak{D}_{z,a}\widehat{G}_N^{(\rm pre)}(z,w) - \mathfrak{D}_{z,a}\widehat{G}_N^{(\rm pre)}(w,z) 
\\
&= \frac{1}{8}(z-a)^3e^{-2wa}
\Big[  \mathrm{I}+\mathrm{II}+\mathrm{III} - \bigl(\mathrm{IV}+\mathrm{V}-\mathrm{VI}\bigr) \Big],
\end{split}
\end{align}
where 
\begin{align*}
\mathrm{I}
&:= \partial_z \bigg[ \sum_{k=0}^{N-1}\frac{2^{2k+5}ae_{k+1}(a^2)(w-a)^2w^{2k+2}}{(2k+1)!f_{k+1}(a^2)f_{k}(a^2)} \Bigl( f_{k+1}(a^2)-\frac{f_k(a^2)}{k+1}z^{2}\Bigr)z^{2k+1}e^{-2za}  \bigg],
\\
\mathrm{II}
&:= \partial_z \bigg[ \sum_{k=0}^{N-1}\frac{2^{k+4} ((2k+3)!!\cL_{k+2}(w,a)-2a^2(2k+1)!!\cL_{k+1}(w,a)) }{ (2a)^{2k+3}\, f_{k+1}(a^2) } e_{2k+2}(2az)e^{-2az} \bigg],
\\
\mathrm{III}
&:=\partial_z\bigg[\sum_{k=0}^{N-1}\frac{2^{k+2} ((2k+3)!!\cL_{k+2}(w,a)-2a^2(2k+1)!!\cL_{k+1}(w,a) )}{(2k+1)!f_{k+1}(a^2)f_{k}(a^2)}
\Bigl(f_{k+1}(a^2)z^{2k+1}-\frac{f_k(a^2)}{k+1}z^{2k+3}
\Bigr)e^{-2za} \bigg],
\\
\mathrm{IV} &:= \partial_z \bigg[ \sum_{k=0}^{N-1}
\frac{2^{2k+6}(w-a)^2}{(2k+2)!}(zw)^{2k+2}e^{-2za}
\bigg],
\\
\mathrm{V}
&:= \partial_z \bigg[ \sum_{k=0}^{N-1}\frac{ e_{k+1}(a^2)(2k+1)!!\cL_{k+1}(w,a)}{2^{k-2}\,a^{2k+1}\, f_{k+1}(a^2)f_k(a^2)}
 e_{2k+2}(2za)e^{-2za} \bigg],
\\
\mathrm{VI} &:= \partial_z \bigg[ \sum_{k=0}^{N-1}\frac{2^{k+4}a\,f_{k+1}(a^2)(2k+1)!!\cL_{k+1}(w,a)}{(2k+2)!f_{k+1}(a^2)f_k(a^2)} z^{2k+2}e^{-2za} \bigg].
\end{align*}
By differentiating the expression for $\mathrm{I}$, we have 
\begin{align*}
e^{2za}\mathrm{I}
&= a(w-a)^2\sum_{k=0}^{N-1}\frac{2^{2k+5}z^{2k}w^{2k+2}}{(2k)!f_k(a^2)}e_{k+1}(a^2)
-a^2(w-a)^2\sum_{k=0}^{N-1}\frac{2^{2k+6}z^{2k+1}w^{2k+2}}{(2k+1)!}e_{k+1}(a^2)
+2^4a(w-a)^2
\\
& \quad -a(w-a)^2\sum_{k=0}^{N-1}\frac{2^{2k+4}(2k+1)z^{2k}w^{2k}}{(2k)!f_k(a^2)}e_k(a^2)
+a^2(w-a)^2\sum_{k=0}^{N-1}\frac{2^{2k+5}z^{2k+1}w^{2k}}{(2k)!f_k(a^2)}e_k(a^2)
\\
& \quad -a^2(w-a)^22^5z -a\frac{(w-a)^22^{2N+4}(2N+1)e_N(a^2)z^{2N}w^{2N}}{(2N)!f_N(a^2)} +a^2\frac{(w-a)^22^{2N+5}e_N(a^2)z^{2N+1}w^{2N}}{(2N)!f_N(a^2)}.
\end{align*}
Similarly, we have 
\begin{align*}
e^{2za}\mathrm{II}
&= a^2\sum_{k=0}^{N-1}\frac{2^{k+5}\cL_{k+1}(w,a)}{(2k+2)!!f_{k+1}(a^2)}z^{2k+2}
-\sum_{k=0}^{N-1}\frac{2^{k+3}(2k+1)\cL_{k+1}(w,a)}{(2k)!!f_{k}(a^2)}z^{2k}
\\
& \quad -\frac{2^{N+3}(2N+1)\cL_{N+1}(w,a)}{(2N)!!f_{N}(a^2)}z^{2N}
-2^{5}a(w-a)^2
\end{align*}
and
\begin{align*}
e^{2za} \mathrm{III}
&=
\mathrm{III}'
-\frac{2^{N+2}(2N+1) ((2N+1)!!\cL_{N+1}(w,a)-2a^2(2N-1)!!\cL_N(w,a))}{(2N)!f_N(a^2)}z^{2N}
\\
&\quad
+a\frac{2^{N+3}(2N+1)\cL_{N+1}(w,a)}{(2N)!!f_N(a^2)}z^{2N+1}
-2a^3\frac{2^{N+3}\cL_{N}(w,a)}{(2N)!!f_N(a^2)}z^{2N+1}
+2^5a^2(w-a)^2z
-2^4a(w-a)^2,
\end{align*}
where 
\begin{align*}
\mathrm{III}'&:=
\sum_{k=0}^{N-1}\frac{2^{k+2}(2k+1)((2k+3)\cL_{k+2}(w,a)-(2k+1)\cL_{k+1}(w,a) )}{(2k)!!f_k(a^2)}z^{2k}
\\
&
\quad
-a\sum_{k=0}^{N-1}\frac{2^{k+3} ((2k+3)\cL_{k+2}(w,a)-(2k+1)\cL_{k+1}(w,a))}{(2k)!!f_k(a^2)}z^{2k+1}
\\
&
\quad
-a^2\sum_{k=0}^{N-1}\frac{2^{k+3}(2k+1)\bigl(\cL_{k+1}(w,a)-\cL_{k}(w,a)\bigr)}{(2k)!!f_k(a^2)}z^{2k}
+a^3\sum_{k=0}^{N-1}\frac{2^{k+4}(\cL_{k+1}(w,a)-\cL_{k}(w,a))}{(2k)!!f_k(a^2)}z^{2k+1}. 
\end{align*}
Note the following two identities: 
\begin{align*}
\cL_{k+1}(w,a)-\cL_k(w,a)
&=
-\bigl(2(w-a)^2-1\bigr)\frac{2^{k+1}a^{2k+1}e_{k}(w^2)}{(2k+1)!!}
-a(w-a)\frac{2^{k+2}e_{k-1}(a^2)}{(2k+1)!!}w^{2k+1}
\\
&\quad
-\frac{2^{k+1}e_k(a^2)}{(2k+1)!!}w^{2k+1}
+(w-a)\frac{2^{k+1}e_k(a^2)}{(2k-1)!!}w^{2k},
\end{align*}
and
\begin{align*}
&\quad
(2k+3)\cL_{k+2}(w,a)-(2k+1)\cL_{k+1}(w,a)
\\
&=
2\cL_{k+1}(w,a)
-(2(w-a)^2-1)\frac{2^{k+2}a^{2k+3}e_{k+1}(w^2)}{(2k+1)!!}
-a(w-a)\frac{2^{k+3}e_k(a^2)w^{2k+3}}{(2k+1)!!}
\\
&\quad
-\frac{2^{k+2}e_{k+1}(a^2)}{(2k+1)!!}w^{2k+3}
+(w-a)(2k+3)\frac{2^{k+2}e_{k+1}(a^2)}{(2k+1)!!}w^{2k+2}.
\end{align*}
Using these identities and after long but straightforward computations, it follows that 
\begin{align*}
\mathrm{III}'&=
\sum_{k=0}^{N-1}\frac{2^{k+3}(2k+1)\cL_{k+1}(w,a)}{(2k)!!f_k(a^2)}z^{2k}
-a\sum_{k=0}^{N-1}\frac{2^{k+4}\cL_{k+1}(w,a)}{(2k)!!f_k(a^2)}z^{2k+1}
\\
&
\quad
-a(w-a)^2\sum_{k=0}^{N-1}\frac{2^{2k+5}z^{2k}w^{2k+2}}{(2k)!f_k(a^2)}e_{k+1}(a^2)
+a(w-a)^2\sum_{k=0}^{N-1}\frac{2^{2k+4}z^{2k}(2k+1)w^{2k}}{(2k)!f_k(a^2)}e_k(a^2)
\\
&
\quad
+w(w-a)^2 2^5\sum_{k=0}^{N-1}\frac{(2zw)^{2k}}{(2k)!}
+a^2(w-a)^2\sum_{k=0}^{N-1}\frac{2^{2k+6}z^{2k+1}w^{2k+2}}{(2k+1)!f_k(a^2)}e_{k+1}(a^2)
\\
&\quad
-a^2(w-a)^2\sum_{k=0}^{N-1}\frac{2^{2k+5}z^{2k+1}w^{2k}}{(2k)!f_k(a^2)}e_k(a^2)
-a(w-a)^22^5\sum_{k=0}^{N-1}\frac{(2zw)^{2k+1}}{(2k+1)!}.
\end{align*}   
Similarly, we have
\begin{align*}
e^{2za}\bigl(\mathrm{IV}-\mathrm{V}-\mathrm{VI}\bigr)
&=
a^2\sum_{k=0}^{N-1}\frac{2^{k+5}\cL_{k+1}(w,a)}{(2k+2)!!f_{k+1}(a^2)}z^{2k+2}
-a\sum_{k=0}^{N-1}\frac{2^{k+4}\cL_{k+1}(w,a)}{(2k)!!f_k(a^2)}z^{2k+1}
\\
&\quad
-w(w-a)^22^5\sum_{k=0}^{N-1}\frac{(2zw)^{2k+1}}{(2k+1)!}
+a(w-a)^22^5\sum_{k=0}^{N-1}\frac{(2zw)^{2k+2}}{(2k+2)!}.
\end{align*}
By combining all of the above identities with \eqref{Claim Lem Part 30}, we obtain the desired equation \eqref{Claim Lem Part 3}. 
This completes the proof. 
\end{proof}

\section{Large-$N$ asymptotic analysis}    
\label{S_Scling limits}

In this section, we prove Theorems~\ref{Thm_CEO} and ~\ref{Thm_scaling limits}. 

\subsection{Bulk and edge scaling limits for the mean diagonal overlap}\label{Subsection_CE scaling limits}

\begin{proof}[Proof of Theorem~\ref{Thm_CEO}]
By \eqref{Conditional Expectation Diagonal Overlap Real}, for $a\in\R$, one can write
\begin{equation}
\label{widehat D11 N1a}
    \widehat{D}_{1,1}^{(N,1)}(a)
= \bigl(\partial_x\widehat{\bfkappa}_{N}^{(\rm g)}(x,a)\bigr|_{x=a} \bigr)^{-1} \frac{N!2^{N-2}}{(2N)!} \lim_{u\to a}
\frac{\widehat{q}_{2N-2}^{\,(\rm pre)}(u)}{(u-a)^3}.
\end{equation}
For $p \in [-1,1]$, we set 
$$
z=\sqrt{N}p+\zeta, \qquad u=\sqrt{N}p+\xi,\qquad a=\sqrt{N}p+\chi, 
$$
where $\zeta\in\C$ and $\xi,\chi\in\R$. 

We first show the bulk case in Theorem~\ref{Thm_CEO}, i.e. $p\in (-1,1).$
By using \eqref{ExpLk0} and \cite[Eq.(2.6)]{ABK22}, there exists a small positive $\epsilon>0$ such that for a sufficiently large $N$,  
\begin{equation}
\label{Hat L function Bulk}
\widehat{\cL}_N(z,a) = \widehat{\cL}_{\rm b}(\zeta,\chi)(1+O(e^{-\epsilon N})),
\end{equation}
where 
\begin{equation}
\label{Hat cL bulk}
\widehat{\cL}_{\rm b}(\zeta,\chi) = (2(\zeta-\chi)^2-1)\sqrt{\pi}e^{(\zeta-\chi)^2}\erf(\zeta-\chi) + 2(\zeta-\chi). 
\end{equation}
On the other hand, by \eqref{fsum}, we have 
\begin{equation}
\label{exp fsum bulk}
    e^{-a^2}f_N(a^2)=N(1-p^2)(1+O(e^{-\epsilon N})),
\end{equation}
see e.g. \cite[Proposition 2.4]{BF24}.
By combining \eqref{Hat L function Bulk} with \eqref{exp fsum bulk} and Stirling's formula, we have 
\begin{equation}
    \label{Difference Lfunction}
(2N+3)!!\widehat{\cL}_{N+2}(z,a)
-
2a^2(2N+1)!!\widehat{\cL}_{N+1}(z,a)
=
(1-p^2)
\widehat{\cL}_{\rm b}(\zeta,\chi)2^{N+5/2}N^{N+2}e^{-N}(1+O(N^{-1/2})).
\end{equation}
Similarly, by combining \eqref{qeven3} with \eqref{Hat L function Bulk}, \eqref{exp fsum bulk}, \eqref{Difference Lfunction}, and Stirling's formula, we have 
\begin{equation}
\label{Tilde q even bulk}
    \widehat{q}_{2N}^{\,(\rm pre)}(z)
    =
    (1-p^2)
    2^{N+5/2}N^{N+2}e^{-N}
\widehat{\cL}_{\rm b}(\zeta,\chi)(1+O(N^{-1/2})).
\end{equation}
It also follows from \cite[Proposition 3.1]{BE23} that 
\[
\bigl(\partial_x
\widehat{\bfkappa}_{N}^{(\rm g)}(x,a)\bigr|_{x=a} \bigr)^{-1}
=
\frac{1}{2}+O(e^{-\epsilon N}). 
\]
Therefore, by \eqref{widehat D11 N1a}, \eqref{exp fsum bulk}, \eqref{Tilde q even bulk}, Stirling's formula, and Taylor expansion for \eqref{Hat cL bulk}, as $N\to\infty$, 
we have 
\[
\bigl(\partial_x\widehat{\bfkappa}_{N}^{(\rm g)}(x,a)\bigr|_{x=a} \bigr)^{-1}e^{-2a^2}
\frac{N Z_{N-1}^{(\rm over)}(a)}{Z_{N}^{(\rm g)}}
=
\frac{2}{3}N(1-p^2)(1+o(1)),
\]
uniformly for $\chi$ in a compact subset of $\R$. 
This completes the proof of the bulk case.

\medskip 

Next, we complete the proof of the edge case in Theorem~\ref{Thm_CEO}. Due to the symmetry, it suffices to consider the case $p=1$.
Recall that by \cite[Eq.(8.8.9)]{NIST}, the incomplete gamma function satisfies the asymptotic behaviour 
\begin{equation}
\label{Incomplete Gamma Asymptotics}
Q(s+1,s+\sqrt{s}z)
=
\frac{1}{2}\erfc\bigl(\frac{z}{\sqrt{2}}\bigr)
+
\frac{e^{-\frac{z^2}{2}}}{\sqrt{2\pi}}
\frac{z^2+2}{3}
\frac{1}{\sqrt{s}}+O\Bigl(\frac{1}{s}\Bigr), \qquad
s\to\infty,
\end{equation}
uniformly for $z$ in a compact subset of $\C$.
From \cite[Proof of Theorem 2.1]{ABK22} or \cite[Proposition 3.1]{BE23}, as $N\to\infty$, we have 
\begin{equation}
\label{GinSE Edge 1}
\widehat{\bfkappa}_N^{(\rm g)}(z,a)
=
\widehat{\kappa}_{\rm e}^{(\rm g)}(\zeta,\chi)+O(N^{-1/2}),
\end{equation}
where 
\begin{align}
 \begin{split}
\label{GinSE Edge 2}
  \widehat{\kappa}_{\rm e}^{(\rm g)}(\zeta,\chi)   = \frac{e^{(\zeta-\chi)^2}}{\sqrt{2}} \int_{-\infty}^{0}e^{-2(\zeta-u)^2}\erfc(\sqrt{2}(\chi-u))
-e^{-2(\chi-u)^2}\erfc(\sqrt{2}(\zeta-u)) \, du.
   \end{split}
\end{align}
Note here that \eqref{GinSE Edge 2} satisfies 
\begin{equation}
\label{GinSE Edge ODE}
\partial_{\zeta}\widehat{\kappa}_{\rm e}^{(\rm g)}(\zeta,\chi)
=
2(\zeta-\chi)\widehat{\kappa}_{\rm e}^{(\rm g)}(\zeta,\chi)
+
\erfc(\zeta+\chi)-\frac{e^{(\zeta-\chi)^2-2\zeta^2}}{\sqrt{2}}\erfc(\sqrt{2}\chi).
\end{equation}

We write 
\begin{equation}
\label{def of difference between LN+1 LN}
(2N+1)\widehat{\cL}_{N+1}(w,a)-2a^2\widehat{\cL}_{N}(w,a)
=
\widehat{\mathcal{L}}_N^{\,(1)}(w,a)
+
\widehat{\mathcal{L}}_N^{\,(2)}(w,a),
\end{equation}
where 
\begin{align}
\label{def of LN1}
\widehat{\mathcal{L}}_N^{\,(1)}(w,a)&:=
2N(\widehat{\cL}_{N+1}(w,a)-\widehat{\cL}_{N}(w,a))
-4\chi\sqrt{N}\widehat{\cL}_{N}(w,a),
\\
\label{def of LN2}
\widehat{\mathcal{L}}_N^{\,(2)}(w,a)&:=
\widehat{\cL}_{N+1}(w,a)-2\chi^2\widehat{\cL}_{N}(w,a). 
\end{align}
By \eqref{ExpLk2} and \eqref{GinSE Edge 1}, as $N\to\infty$, we have
\[
\widehat{\cL}_N(w,a)
=
(\eta-\chi)^2\partial_{\eta} \Big[\frac{\widehat{\kappa}_{\rm e}^{(\rm g)}(\eta,\chi)}{\eta-\chi}\Big]+O(N^{-1/2}).
\]
Furthermore, it follows from \eqref{GinSEKernel} and \eqref{ExpLk2} that
\begin{align}
\begin{split}
&\quad 2N(\widehat{\cL}_{N+1}(w,a)-\widehat{\cL}_{N}(w,a))
=
 2N(w-a)^2\partial_w \frac{\widehat{\bfkappa}_{N+1}^{(\rm g)}(w,a)-\widehat{\bfkappa}_{N}^{(\rm g)}(w,a)}{w-a}
\\
&=
\sqrt{N}
(\eta-\chi)^2
\partial_{\eta} \frac{e^{(\eta-\chi)^2}(e^{-2\eta^2}\erfc(\sqrt{2}\chi)-e^{-2\chi^2}\erfc(\sqrt{2}\eta))}{\sqrt{2}(\eta-\chi)}
(1+O(N^{-1/2})). 
\end{split}
\end{align}
Therefore, by \eqref{def of difference between LN+1 LN}, we have 
\begin{align}
    \begin{split}
\label{Aymptotic Edge 1}
&\quad
2^{-N} N^{-(N+1/2)}e^{N}(2N-1)!!\bigl( (2N+1)\widehat{\cL}_{N+1}(z,a)-2a^2\widehat{\cL}_N(z,a)\bigr)
\\
&=
(\zeta-\chi)^2\partial_{\zeta}\Bigl[
\frac{e^{(\zeta-\chi)^2}\bigl(e^{-2\zeta^2}\erfc(\sqrt{2}\chi)-e^{-2\chi^2}\erfc(\sqrt{2}\zeta)\bigr)-4\sqrt{2}\chi\widehat{\kappa}_{\rm e}^{(\rm g)}(\zeta,\chi)}{\zeta-\chi}
\Bigr]
(1+O(N^{-1/2})),
    \end{split}
\end{align}
as $N\to\infty$.

On the other hand, by \eqref{Incomplete Gamma Asymptotics}, we have
\begin{equation}
    \label{Aymptotic Edge 2}
    2^{N+2}ae^{a^2}Q(N+1,a^2)z^{2N}e^{-2za}
    =
    2^{N+1}N^{N+1/2}e^{-N}e^{-2\zeta^2+(\zeta-\chi)^2}\erfc(\sqrt{2}\chi)
    (1+O(N^{-1/2})).
\end{equation}
Let us choose $u=\sqrt{N}+\xi$ for $\xi\in\R$ so that $u\to a$, when $\xi\to\chi$.
Combining \eqref{Aymptotic Edge 1} and \eqref{Aymptotic Edge 2} with \eqref{Skeq Even Hat}, we have
\begin{align}
\begin{split}
\label{Hat q even edge limit}
\widehat{q}_{2N-2}^{\,(\rm pre)}(u)
\sim
2^{N}
N^{N+1/2}e^{-N}
\frac{4}{3}\Bigl(
\sqrt{\frac{2}{\pi}}e^{-4\chi^2}
+
2\chi e^{-2\chi^2}\erfc(\sqrt{2}\chi)-4\sqrt{2}\chi \erfc(2\chi)
\Bigr)
(\xi-\chi)^3.
\end{split}    
\end{align}
Hence, combining \eqref{widehat D11 N1a}, \eqref{GinSE Edge ODE}, \eqref{Aymptotic Edge 1}, \eqref{Aymptotic Edge 2}, \eqref{Hat q even edge limit}, together with Stirling's formula, the proof is complete. 
\end{proof}

\subsection{Bulk and edge scaling limits for the eigenvalue correlation functions of the overlap weight}

In this subsection, we prove Theorem~\ref{Thm_scaling limits}. 
According to \eqref{def of skew kernel overlap} and \eqref{WbfK}, we first define 
\begin{equation}
\label{wbfK over}
\widetilde{\bfkappa}_N^{\,(\rm over)}(z,w)
:=
e^{2a^2-2za-2wa}
(z-a)^4(w-a)^4
\bfkappa_N^{(\rm over)}(z,w). 
\end{equation}
Then by \eqref{def of Pfaff Struc}, we have 
\begin{align}
\begin{split}
\label{finite N overlap k point correlation function}
&\quad
\bfR_{N,k}^{(\rm over)}(z_1,z_2,\dots,z_k) = \Pf \bigg[\begin{pmatrix}
\bfkappa_{N}^{(\rm over)}(z_j,z_{\ell}) & \bfkappa_{N}^{(\rm over)}(z_j,\zbar_{\ell})  \\ 
\bfkappa_{N}^{(\rm over)}(\zbar_j,z_{\ell}) & \bfkappa_{N}^{(\rm over)}(\zbar_j,\zbar_{\ell})
\end{pmatrix}
\bigg]_{j,\ell=1}^k
\prod_{j=1}^k (\zbar_j-z_j) 
\omega^{({\rm over})}(z_j)
\\
&= \Pf \bigg[\begin{pmatrix}
\dfrac{\phi_a(z_j)\widetilde{\bfkappa}_{N}^{(\rm over)}(z_j,z_{\ell})\phi_a(z_{\ell})}{((\zeta_j-\chi)(\zeta_{\ell}-\chi))^4} &
\dfrac{\phi_a(z_j)\widetilde{\bfkappa}_{N}^{(\rm over)}(z_j,\zbar_{\ell})\phi_a^{-1}(z_{\ell})}{((\zeta_j-\chi)(\overline{\zeta}_{\ell}-\chi))^4}  
\smallskip 
\\ 
\dfrac{\phi_a^{-1}(z_j)\widetilde{\bfkappa}_{N}^{(\rm over)}(\zbar_j,z_{\ell})\phi_a(z_\ell)}{((\overline{\zeta}_j-\chi)(\zeta_{\ell}-\chi))^4} & 
\dfrac{\phi_a^{-1}(z_j)\widetilde{\bfkappa}_{N}^{(\rm over)}(\zbar_j,\zbar_{\ell})\phi_a^{-1}(z_{\ell})}{((\overline{\zeta}_j-\chi)(\overline{\zeta}_{\ell}-\chi))^4}
\end{pmatrix}
\bigg]_{j,\ell=1}^k 
\prod_{j=1}^k (\overline{\zeta}_j-\zeta_j) 
\omega_s(\zeta_j),
\end{split}
\end{align}
where $\omega_s(\zeta)$ is given by \eqref{Omega_B}, and $\phi_a(z)=e^{2ia\im(z)}$ for $z\in\C$ and $a\in\R$. 
For $j=1,\dots,k$, let 
$$
z_j=\sqrt{N}p+\zeta_j, \qquad p \in [-1,1]. 
$$
We also write 
\begin{equation}
z=\sqrt{N}p+\zeta, \qquad w=\sqrt{N}p+\eta, \qquad a=\sqrt{N}p+\chi
\end{equation}
for $\zeta,\eta \in \C$ and $\chi\in\R$. 
We first prove the bulk case of Theorem~\ref{Thm_scaling limits}.

\begin{proof}[Proof of Theorem~\ref{Thm_scaling limits}: bulk case]
By \eqref{SkewKernelOverlap}, let us denote 
\begin{equation}
    \label{Limit pre kernel bulk}
    \widetilde{\kappa}_{\rm b}^{(\rm pre)}(\zeta,\eta)
:=\lim_{N\to\infty}\widetilde{\bfkappa}_N^{(\rm pre)}(z,w).
\end{equation}
Then, by Theorem~\ref{Thm_ODE_PreKernel}, \eqref{Hat L function Bulk} and Stirling's formula, we have 
\begin{equation}
    \label{Bulk ODE pre}
\Bigl[\partial_\zeta^2-2\frac{(\zeta-\chi)^2+1}{\zeta-\chi}\partial_{\zeta}-2\Bigr]
\widetilde{\kappa}_{\rm b}^{(\rm pre)}(\zeta,\eta|\chi) = 4(\zeta-\chi)^2(\eta-\chi)^3e^{2(\zeta-\chi)(\eta-\chi)}.
\end{equation}
Since the differential operator is translation invariant under a shift with respect to the horizontal direction along the real line, and the inhomogeneous term in \eqref{Bulk ODE pre} is same as \eqref{Origin_Limit_ODE} shifted by $\chi$, we find that the unique solution for \eqref{Bulk ODE pre} with the initial conditions 
$$
\widetilde{\kappa}_{\rm b}^{(\rm pre)}(\zeta,\zeta)=0, \qquad \partial_{\zeta}\widetilde{\kappa}_{\rm b}^{(\rm pre)}(\zeta,\eta)|_{\zeta=\chi}=0
$$ 
is given by
\begin{equation} \label{tilde pre kernel bulk}
\widetilde{\kappa}_{\rm b}^{(\rm pre)}(\zeta,\eta) =  (\zeta-\chi)^3 (\eta-\chi)^3\,
\varkappa_{\,\rm o}^{(\mathrm{pre})}(\zeta-\chi,\eta-\chi). 
\end{equation}
Here, $\varkappa_{\,\rm o}^{(\mathrm{pre})}(z,w)$ is given by \eqref{Limit Pre Bulk Origin}.

For $z,w\in\C$ and $u,a\in\R$, according to Proposition~\ref{Prop_Christoffel pertubation}, we define 
\begin{equation}
\label{u over pre kernel}
\bfkappa_{N-1,u}^{({\rm over})}(z,w)
:= \frac{1}{(z-a)(w-a)}
\bigg( 
\bfkappa_{N}^{({\rm pre})}(z,w)
-\bfkappa_{N}^{({\rm pre})}(z,u)\frac{q_{2N}^{({\rm pre})}(w)}{q_{2N}^{({\rm pre})}(u)}
+\bfkappa_{N}^{({\rm pre})}(w,u)\frac{q_{2N}^{({\rm pre})}(z)}{q_{2N}^{({\rm pre})}(u)} \bigg). 
\end{equation}
Let us denote 
\begin{equation}
\label{wtilde u kappa over}
    \widetilde{\bfkappa}_{N-1,u}^{(\rm over)}(z,w)
    :=
    ((z-a)(w-a))^4e^{2a^2-2za-2wa}\bfkappa_{N-1,u}^{({\rm over})}(z,w).
\end{equation}
Clearly, we have 
$$\widetilde{\bfkappa}_{N-1,u}^{(\rm over)}(z,w)\to \widetilde{\bfkappa}_{N-1}^{(\rm over)}(z,w), \qquad u \to a. $$   
Letting $u=\sqrt{N}p+\xi\to a=\sqrt{N}p+\chi$ with $p\in(-1,1)$ and $\xi,\chi\in\R$, it follows from \eqref{Tilde q even bulk}, \eqref{Limit pre kernel bulk}, and \eqref{tilde pre kernel bulk} that as $N \to \infty,$ 
\begin{equation}
\label{u limit bulk over}
\widetilde{\bfkappa}_{N-1,u}^{(\rm over)}(z,w)
\sim
\widetilde{\kappa}_{\rm b}^{(\rm pre)}(\zeta,\eta)
-
\widetilde{\kappa}_{\rm b}^{(\rm pre)}(\zeta,\xi)
\frac{\widehat{\cL}_{\rm b}(\eta,\chi)}{\widehat{\cL}_{\rm b}(\xi,\chi)}
+
\widetilde{\kappa}_{\rm b}^{(\rm pre)}(\eta,\xi)
\frac{\widehat{\cL}_{\rm b}(\zeta,\chi)}{\widehat{\cL}_{\rm b}(\xi,\chi)},
\end{equation}
uniformly for $\zeta,\eta$ in compact subsets of $\C$ and $\xi,\chi$ in a compact subset of $\R$. 
Note that by Taylor expansion, as $\xi\to\chi$, we have 
\[
\frac{\widetilde{\kappa}_{\rm b}^{(\rm pre)}(\zeta,\xi)}{\widehat{\cL}_{\rm b}(\xi,\chi)}
=
\frac{1}{4}\bigl((2(\zeta-\chi)^2-1)e^{(\zeta-\chi)^2}-(\zeta-\chi)^2+1 \bigr)
+O(\zeta-\chi).
\]
Therefore, combining \eqref{u limit bulk over} with the above and dividing it by $(\zeta-\chi)^4(\eta-\chi)^4$, 
we obtain \eqref{Thm Overlap Bulk Kernel}. 
This completes the proof of the bulk case.
\end{proof}

We now prove the edge case of Theorem~\ref{Thm_scaling limits}. For this purpose, we need some preparations. First, let us denote 
\begin{align}
\label{FFEdge}
F(\chi)&:=e^{-2\chi^2}-\sqrt{2\pi}\chi\erfc(\sqrt{2}\chi), 
\\
\label{cS1}
\mathcal{S}_1(\zeta,\chi)
&:=e^{(\zeta-\chi)^2}(2(\zeta-\chi)^2-1),
\\
\label{cS2}
\mathcal{S}_2(\zeta,\chi)
&:=\sqrt{\pi}e^{(\zeta-\chi)^2}\erf(\zeta-\chi)(2(\zeta-\chi)^2-1)+2(\zeta-\chi),
\\
\label{cW_wronskian}
\mathcal{W}(\zeta,\chi)
&:=-8e^{(\zeta-\chi)^2}(\zeta-\chi)^2,
\\
\label{frakI}
\mathcal{E}(\eta,\chi)
&:= \partial_{\eta}
\bigg[ \frac{e^{(\eta-\chi)^2}(e^{-2\eta^2}\erfc(\sqrt{2}\chi)-e^{-2\chi^2}\erfc(\sqrt{2}\eta))-4\sqrt{2}\chi \widehat{\kappa}_{\rm e}^{(\rm g)}(\eta,\chi)}{\eta-\chi} \bigg],
\\
    \begin{split}
    \label{inhomogeneous cF}
\mathcal{F}(\zeta,\eta|\chi)  
&:=
2(\zeta-\chi)^3(\eta-\chi)^3e^{2(\zeta-\chi)(\eta-\chi)}\erfc(\zeta+\eta)
\\
&\quad
-
\frac{2(\zeta-\chi)^3(\eta-\chi)^2e^{-2\zeta^2-2\eta^2+(\zeta-\chi)^2+(\eta-\chi)^2}}{\sqrt{\pi}}
\Bigl(
1
-
\sqrt{\frac{\pi}{2}}(\zeta+\chi)
\frac{\erfc(\sqrt{2}\chi)}{F(\chi)}
\Bigr)
\\
&\quad
-\frac{2(\zeta-\chi)^3(\eta-\chi)^2e^{-2\zeta^2+(\zeta-\chi)^2}}{F(\chi)}
\partial_{\eta}\Big[\frac{\widehat{\kappa}_{\rm e}^{(\rm g)}(\eta,\chi)}{\eta-\chi}\Big] 
\\
&\quad + \frac{(\zeta+\chi)(\zeta-\chi)^3(\eta-\chi)^2
e^{-2\zeta^2+(\zeta-\chi)^2}}{\sqrt{2}F(\chi)}
\mathcal{E}(\eta,\chi).
    \end{split}
\end{align}
These are building blocks to define 
\begin{align} 
\begin{split}
\label{cA_Edge}
\mathcal{A}(\zeta,\chi)
&:=
(\zeta-\chi)^2
\Bigl(
2e^{-2\zeta^2+(\zeta-\chi)^2}
\erfc(\sqrt{2}\chi)
+
\mathcal{E}(\zeta,\chi)
\Bigr)
,
\end{split}
\\
\label{cB_Edge}
\mathcal{B}(\chi)
&:=
\frac{4}{3}
\Bigl[
\sqrt{\frac{2}{\pi}}
e^{-4\chi^2}
+
2\chi e^{-2\chi^2}\erfc(\sqrt{2}\chi)
-
4\sqrt{2}\chi \erfc(2\chi)
\Bigr],
\\
\label{cC_Edge}
\mathcal{C}(\zeta,\chi)
&:=
\frac{8}{3\,\mathcal{S}_2(\zeta,\chi)}
\int_{\zeta}^{\chi}
\frac{
\mathcal{S}_1(\zeta,\chi)\mathcal{S}_2(t,\chi)
-\mathcal{S}_2(\zeta,\chi)\mathcal{S}_1(t,\chi)}{(t-\chi)\mathcal{W}(t,\chi)}\mathcal{F}(t,\zeta|\chi)\, dt,
\\
\begin{split}
\label{cK_Edge}
\mathcal{K}(\zeta,\eta|\chi)
&:=
\frac{\mathcal{S}_1(\zeta,\chi)\mathcal{S}_2(\eta,\chi)-\mathcal{S}_1(\eta,\chi)\mathcal{S}_2(\zeta,\chi)}{\mathcal{S}_2(\eta,\chi)}   \int_{\eta}^{\chi}
\frac{\mathcal{S}_2(t,\chi)\mathcal{F}(t,\eta|\chi)}{(t-\chi)\mathcal{W}(t,\chi)}\, dt   
\\
&\quad
+
\int_{\eta}^{\zeta}
\frac{\mathcal{S}_1(t,\chi)\mathcal{S}_2(\zeta,\chi)-\mathcal{S}_1(\zeta,\chi)\mathcal{S}_2(t,\chi)}{(t-\chi)\mathcal{W}(t,\chi)}
\mathcal{F}(t,\eta|\chi)\,dt.
    \end{split}
\end{align}

Indeed, \eqref{cS1} and \eqref{cS2} are fundamental solutions to the differential equation \eqref{Pre Limit Edge ODE 2} below, and \eqref{cW_wronskian} is the Wronskian of \eqref{cS1} and \eqref{cS2}.  

\begin{proof}[Proof of Theorem~\ref{Thm_scaling limits}: edge case]
Let $z=\sqrt{N}+\zeta,u=\sqrt{N}+\xi,a=\sqrt{N}+\chi$ with $u\to a$, i.e., as $\xi\to\chi$. 
As before, we first compute the asymptotic behaviour of $\mathrm{I}_N$, $\mathrm{II}_N$, $\mathrm{III}_N$ and $\mathrm{IV}_N$ given in \eqref{def of IN}, \eqref{def of IIN}, \eqref{def of IIIN}, and \eqref{def of IVN}, respectively. After that, we apply Proposition~\ref{Prop_Christoffel pertubation}. 

Recall that $p=1$. 
For \eqref{def of IN} and by using \eqref{Incomplete Gamma Asymptotics}, as $N\to\infty$, we have
\begin{equation}
  \label{Edge Asym 1}
  \mathrm{I}_N(z,w) = 2(\zeta-\chi)^3(\eta-\chi)^3e^{2(\zeta-\chi)(\eta-\chi)}\erfc(\zeta+\eta)(1+O(N^{-1/2})).
\end{equation}
Similarly, for \eqref{def of IIN} and by \eqref{Incomplete Gamma Asymptotics}, as $N\to\infty$, we have 
\begin{equation}
\mathrm{II}_N(z,w)
=
\frac{2(\zeta-\chi)^3(\eta-\chi)^2e^{-2\zeta^2-2\eta^2+(\zeta-\chi)^2+(\eta-\chi)^2}}{\sqrt{\pi}}
\Big( 1 - \sqrt{\frac{\pi}{2}}(\zeta+\chi)
\frac{\erfc(\sqrt{2}\chi)}{F(\chi)}
\Big)
(1+O(N^{-1/2})).
\end{equation}
By \eqref{def of IIIN}, \eqref{def of IVN} \eqref{def of difference between LN+1 LN}, \eqref{def of LN1}, \eqref{def of LN2}, and \eqref{Aymptotic Edge 1}, as $N\to\infty$, we have
\begin{align}
\begin{split}
&\quad
-\mathrm{III}_N(z,w)+\mathrm{IV}_N(z,w)
\\
&\sim
-\frac{2(\zeta-\chi)^3(\eta-\chi)^2e^{-2\zeta^2+(\zeta-\chi)^2}}{F(\chi)}
\partial_{\eta} \Big[ \frac{\widehat{\kappa}_{\rm e}^{(\rm g)}(\eta,\chi)}{\eta-\chi} \Big] 
+
\frac{(\zeta+\chi)(\zeta-\chi)^3(\eta-\chi)^2
e^{-2\zeta^2+(\zeta-\chi)^2}}{\sqrt{2}F(\chi)}
\mathcal{E}(\eta,\chi). 
\end{split}
\end{align}
Combining all of the above with \eqref{def of IN}, \eqref{def of IIN}, \eqref{def of IIIN}, and \eqref{def of IVN}, as $N\to\infty$, we have 
\begin{equation}
\label{IN INN IIIN IVN sim cFN}
\mathrm{I}_N(z,w)
-
\mathrm{II}_N(z,w)
-
\mathrm{III}_N(z,w)
+
\mathrm{IV}_N(z,w)
\sim
\mathcal{F}(\zeta,\eta|\chi),
\end{equation}
uniformly for $\zeta,\eta$ in compact subsets of $\C$ and $\chi$ in a compact subset of $\R$, where $\mathcal{F}(\zeta,\eta|\chi)$ is given by \eqref{inhomogeneous cF}.

Now, we are ready to complete the proof. 
By \eqref{u over pre kernel} and \eqref{wtilde u kappa over}, 
let us denote 
\begin{equation}
\widetilde{\kappa}_{\mathrm{e}}^{(\mathrm{pre})}(\zeta,\eta)
=
    \lim_{N\to\infty}   \widetilde{\bfkappa}_N^{(\mathrm{pre})}(z,w).
\end{equation}
Then, by Theorem~\ref{Thm_ODE_PreKernel} and \eqref{IN INN IIIN IVN sim cFN}, we have the following limiting differential equation:
\begin{equation}
\label{Pre Limit Edge ODE 2}
\Bigl[
\partial_{\zeta}^2-\frac{2((\zeta-\chi)^2+1)}{\zeta-\chi}\partial_{\zeta}-2
\Bigr]
\widetilde{\kappa}_{\rm e}^{(\rm pre)}(\zeta,\eta)
=
\frac{\mathcal{F}(\zeta,\eta|\chi)}{\zeta-\chi}.
\end{equation}
As previously mentioned, the fundamental solutions for the homogeneous part of \eqref{Pre Limit Edge ODE 2} are given by \eqref{cS1} and \eqref{cS2}, and their Wronskian is given by \eqref{cW_wronskian}. 
Therefore, the general solution to \eqref{Pre Limit Edge ODE 2} is given by 
\begin{equation}
\label{widetilde kappa edge}
\widetilde{\kappa}_{\rm e}^{(\rm pre)}(\zeta,\eta|\chi)
=
c_1\mathcal{S}_1(\zeta,\chi)
+
c_2\mathcal{S}_2(\zeta,\chi)
+
\mathcal{T}(\zeta,\eta|\chi), 
\end{equation}
where
\begin{equation}
\label{cT special solution}
\mathcal{T}(\zeta,\eta|\chi)
:=
    \int_{\eta}^{\zeta}
\frac{\mathcal{S}_1(t,\chi)\mathcal{S}_2(\zeta,\chi)-\mathcal{S}_1(\zeta,\chi)\mathcal{S}_2(t,\chi)}{(t-\chi)\mathcal{W}(t,\chi)}
\mathcal{F}(t,\eta|\chi)\, dt.
\end{equation}
We shall determine constants $c_1,c_2$, which may depend on $\eta,\chi$. 
Due to skew-symmetry of the kernel, we have 
\begin{equation}
\label{solution skew-symmetry}
c_1\mathcal{S}_1(\eta,\chi)+c_2\mathcal{S}_2(\eta,\chi)=0.
\end{equation}
By combining the asymptotic behaviour
\begin{equation}
\label{def of partial cS1cS2}
\partial_{\zeta}\mathcal{S}_1(\zeta,\chi)=2(\zeta-\chi)+6(\zeta-\chi)^3+O((\zeta-\chi)^4),\qquad
\partial_{\zeta}\mathcal{S}_2(\zeta,\chi)=8(\zeta-\chi)^2+O((\zeta-\chi)^4),
\end{equation}
as $\zeta \to \chi$, and 
\[
\partial_{\zeta}\mathcal{T}(\zeta,\eta|\chi)
= \int_{\eta}^{\zeta}
\frac{\partial_{\zeta}\mathcal{S}_2(\zeta,\chi)\mathcal{S}_1(t,\chi)}{(t-\chi)\mathcal{W}(t,\chi)}
\mathcal{F}(t,\eta|\chi) \, dt
-
\int_{\eta}^{\zeta}
\frac{\partial_{\zeta}\mathcal{S}_1(\zeta,\chi)\mathcal{S}_2(t,\chi)\mathcal{F}(t,\eta|\chi)}{(t-\chi)\mathcal{W}(t,\chi)} \, dt,
\]
with \eqref{solution skew-symmetry}, 
we have
\[
c_1=
\int_{\eta}^{\chi}
\frac{\mathcal{S}_2(t,\chi)\mathcal{F}(t,\eta|\chi)}{(t-\chi)\mathcal{W}(t,\chi)} \, dt,\qquad
c_2=
-\frac{\mathcal{S}_1(\eta,\chi)}{\mathcal{S}_2(\eta,\chi)}
\int_{\eta}^{\chi}
\frac{\mathcal{S}_2(t,\chi)\mathcal{F}(t,\eta|\chi)}{(t-\chi)\mathcal{W}(t,\chi)} \, dt.
\]
Therefore, we obtain
\begin{equation}
\label{pre kernel edge solution}
\widetilde{\kappa}_{\rm e}^{(\rm pre)}(\zeta,\eta)
=
\frac{\mathcal{S}_1(\zeta,\chi)\mathcal{S}_2(\eta,\chi)-\mathcal{S}_1(\eta,\chi)\mathcal{S}_2(\zeta,\chi)}{\mathcal{S}_2(\eta,\chi)}   \int_{\eta}^{\chi}
\frac{\mathcal{S}_2(t,\chi)\mathcal{F}(t,\eta|\chi)}{(t-\chi)\mathcal{W}(t,\chi)}dt     
+
\mathcal{T}(\zeta,\eta|\chi)
=
\mathcal{K}(\zeta,\eta|\chi),
\end{equation}
where $\mathcal{K}(\zeta,\eta|\chi)$ is given by \eqref{cK_Edge}. 
By skew-symmetry of the kernel \eqref{cK_Edge}, we have 
\begin{align*}
\mathcal{K}(\zeta,u|\chi)
&= - \mathcal{K}(u,\zeta|\chi) =
-
\frac{\mathcal{S}_1(u,\chi)\mathcal{S}_2(\zeta,\chi)-\mathcal{S}_1(\zeta,\chi)\mathcal{S}_2(u,\chi)}{\mathcal{S}_2(\zeta,\chi)}   \int_{\zeta}^{\chi}
\frac{\mathcal{S}_2(t,\chi)}{(t-\chi)\mathcal{W}(t,\chi)}\mathcal{F}(t,\zeta|\chi)dt     
-
\mathcal{T}(u,\zeta|\chi). 
\end{align*}
To compute the Taylor expansion of \eqref{cK_Edge} at $u=\chi$, note that 
\begin{align*}
\bigl(
\mathcal{S}_1(u,\chi)\mathcal{S}_2(\zeta,\chi)-\mathcal{S}_1(\zeta,\chi)\mathcal{S}_2(u,\chi)
\bigr)|_{u=\chi}
&=
-\mathcal{S}_2(\zeta,\chi),
\\
\partial_{u}\bigl(
\mathcal{S}_1(u,\chi)\mathcal{S}_2(\zeta,\chi)-\mathcal{S}_1(\zeta,\chi)\mathcal{S}_2(u,\chi)
\bigr)|_{u=\chi}
&=0,
\\
\partial_{u}^2\bigl(
\mathcal{S}_1(u,\chi)\mathcal{S}_2(\zeta,\chi)-\mathcal{S}_1(\zeta,\chi)\mathcal{S}_2(u,\chi)
\bigr)|_{u=\chi}
&=2\mathcal{S}_2(\zeta,\chi),
\\
\partial_{u}^3\bigl(
\mathcal{S}_1(u,\chi)\mathcal{S}_2(\zeta,\chi)-\mathcal{S}_1(\zeta,\chi)\mathcal{S}_2(u,\chi)
\bigr)|_{u=\chi}
&=-16\mathcal{S}_1(\zeta,\chi),
\end{align*}
and 
\begin{align*}
(\mathcal{S}_1(t,\chi)\mathcal{S}_2(u,\chi)-\mathcal{S}_1(u,\chi)\mathcal{S}_2(t,\chi))|_{t=u}
&=0,
\\
\partial_{u}(\mathcal{S}_1(t,\chi)\mathcal{S}_2(u,\chi)-\mathcal{S}_1(u,\chi)\mathcal{S}_2(t,\chi))|_{t=u}
&=-8e^{(u-\chi)^2}(u-\chi)^2
,
\\
\partial_{u}^2
(\mathcal{S}_1(t,\chi)\mathcal{S}_2(u,\chi)-\mathcal{S}_1(u,\chi)\mathcal{S}_2(t,\chi))|_{t=u}
&=-16e^{(u-\chi)^2}(u-\chi)((u-\chi)^2+1).
\end{align*}
By Leibniz integral rule, we have 
\begin{align*}
\partial_u^2\mathcal{T}(u,\zeta|\chi) 
&=
\frac{\mathcal{F}(u,\zeta|\chi)}{u-\chi}
+
\int_{\zeta}^u
\frac{\partial_u^2\bigl(\mathcal{S}_1(t,\chi)\mathcal{S}_2(u,\chi)-\mathcal{S}_1(u,\chi)\mathcal{S}_2(t,\chi)\bigr)}{(t-\chi)\mathcal{W}(t,\chi)}
\mathcal{F}(t,\zeta|\chi) \, dt,
\\
\partial_u^3\mathcal{T}(u,\zeta|\chi) 
&=
\partial_u
\frac{\mathcal{F}(u,\zeta|\chi)}{u-\chi}
+
2\frac{(u-\chi)^2+1}{(u-\chi)^2}
\mathcal{F}(u,\zeta|\chi)
+
\int_{\zeta}^u
\frac{\partial_u^3\bigl(\mathcal{S}_1(t,\chi)\mathcal{S}_2(u,\chi)-\mathcal{S}_1(u,\chi)\mathcal{S}_2(t,\chi)\bigr)}{(t-\chi)\mathcal{W}(t,\chi)}
\mathcal{F}(t,\zeta|\chi) \, dt. 
\end{align*}
Combining the above with \eqref{inhomogeneous cF} and \eqref{def of partial cS1cS2}, we have 
\[
\partial_{u}^{j}
\mathcal{K}(\zeta,u|\chi)
|_{u=\chi}
=0,\qquad
j=0,1,2,
\]
and 
\[
\partial_{u}^3
\mathcal{K}(\zeta,u|\chi)
|_{u=\chi}
=
-\partial_{u}^3
\mathcal{K}(u,\zeta|\chi)
|_{u=\chi}
=
\frac{16}{\mathcal{S}_2(\zeta,\chi)}
\int_{\zeta}^{\chi}
\frac{
\mathcal{S}_1(\zeta,\chi)\mathcal{S}_2(t,\chi)
-\mathcal{S}_2(\zeta,\chi)\mathcal{S}_1(t,\chi)}{(t-\chi)\mathcal{W}(t,\chi)}\mathcal{F}(t,\zeta|\chi) \, dt.
\]
Therefore, as $u\to\chi$, we obtain
\[
\widetilde{\kappa}_{\mathrm{{e}}}^{(\mathrm{pre})}(\zeta,u)
=
\mathcal{C}(\zeta,\chi)
(u-\chi)^3
+O\bigl((u-\chi)^4 \bigr),
\]
where $\mathcal{C}(\zeta,\chi)$ is given by \eqref{cC_Edge}.
By \eqref{Skeq Even Hat}, \eqref{Aymptotic Edge 1}, \eqref{Aymptotic Edge 2}, and \eqref{Hat q even edge limit}, 
as $N\to\infty$, we have 
\begin{equation}
    \frac{\widehat{q}_{2N}^{\,(\rm pre)}(z)}{\widehat{q}_{2N}^{\,(\rm pre)}(u)}
    \sim
    \frac{\mathcal{A}(\zeta,\chi)}{(\xi-\chi)^3\mathcal{B}(\chi)},
\end{equation}
uniformly for $\zeta$ in a compact subset of $\C$ and $\chi$ in a compact subset of $\R$, 
where $\mathcal{A}(\zeta,\chi)$ and $\mathcal{B}(\chi)$ are given by \eqref{cA_Edge} and \eqref{cB_Edge}.
Hence, combining all of the above with \eqref{u over pre kernel} and \eqref{wtilde u kappa over}, we have 
\begin{equation}
 \lim_{N\to\infty}
 \widetilde{\bfkappa}_{N-1}^{(\rm over)}(z,w)
 =
 \widetilde{\kappa}_{\mathrm{e}}^{(\rm over)}(\zeta,\eta)
\end{equation}
uniformly for $\zeta,\eta$ in compact subsets of $\C$, and for $\chi$ in a compact subset of $\R$, where  
\begin{align}
\begin{split}
\label{overlap edge kernel widetilde}
\widetilde{\kappa}_{\mathrm{e}}^{(\rm over)}(\zeta,\eta)
&=
\mathcal{K}(\zeta,\eta|\chi)
-
\frac{\mathcal{A}(\eta,\chi)\mathcal{C}(\zeta,\chi)}{\mathcal{B}(\chi)}
+
\frac{
\mathcal{A}(\zeta,\chi)\mathcal{C}(\eta,\chi)}{\mathcal{B}(\chi)}.
\end{split}    
\end{align}
By dividing \eqref{overlap edge kernel widetilde} by $(\zeta-\chi)^4(\eta-\chi)^4$ and combining it with \eqref{finite N overlap k point correlation function}, we obtain \eqref{Lim overlap k point correlation}. This completes the proof.

\end{proof}

\appendix

\section{Planar orthogonal polynomials and a pre-overlap weight function of different variance} \label{Appendix_Planar OP}

In this appendix we explore the effect of changing the variance upon the construction of planar orthogonal polynomials. Due the nature of the point insertion this does not lead to a simple rescaling of the argument of the polynomials, but rather affects their coefficients in a non-trivial way.

Let $\mu$ be a positive Borel measure on $\C$ with an infinite number of points in its support on a domain $D$. 
We define the 
inner product with respect to $\mu$ on $D$ by  
\begin{equation}
\label{OPinner}
\langle f, g\rangle :=\int_{D} f(z)\overline{g(z)} \, d\mu(z).
\end{equation}
A family of polynomials $(p_k)_{ k \in \mathbb{Z} }$ is called planar orthogonal polynomials (OP) associated with $d\mu(z)=w(z)dA(z)$ if 
    \begin{equation}
    \label{OP norm}
    \langle p_k, p_\ell \rangle = h_k \,\delta_{k,\ell},
    \end{equation}
    where $h_k$ is the squared norm and $\delta$ is the Kronecker delta.

Below we construct the planar OPs $p_k^{ \rm (pre) }(z)$ associated with the parameter dependent weight function on $D=\mathbb{C}$,
\begin{equation}
\label{pre overlap weight c-version}
    \omega_{\sigma}^{(\rm pre)}(z)
    :=
    \left(1+|z-a|^2\right)
    \exp\left(-\sigma|z|^2\right),\qquad
    z,a\in\C,\quad \sigma>0,
\end{equation}
generalising the weight \eqref{PreOverlapWeight} with $\sigma=2$. Notice that when changing variables $z\to z'=\sqrt{\sigma}\,z$ and denoting $a'=\sqrt{\sigma}\,a$, the weight changes to 
\begin{equation}
  \omega_{\sigma}^{(\rm pre)}(z') \, dA(z')
  =\frac{1}{\sigma^2}
    \left(\sigma+|z'-a'|^2\right)\exp\left(-|z'|^2\right) \,dA(z'),
\end{equation}
changing the prefactor and thus the content of the point charge insertion.


\begin{prop}[\textbf{Planar orthogonal polynomial associated with weight of changing variance}] \label{Prop_OP pre}
For any given integer $p\in\mathbb{N}\cup\{0\}$ and $\sigma>0$, let 
\begin{equation}
    \label{Fpc}
    F_{p,\sigma}(x):=
\sum_{j=0}^p
\frac{\Gamma(p-j+\sigma+1)}{\Gamma(j+1)\Gamma(p-j+1)}x^j.
\end{equation}
The monic planar OPs $p_{k,\sigma}^{ \rm (pre) }$ associated with \eqref{pre overlap weight c-version} are then obtained as 
\begin{equation}
\label{OPs c-version}
p_{k,\sigma}^{ \rm (pre) }(z) :=\sum_{j=0}^{k}{a}^{k-j}\frac{F_{j,\sigma}(\sigma|a|^2)}{F_{k,\sigma}(\sigma|a|^2)}z^j,
\end{equation} 
with squared norms
\begin{equation}
\label{OPcnorm}
h_{k,\sigma}^{\rm (pre)}= \frac{(k+1)!}{\sigma^{k+2}}\frac{F_{k+1,\sigma}(\sigma|a|^2)}{F_{k,\sigma}(\sigma|a|^2)}.
\end{equation}
Furthermore, they satisfy the following non-standard three term recurrence relation 
for $k\geq 1$:
\begin{equation}  \label{non-stand 3 term c}
z\, p_{k,\sigma}^{ \rm (pre) }(z) =p_{k+1,\sigma}^{ \rm (pre) }(z)+b_k\, p_{k,\sigma}^{ \rm (pre) }(z) +z \, c_k \,p_{k-1,\sigma}^{ \rm (pre) }(z),
\end{equation} 
where 
\begin{equation}
b_k=-a\, \frac{ F_{k,\sigma}(\sigma|a|^2) }{ F_{k+1,\sigma}(\sigma|a|^2) } ,
\qquad  c_k=a\, \frac{ F_{k-1,\sigma}(\sigma|a|^2 ) }{ F_{k,\sigma}(\sigma|a|^2 ) }. 
\end{equation}
\end{prop}

We also mention that the corresponding planar OPs $p_k^{ \rm (over) }$ with respect to the perturbed weight (now depending on $\sigma$) could also be constructed,  using Proposition~\ref{Prop_OP pre} and \cite[Eq.(3.6)]{AV03}.
\begin{proof}
We  denote the moment matrix associated with \eqref{pre overlap weight c-version} by 
\begin{equation}
\label{moment-def}
\mathfrak{m}_{i,j}:=\int_{\C}z^i\zbar^j\omega_{\sigma}^{(\rm pre)}(z) \, dA(z).
\end{equation}
It is well known that the kernel of OP can be obtained form the inverse moment matrix. In turn, if we derive an LDU decomposition of $\mathfrak{m}_{i,j}$, where $L$ and $U$ are lower and upper triangular matrices, respectively, the latter two can be easily inverted, and the planar OPs follow in turn as well. This strategy was employed in \cite{ATTZ20} which we shall follow closely.

In a first step we split off a common factor to define a reduced moment matrix $\mu_{i,j}$
\begin{equation}
\label{mudef}
\mathfrak{m}_{i,j}
=i!\frac{1}{\sigma^{i+2}}\mu_{i,j},\qquad
\mu_{i,j}:=\Bigl(a\overline{a}\sigma+i+1+\sigma\Bigr)\delta_{i,j}-\overline{a}(i+1)\delta_{i+1,j}-a\sigma\delta_{i,j+1}.
\end{equation}
If we decompose the matrix $\mu=(\mu_{i,j})$ as a LDU decomposition $\mu=LDU$, where $D=(D_{p,q})$, $L=(L_{p,q})$, and $U=(U_{p,q})$ with 
\begin{equation}
D_{p,q}:=d_p\delta_{p,q},\quad L_{p,q}:=\delta_{p,q}+l_{p}\delta_{p,q+1},
\quad U_{p,q}=\delta_{p,q}+u_q\delta_{q,p+1}\  \mbox{for}\ p,q\in\mathbb{Z}_{\geq0}
\label{LDUdef}
\end{equation}
Then, multiplying out $LDU$ and comparing with \eqref{mudef}, we can read off the following
\begin{equation}
d_pu_{p+1}=-\overline{a}(p+1),\qquad
d_pl_{p+1}=-a\sigma,\qquad
d_{p-1}l_pu_p\mathbf{1}_{p\geq1}+d_p
=a\overline{a}\sigma+p+1+\sigma. 
\end{equation}
It follows from the above relationships that 
\begin{equation}
\label{dkrec}
\frac{xp}{d_{p-1}}+d_p=x+p+1+\sigma,
\qquad
d_0=x+1+\sigma, 
\end{equation}
after defining 
\begin{equation}
\label{x-def}
x:=\sigma|a|^2.
\end{equation}
Let us denote $d_p=r_{p+1}/r_{p}$, setting $r_0=1$ and thus $r_1=d_0$. Then, we obtain the following linearised recurrence for the $r_k$ from \eqref{dkrec}
\begin{equation}
r_{p+1}+xpr_{p-1}=(x+p+1+\sigma)r_p.
\label{rkrec}
\end{equation}
It is not difficult to show by induction that the unique solution for the above recurrence equation is given by 
\begin{equation}
r_p
=\frac{\Gamma(p+1)}{\Gamma(\sigma+1)}
F_{p,\sigma}(x),
\label{rksol}
\end{equation}
with $F_{p,\sigma}(x)$ defined in \eqref{Fpc}.
We can now include the additional factor in \eqref{mudef} to convert the LDU decomposition of $\mu$ into a LDU decomposition of the matrix $\mathfrak{m}=(\mathfrak{m}_{i,j})$. Updating the notations, and defining $F_{-1,\sigma}(x):=0$, we find that $\mathfrak{m}=LDU$, where for $p,q=0,1,2,\dots$,
\begin{equation}
    L_{p,q}=\delta_{p,q}-a\frac{F_{p-1,\sigma}(x)}{F_{p,\sigma}(x)}\delta_{p,q+1},
    \qquad
    D_{p,p}=\frac{(p+1)!}{\sigma^{p+2}}\frac{F_{p+1,\sigma}(x)}{F_{p,\sigma}(x)},
    \qquad
    U_{p.q}=\delta_{p,q}-\overline{a}\frac{F_{p-1,\sigma}(x)}{F_{p,\sigma}(x)}\delta_{q,p+1}.
    \label{LDUsol}
\end{equation}
The inverse matrices of $L$ and $U$ are given by 
\[
(L^{-1})_{p,q}
=
\begin{cases}
    0, & q>p,
    \smallskip 
    \\
    1, & q=p,
    \smallskip 
    \\
    a^{p-q}\frac{F_{q,\sigma}(x)}{F_{p,\sigma}(x)}, & q<p,
\end{cases}
\qquad
(U^{-1})_{p,q}
=
\begin{cases}
    \overline{a}^{q-p}\frac{F_{p,\sigma}(x)}{F_{q,\sigma}(x)}, & q>p,
    \smallskip 
    \\
    1, & q=p,
    \smallskip 
    \\
    0, & q<p.
\end{cases}
\]
Since the planar orthogonal polynomials $\{p_{k,\sigma}^{(\rm pre)}\}_k(z)$ associated with \eqref{pre overlap weight c-version} are given by
\begin{equation}
p_{k,\sigma}^{(\rm pre)}(z)=\sum_{j=0}^k({L}^{-1}_{k,j})z^j, 
\end{equation}
we obtain \eqref{OPs c-version}. (We mention that compared to \cite{ATTZ20}, our inner product \eqref{OPinner} has complex conjugation on the right factor, thus there is no conjugation of $L^{-1}$ here.)
The squared norms \eqref{OPcnorm} follow from 
\begin{equation}
 \langle p_k, p_\ell \rangle = D_{k,k} \,\delta_{k,\ell},
\end{equation}
with \eqref{LDUsol}.
Furthermore, it is straightforward to check that the polynomials \eqref{OPs c-version} satisfy the non-standard three-term recurrence \eqref{non-stand 3 term c}. It follows when inserting the definition \eqref{OPs c-version} and comparing coefficients, using the recurrence \eqref{rkrec} in terms of the $r_p$ from \eqref{rksol}.
\end{proof}

\begin{ex}[$\sigma=1,2$]
For $\sigma=1$, 
\begin{equation}
\label{Fpc=1}
F_{p,1}(x)=\sum_{j=0}^p
(p+1-j)\frac{x^j}{j!}
= (p+1)e_p(x)-xe_{p-1}(x)=
f_p(x),
\end{equation}
where $f_p(x)$ is given by \eqref{fsum}. 
Here, $e_p(x)$ is the truncated exponential given by \eqref{exponential_sum}, and we set $e_{-1}(x)=e_{-2}(x)\equiv 0$.
Then, the OPs $p_{k,1}^{(\rm pre)}$ associated with \eqref{pre overlap weight c-version} and the squared norms \eqref{OPcnorm} can be written as 
\begin{equation}
 p_{k,1}^{(\rm pre)}(z) = \sum_{j=0}^{k}a^{k-j}\frac{f_j(|a|^2)}{f_k(|a|^2)}z^j,\qquad
 h_{k,1}^{(\rm pre)}:=(k+1)!\frac{f_{k+1}(|a|^2)}{f_k(|a|^2)}.
\end{equation}
This was derived in \cite[Subsection 3.4]{ATTZ20}. 

For $\sigma=2$ we have 
\begin{equation} \label{def of F_p}
F_{p,2}(x)=
(p+2)(p+1)e_p(x)-2(p+1)xe_{p-1}(x)+x^2e_{p-2}(x).
\end{equation}
Then, \eqref{pre overlap weight c-version} becomes \eqref{PreOverlapWeight}, and hence, the OPs $p_{k,2}^{ \rm (pre) }$ associated with \eqref{PreOverlapWeight} and the squared norms \eqref{OPcnorm} can be written as 
\begin{equation}
\label{OPs}
p_{k,2}^{ \rm (pre) }(z) =\sum_{j=0}^{k}a^{k-j}\frac{F_{j,2}(2|a|^2)}{F_{k,2}(2|a|^2)}z^j,\qquad
h_{k,2}^{(\rm pre)} = \frac{(k+1)!}{2^{k+2}}\frac{F_{k+1,2}(2|a|^2)}{F_{k,2}(2|a|^2)}.
\end{equation} 
Both examples satisfy the non-standard three term recurrence relation \eqref{non-stand 3 term c}.
\end{ex}

\subsection*{Acknowledgments}
We thank Mark Crumpton, Yan Fyodorov, Hiroyuki Ochiai, Tomoyuki Shirai, Roger Tribe, and Oleg Zaboronski for their interest and helpful discussions.
Gernot Akemann was partly was supported the Deutsche Forschungsgemeinschaft (DFG) grant SFB 1283/2 2021-317210226 and a Leverhulme Visiting Professorship, grant VP1-2023-007. He is indebted to the School of Mathematics, University of Bristol, where part of this research was completed.
Sung-Soo Byun was supported by the POSCO TJ Park Foundation (POSCO Science Fellowship), by the New Faculty Startup Fund at Seoul National University and by the National Research Foundation of Korea funded by the Korea government (NRF-2016K2A9A2A13003815, RS-2023-00301976, RS-2025-00516909).
Kohei Noda was partially supported by WISE program (JSPS) at Kyushu University, JSPS KAKENHI Grant Numbers 18H01124, 23H01077, JP22H05105, and JP23K25774, and the Deutsche Forschungsgemeinschaft (DFG) grant SFB 1283/2 2021-317210226.

\bibliographystyle{abbrv}

\begin{thebibliography}{100}

 
\bibitem{Ak05} G. Akemann, \emph{The complex Laguerre symplectic ensemble of non-Hermitian matrices}, Nuclear Phys. B \textbf{730} (2005), 253--299.
 
 \bibitem{AB07} G. Akemann and F. Basile, \emph{Massive partition functions and complex eigenvalue correlations in matrix models with symplectic symmetry}, Nuclear Phys. B \textbf{766} (2007), 150--177.


\bibitem{ABK22} G. Akemann, S.-S. Byun and N.-G. Kang, \emph{Scaling limits of planar symplectic ensembles}, SIGMA Symmetry Integrability Geom. Methods Appl. \textbf{18} (2022), Paper No. 007, 40pp.


\bibitem{AEP22} G. Akemann, M. Ebke and I. Parra, \emph{Skew-orthogonal polynomials in the complex plane and their Bergman-like kernels}, Comm. Math. Phys. \textbf{389} (2022), 621--659.
 
  
\bibitem{ATTZ20} G. Akemann, R. Tribe, A. Tsareas and O. Zaboronski, 
\emph{On the determinantal structure of conditional overlaps for the complex Ginibre ensemble},  Random Matrices Theory Appl. \textbf{9} (2020), 2050015.

\bibitem{ATTZ20a} G. Akemann, R. Tribe, A. Tsareas and O. Zaboronski, 
\emph{Determinantal structure and bulk universality of conditional overlaps in the complex Ginibre ensemble}, Acta Phys. Pol. \textbf{51} (2020), 1611--1626.


 \bibitem{AFK20} G. Akemann, Y. F\"{o}rster and M. Kieburg, \emph{Universal eigenvector correlations in quaternionic Ginibre ensembles}, J. Phys. A. \textbf{53}, (2020), 145201.

 \bibitem{AV03} G. Akemann and G. Vernizzi, \emph{Characteristic polynomials of complex random matrix models}, Nuclear Phys. B \textbf{660} (2003), 532--556. 

\bibitem{AKS23} Y. Ameur, N.-G. Kang and S.-M. Seo, \emph{The random normal matrix model: Insertion of a point charge}, Potential Anal. \textbf{58} (2023), 331--372.

\bibitem{AHM11} Y. Ameur, H. Hedenmalm and N. Makarov, \emph{Fluctuations of eigenvalues of random normal matrices}, Duke Math. J. \textbf{159} (2011), 31--81.

\bibitem{BNST17} S. Belinschi, M. A. Nowak, R. Speicher and W. Tarnowski, \emph{Squared eigenvalue condition numbers and eigenvector correlations from the single ring theorem}, J. Phys. A {\bf 50} (2017), 105204.


\bibitem{BC12} F. Benaych-Georges and F. Chapon, \emph{Random right eigenvalues of Gaussian quaternionic matrices}, Random Matrices Theory Appl. \textbf{1} (2012), 1150009.
 
\bibitem{BZ18} F. Benaych-Georges and O. Zeitouni, \emph{Eigenvectors of non normal random matrices}, Electron. Commun. Probab. \textbf{23} (2018), 1--12.

 
\bibitem{BD21} P. Bourgade and G. Dubach, \emph{The distribution of overlaps between eigenvectors of Ginibre matrices}, Probab. Theory Relat. Fields \textbf{177} (2020), 397--464.


 \bibitem{BY17} P. Bourgade and H.-T. Yau, \emph{The eigenvector moment flow and local quantum unique ergodicity}, Comm. Math. Phys. {\bf350} (2017), 231--278.


\bibitem{BSV17} Z. Burda, B.J. Spisak and P. Vivo, \textit{Eigenvector statistics of the product of Ginibre matrices}, Phys. Rev. E {\bf 95} (2017), 022134.

\bibitem{BC23} S.-S. Byun and C. Charlier, \emph{On the almost-circular symplectic induced Ginibre ensemble}, Stud. Appl. Math. \textbf{150} (2023), 184--217.

\bibitem{BE23} S.-S. Byun and M. Ebke, \emph{Universal scaling limits of the symplectic elliptic Ginibre ensembles}, Random Matrices Theory Appl. \textbf{12} (2023), 2250047.

\bibitem{BES23} S.-S. Byun, M. Ebke and S.-M. Seo, \emph{Wronskian structures of planar symplectic ensembles}, Nonlinearity \textbf{36} (2023), 809--844.

\bibitem{BF23a} S.-S. Byun and P. J. Forrester, \emph{Spherical induced ensembles with symplectic symmetry}, SIGMA Symmetry Integrability Geom. Methods Appl. \textbf{19} (2023), 033, 28pp.



\bibitem{BF24} S.-S.~Byun and P. J.~Forrester, \emph{Progress on the study of the Ginibre ensembles}, KIAS Springer Ser. Math. \textbf{3} Springer, 2025, 221pp.

 
\bibitem{BLY21} S.-S. Byun, S.-Y. Lee and M. Yang,  \emph{Lemniscate ensembles with spectral singularity}, arXiv:2107.07221.
 

\bibitem{BN24} S.-S. Byun and K. Noda, \emph{Scaling limits of complex and symplectic non-Hermitian Wishart ensembles}, J. Approx. Theory \textbf{308} (2025), 106148.
 


\bibitem{CM98} J. T. Chalker and B. Mehlig, \emph{Eigenvector statistics in non-Hermitian random matrix ensembles}, Phys. Rev. Lett. \textbf{81} (1998), 3367--3370.

\bibitem{CM00} J. T. Chalker and B. Mehlig, \emph{Statistical properties of eigenvectors in non-Hermitian Gaussian random matrix ensembles}, J. Math. Phys. \textbf{41} (2000), 3233--3256.

\bibitem{Ch22} C. Charlier, \emph{Asymptotics of determinants with a rotation-invariant weight and discontinuities along circles}, Adv. Math. \textbf{408} (2022), 108600.

\bibitem{CS22} C. Cipolloni and D. Schr\"{o}der, \emph{On the condition number of the shifted real Ginibre ensemble}, SIAM J. Matrix Anal. Appl. \textbf{43} (2022), 1469--1487.

\bibitem{CEHS23} C. Cipolloni, L. Erd\"{o}s, J. Henheik and D. Schr\"{o}der,
\emph{Optimal lower bound on eigenvector overlaps for non-Hermitian random matrices}, J. Funct. Anal. \textbf{287} (2024), 110495.


\bibitem{CR22} N. Crawford and R. Rosenthal, \emph{Eigenvector correlations in the complex Ginibre ensemble}, Ann. Appl. Probab. \textbf{32} (2022), 2706--2754.


\bibitem{CFW24} M. J. Crumpton, Y. V. Fyodorov and T. R. W\"{u}rfel, \emph{Mean eigenvector self-overlap in the real and complex elliptic Ginibre ensembles at strong and weak non-Hermiticity}, Ann. Henri Poincar\'{e} (2025). https://doi.org/10.1007/s00023-024-01530-2.

\bibitem{CW24} M. Crumpton, T. R. W\"{u}rfel, \emph{Spectral density of complex eigenvalues and associated mean eigenvector self-overlaps at the edge of elliptic Ginibre ensembles}, arXiv:2405.02103.


\bibitem{DLMS19} D. S. Dean, P. Le Doussal, S. N. Majumdar and G. Schehr, \emph{Noninteracting fermions in a trap and random matrix theory}, J. Phys. A \textbf{52} (2019), 144006.
 
\bibitem{D21v1} G. Dubach, {\it On eigenvector statistics in the spherical and truncated unitary ensembles}, Electron. J. Probab. \textbf{26} (2021), Paper No. 124, 29 pp.

\bibitem{D21v2} G. Dubach, \emph{Symmetries of the quaternionic Ginibre ensemble}, Random Matrices Theory Appl. \textbf{10} (2021), 2150013.

\bibitem{D23} G. Dubach, \emph{Explicit formulas concerning eigenvectors of weakly non-unitary matrices}, Electron. Commun. Probab. \textbf{28} (2023), Paper No. 6, 11 pp.

  

\bibitem{EC23} L. Erd\"{o}s and H. C. Ji, \emph{Wegner estimate and upper bound on the eigenvalue condition number of non-Hermitian random matrices}, Comm. Pure Appl. Math. (Online), https://doi.org/10.1002/cpa.22201, arXiv:2301.04981.


\bibitem{EKY21} S. Esaki, M. Katori and S. Yabuoku, \emph{Eigenvalues, eigenvector-overlaps, and regularized Fuglede-Kadison determinant of the non-Hermitian matrix-valued Brownian motion}, arXiv:2306.00300. 
 
\bibitem{Forrester10} P. J. Forrester, \emph{Log-gases and random matrices},
 Princeton University Press, Princeton, NJ, 2010.

\bibitem{Fo16} P. J. Forrester, \emph{Analogies between random matrix ensembles and the one-component plasma in two-dimensions}, Nuclear Phys. B \textbf{904} (2016), 253--281. 

 
\bibitem{FM02} Y. V. Fyodorov and B. Mehlig, \emph{Statistics of resonances and nonorthogonal eigenfunctions in a model for single-channel chaotic scattering}, Phys. Rev. E. {\bf 66} (2002), 045202.


\bibitem{FS12} Y. V. Fyodorov and D. V. Savin, \emph{Statistics of resonance width shifts as a signature of eigenfunction non-orthogonality}, Phys. Rev. Lett. {\bf 108} (2012), 184101.

\bibitem{F18} Y. V. Fyodorov,  \emph{On statistics of bi-orthogonal eigenvectors in real and complex Ginibre ensembles: combining partial Schur decomposition with supersymmetry}, Comm. Math. Phys. {\bf 363} (2018), 579--603.
 

\bibitem{FT21} Y. V. Fyodorov and W. Tarnowski, \emph{Condition numbers for real eigenvalues in the real elliptic Gaussian ensemble}, Ann. Henri Poincar\'{e} {\bf 22}, (2021), 309--330.


\bibitem{FS03} Y. V. Fyodorov and H.-J. Sommers, \emph{Random matrices close to hermitian or unitary: overview of methods and results}, J. Phys. A {\bf 36} (2003), 3303--3347.
 

\bibitem{HW21} H. Hedenmalm and A. Wennman, \emph{Planar orthogonal polynomials and boundary universality in the random normal matrix model}, Acta Math. \textbf{227} (2021), 309--406.

\bibitem{GW18} J. Grela and P. Warchol, \emph{Full Dysonian dynamics of the complex Ginibre ensemble}, J. Phys. A {\bf 51} (2018), 425203.

\bibitem{Gros14} J.-B. Gros, U. Kuhl, O. Legrand, F. Mortessagne, E. Richalot and D.V. Savin, \emph{Experimental width shift distribution: A test of nonorthogonality for local and global perturbations}, Phys. Rev. Lett. {\bf 113} (2014), 224101.
 
\bibitem{Has00} M. B. Hastings, \emph{Fermionic mapping for eigenvalue correlation functions of weakly non-Hermitian symplectic ensemble}, Nucl. Phys. B {\bf 572} (2000), 535--546. 
 
\bibitem{IM95} M. E. H. Ismail and D. R. Masson, \emph{Generalized orthogonality and continued fractions}, J. Approx. Theory \textbf{83} (1995), 1--40.
 
\bibitem{JNNPZ99} R. A. Janik, W. N\"{o}renberg, M. A. Nowak, G. Papp and I. Zahed, \emph{ Correlations of eigenvectors for non Hermitian random matrix models}, Phys. Rev. E. {\bf 60} (1999), 2699--2705.

\bibitem{Kan02} E. Kanzieper, \emph{Eigenvalue correlations in non-Hermitean symplectic random matrices}, J. Phys. A \textbf{35} (2002), 6631--6644.


\bibitem{KS23} J. S. Kim and D. Stanton, \emph{Combinatorics of orthogonal polynomials of type $R_I$}, Ramanujan J. \textbf{61} (2023), 329--390.

 
\bibitem{KY13} A. Knowles and J. Yin,  \emph{Eigenvector distribution of Wigner matrices}, Probab. Theory Related Fields, {\bf155} (2013), 543--582.

\bibitem{KE99} A. V. Kolesnikov and K. B. Efetov, \emph{Distribution of complex eigenvalues for symplectic ensembles
of non-Hermitian matrices}, Waves in Random Media {\bf 9} (1999), 71--82 
 



\bibitem{LY23} S.-Y. Lee and M. Yang, \emph{Strong asymptotics of planar orthogonal polynomials: Gaussian weight perturbed by finite number of point charges}, Comm. Pure Appl. Math. \textbf{76} (2023), 2888--2956.

\bibitem{LOR12} T. A. Loring,  {\it Factorization of matrices of quaternions}, Expo. Math. {\bf 30} (2012), 250.

  
\bibitem{MS01} B. Mehlig and M. Santer,  \emph{Universal eigenvector statistics in a quantum scattering ensemble}, Phys. Rev. E. (2001) {\bf 63}, 020105.


\bibitem{Mehta}
M. L. Mehta, \emph{Random Matrices}, Academic Press, New York, second edition, 1991.



\bibitem{Noda23a} K. Noda, \emph{Determinantal structure of the overlaps of the induced Ginibre unitary ensemble}, arXiv:2310.15362.

\bibitem{Noda23b} K. Noda, \emph{Determinantal structure of the overlaps for induced spherical unitary ensemble}, to appear in Random Matrices Theory Appl.

\bibitem{N18} M. A. Nowak and W. Tarnowski, \emph{Probing non-orthogonality of eigenvectors in non-Hermitian matrix models: diagrammatic approach}, J. High Energ. Phys. \textbf{2018} (2018), 152.


\bibitem{NIST} F. W. Olver, D. W. Lozier, R. F. Boisvert, and C. W. Clark (Editors). NIST Handbook of Mathematical Functions. Cambridge University Press, Cambridge, 2010.

\bibitem{RVW16} S. O'Rourke, V. Vu and K. Wang, \emph{Eigenvectors of random matrices: a survey},  J. Combin. Theory. Ser. A {\bf 144} (2016), 361--442.

 
\bibitem{SS23} A. Serebryakov and N. Simm, \emph{Schur function expansion in non-Hermitian ensembles and averages of characteristic polynomials}, Ann. Henri Poincar\'e (Online), https://doi.org/10.1007/s00023-024-01483-6, arXiv:2310.20686.


\bibitem{Se24} S. Serfaty, \emph{Lectures on Coulomb and Riesz Gases}, 2024. 


\bibitem{ST03} T. Shirai and Y. Takahashi, \emph{Random point fields associated with certain Fredholm determinants I: fermion, Poisson and boson point processes}, J. Funct. Anal. \textbf{205} (2003), 414--463.

\bibitem{TV12} T. Tao and V. Vu, \emph{Random matrices: universal properties of eigenvectors}, Random Matrices Theory Appl. \textbf{1} (2012), 1150001.

\bibitem{TAR24} W. Tarnowski, \emph{Condition numbers for real eigenvalues of real elliptic ensemble: weak non-normality at the edge}, J. Phys. A \textbf{57} (2024), 255204.

 
\bibitem{WS15} M. Walters and S. Starr, \emph{A note on mixed matrix moments for the complex Ginibre ensemble}, J. Math. Phys. {\bf 56} (2015), 013301.

 
\bibitem{WTF23} T. R. W\"{u}rfel, M. J. Crumpton and Y. V. Fyodorov, \emph{Mean left-right eigenvector self-overlap in the real Ginibre ensemble}, Random Matrices Theory Appl. \textbf{13}, (2024), 2450017.

\bibitem{Y20} S. Yabuoku, \emph{Eigenvalue processes of elliptic Ginibre ensemble and their overlaps}, Int. J. Math. Ind. \textbf{12} (2020), 2050003, 17 pp.

\end{thebibliography}

\end{document}